\documentclass[a4paper,UKenglish]{lipics}

\newif\ifdraft
\draftfalse

\newif\iftechreport
\techreporttrue
\usepackage{amssymb}
\usepackage{microtype}
\usepackage{enumerate}
\usepackage{hyperref}
\usepackage{amsmath,amssymb}
\usepackage{xspace}
\usepackage{url}
\usepackage[defblank]{paralist}
\usepackage{multirow}
\usepackage{array}
\usepackage{soul}
\usepackage{booktabs}
\usepackage{comment}

\usepackage{listings}
\usepackage{tikz}
\usetikzlibrary{arrows, fit}
\usetikzlibrary{shapes}

\usepackage{placeins}

\newcolumntype{P}[1]{>{\centering\arraybackslash}p{#1}}

\input{macros-def}
\definecolor{tboxcolor}{HTML}{116699}
\definecolor{aboxcolor}{HTML}{cc6633}
\definecolor{schemecolor}{HTML}{cccccc}

\definecolor{mygray}{HTML}{cccccc}
\definecolor{mydarkgray}{HTML}{555555}
\definecolor{myorange}{HTML}{cc6633}
\definecolor{myred}{HTML}{C00D2A}
\definecolor{mygreen}{HTML}{009B5C}
\definecolor{myblue}{HTML}{006FB3}

\colorlet{conccolor}{myblue}
\colorlet{rolecolor}{black}
\colorlet{indcolor}{myred}
\colorlet{homcolor}{mydarkgray}
\colorlet{autocolor}{mygreen!50!mygray}

\colorlet{tboxlcolor}{mygreen}%myblue}
\colorlet{tboxrcolor}{myblue}
\colorlet{aboxcolor}{mydarkgray}%red}

\pgfdeclarelayer{background}
\pgfdeclarelayer{foreground}
\pgfsetlayers{background,main,foreground}

\pgfmathparse{atan2(0,1)}
\ifdim\pgfmathresult pt=0pt % atan2(y, x)
  \tikzset{declare function={atanXY(\x,\y)=atan2(\y,\x);atanYX(\y,\x)=atan2(\y,\x);}}
\else                       % atan2(x, y)
  \tikzset{declare function={atanXY(\x,\y)=atan2(\x,\y);atanYX(\y,\x)=atan2(\x,\y);}}
\fi

\tikzset{ %
  background rectangle/.style= {fill=white, rounded corners}, %
  table/.style={draw, mydarkgray, rounded corners, inner sep=0},
  player1/.style={rounded corners=6, outer sep=0.05cm, inner ysep=0.05cm,
    minimum width=1cm, minimum height=0.5cm, draw=black, fill=white}, %
  player1a/.style={rounded corners=3, inner ysep=0.04cm, draw=black,
    fill=white, text=myred}, %
  player0/.style={rounded corners=0, outer sep=0.05cm, inner ysep=0.05cm,
    minimum width=1cm, minimum height=0.5cm, draw=mydarkgray, fill=mygray}, %
  move0/.style={myred},%
  move1/.style={myblue},%
  deadend/.style={double, double distance=1pt},%
  pebble/.style={fill, circle, black, text=white, inner sep=0.5pt, outer
    sep=0.1pt, minimum size=0.2cm},%
  individual/.style={draw, circle, outer sep=0.05cm, inner sep=1.2pt, thick,
  }, %
  constant/.style={fill}, %
  autree/.style={rectangle, outer sep=0.05cm, inner sep=0.5pt}, %
  aunode/.style={individual, autocolor}, %
  aurun/.style={-stealth', thick, mygreen}, %
  austate/.style={-latex, color=autocolor, text=mydarkgray}, %
  region/.style={rounded corners, color=mygray, fill=mygray!30, text=black}, %
  entity/.style={rounded corners=2, inner sep=0.04cm, outer sep=0.02cm},%
  abstractentity/.style={draw, fill=white, inner sep=0.1cm, outer sep=0.02cm,
    drop shadow={opacity=.4,shadow xshift=0.04, shadow yshift=-0.04}},%
  abox/.style={cylinder, shape border rotate=90, aspect=0.4, thick, minimum
    width=1.8cm, minimum height=1.5cm, inner ysep=0.1cm, outer sep=0.05cm, %
    draw=aboxcolor, fill=aboxcolor!35},%
  aboxl/.style={abox, draw=aboxcolor!60, fill=aboxcolor!20},%
  aboxr/.style={abox},%
  abox3d/.style={abox, draw=aboxcolor!80, %
    left color=aboxcolor!70, right color=aboxcolor!80, middle color=white,%
    append after command={ let%
      \p{cyl@center} = ($(\tikzlastnode.before top)!0.5!(\tikzlastnode.after top)$), %
      \p{cyl@x} = ($(\tikzlastnode.before top)-(\p{cyl@center})$), %
      \p{cyl@y} = ($(\tikzlastnode.top) -(\p{cyl@center})$) in %
      (\p{cyl@center}) edge[draw=none, fill=aboxcolor!40, to path={ ellipse%
        [x radius=veclen(\p{cyl@x})-1\pgflinewidth-0.05cm, %
        y radius=veclen(\p{cyl@y})-1\pgflinewidth-0.05cm, rotate=atanXY(\p{cyl@x})]}]
      () }%
  },%
  abox3dl/.style={abox, draw=aboxcolor!40, %
    left color=aboxcolor!30, right color=aboxcolor!40, middle color=white,%
    append after command={ let%
      \p{cyl@center} = ($(\tikzlastnode.before top)!0.5!(\tikzlastnode.after top)$), %
      \p{cyl@x} = ($(\tikzlastnode.before top)-(\p{cyl@center})$), %
      \p{cyl@y} = ($(\tikzlastnode.top) -(\p{cyl@center})$) in %
      (\p{cyl@center}) edge[draw=none, fill=aboxcolor!20, to path={ ellipse%
        [x radius=veclen(\p{cyl@x})-1\pgflinewidth-0.05cm, %
        y radius=veclen(\p{cyl@y})-1\pgflinewidth-0.05cm, rotate=atanXY(\p{cyl@x})]}]
      () }%
  },%
  abox3dr/.style={abox3d}, %
  tbox/.style={thick, outer sep=0.05cm, inner sep=0.2cm, drop
    shadow={opacity=0.35}, draw=tboxcolor, fill=tboxcolor!50},%
  tboxl/.style={tbox, draw=tboxlcolor, fill=tboxlcolor!40}, %
  tboxr/.style={tbox, draw=tboxrcolor, fill=tboxrcolor!40}, %
  tbox3d/.style={tbox, draw=tboxcolor!85, inner color=tboxcolor!45, outer
    color=tboxcolor!65},%
  tbox3dl/.style={tbox, draw=tboxlcolor!60, inner color=tboxlcolor!20, outer
    color=tboxlcolor!40},%
  tbox3dr/.style={tbox, draw=tboxrcolor!85, inner color=tboxrcolor!45, outer
    color=tboxrcolor!65},%
  kbbox/.style={draw=mygray, rounded corners, very thick, outer sep=0cm, inner
    sep=0.2cm},%
  kbbox3d/.style={kbbox, pattern=north east lines, draw=mydarkgray, pattern
    color=mydarkgray!50, rounded corners=10, thick},%
  trans/.style={-stealth', semithick},%
  homomor/.style={-stealth, dashed, thick, color=homcolor},%
  role/.style={-latex, semithick},%
  isa/.style={-open triangle 60, semithick},%
  disj/.style={open triangle 60-open triangle 60, ultra thin,
    postaction={decorate, decoration={markings, mark=at position 0.45 with
        {\draw[-,solid] (-2pt,-2pt) -- (2pt, 2pt); } }}},%
  soledge/.style={dotted, -latex, ultra thick, color=mydarkgray, text=black},%
  path/.style={decorate, decoration={zigzag, segment length=20pt,
      amplitude=2pt},->, ultra thin},%
  mapping/.style={decorate, decoration={coil, aspect=0, segment length=20pt,
      amplitude=1pt},-latex, semithick, mydarkgray},%
  disjmap/.style={decorate, decoration={coil, aspect=0, segment length=20pt,
      amplitude=1pt},-latex, semithick, mydarkgray, postaction={decorate,
      decoration={markings, mark=at position 0.55 with {\draw[-,solid]
          (-2pt,-2pt) -- (2pt, 2pt); } }}},%
  crossout/.style={decorate, decoration={markings, mark=at position 0.55 with
      {\draw[-,solid,thick,myred] (-2pt,-2pt) -- (2pt, 2pt); %
        \draw[-,solid,thick,myred] (-2pt,2pt) -- (2pt, -2pt);} }},%
  obdamapping/.style={-latex, line width=0.1cm, myred},%
  isometricYXZ/.style={x={(1cm,0cm)}, y={(-1.299cm,-0.75cm)}, z={(0cm,1cm)}},%
  inside/.code=\preto\tikz@auto@anchor{\pgf@x-\pgf@x\pgf@y-\pgf@y},%
}

\tikzset{
  sshadow/.style={opacity=.25, shadow xshift=0.05, shadow yshift=-0.06},
}

\def\factor{4}
\def\xradius{2}
\def\yradius{2/\factor}
\def\height{1.05cm}
\def\xandy{2 and 2/\factor}

\tikzset{
  pics/.cd, %
  disc/.style ={
    code = {
      \path [fill=red] (-\xradius,0) -- (-\xradius,-\height) arc
      (180:360:\xandy) -- (\xradius,0) arc (0:180:\xandy);%
      \path [draw=white, ultra thick] (-\xradius,0) -- (-\xradius,-\height) arc
      (180:360:\xandy) -- (\xradius,0)%
      (0,0) ellipse [x radius=\xradius, y radius =\yradius];%
    } },%
}

\tikzset{ %
  taxonomy/.style={anchor=west, grow via three points={one child at (0.2,-0.5)
      and two children at (0.2,-0.5) and (0.2,-1)}, %
    edge from parent path={(\tikzparentnode.south) |- (\tikzchildnode.west)},
    draw=mydarkgray}, %
  webtaxonomy/.style={anchor=west, grow via three points={one child at
      (0.1,-0.5) and two children at (0.1,-0.5) and (0.1,-1)}, %
    edge from parent path={[]}}, %
  highlight/.style={draw=myred, rounded corners, text=black},%
  webhighlight/.style={draw=mygray, fill=mygray!50, rounded corners, text=black}%
}

\def \tikzdots[#1]{
  \begin{scope}[shift={#1}, text=black]
    \node at (0,0.15) {$\cdot$}; \node {$\cdot$}; \node at (0,-0.15) {$\cdot$};
  \end{scope}
}

\def\schemel[#1,#2,#3,#4,#5,#6]#7{ %
  \node[draw, diamond, shape aspect=#3, rotate=#2, minimum size=#1, %
  bottom color=schemecolor!80!black, top color=schemecolor!30,
  color=schemecolor!80, %
  drop shadow={sshadow,color=gray!50}] (#5) at #6
  {\textcolor{schemecolor!90}{$\Sigma_1$}}; %
  \node at #6 {#7};%

  \node[anchor=north] at (#5.south) {#4}; 
}

\def\schemer[#1,#2,#3,#4,#5,#6]#7{ %
  \node[draw, diamond, shape aspect=#3, rotate=#2, minimum size=#1, %
  bottom color=schemecolor!50!black, top color=schemecolor!50,
  color=schemecolor!87!black, %
  drop shadow={sshadow,color=gray!50}] (#5) at #6
  {\textcolor{schemecolor!80!black}{$\Sigma_2$}}; %
  \node at #6 {#7}; %

  \node[anchor=north] at (#5.south) {#4}; 
}

\newcommand{\mongodb}{MongoDB\xspace}

\iftechreport
\title{Expressivity and Complexity of \mongodb Queries (Extended Version)}
\else
\title{Expressivity and Complexity of \mongodb \makebox[0cm][l]{Queries}}
\fi

\author{E.\ Botoeva}
\author{D.\ Calvanese}
\author{B.\ Cogrel}
\author{G.\ Xiao}
\affil{Free University of Bozen-Bolzano \texttt{\textit{lastname}@inf.unibz.it}}

\iftechreport
\else

\authorrunning{E. Botoeva, D. Calvanese, B. Cogrel and G. Xiao} %mandatory. First: Use abbreviated first/middle names. Second (only in severe cases): Use first author plus 'et. al.'

\Copyright{Elena Botoeva, Diego Calvanese, Benjamin Cogrel and Guohui Xiao}%mandatory, please use full first names. LIPIcs license is "CC-BY";  http://creativecommons.org/licenses/by/3.0/

\fi

\begin{document}

\maketitle

\begin{abstract}
  A significant number of novel database architectures and data models have
  been proposed during the last decade. While some of these new systems have
  gained in popularity, they lack a proper formalization, and a precise
  understanding of the expressivity and the computational properties of the
  associated query languages.  In this paper, we aim at filling this gap, and
  we do so by considering \mongodb, a widely adopted document database system
  managing complex (tree structured) values represented in a JSON-based data
  model, equipped with a powerful query mechanism.  We provide a formalization
  of the \mongodb data model, and of a core fragment, called \mquery, of the
  \mongodb query language.  We study the expressivity of \mquery, showing its
  equivalence with nested relational algebra.  We further investigate the
  computational complexity of significant fragments of it, obtaining several
  (tight) bounds in combined complexity, which range from \LOGSPACE to
  alternating exponential-time with a polynomial number of alternations.  As a
  consequence, we obtain also a characterization of the combined complexity of
  nested relational algebra query evaluation.
\end{abstract}

\section{Introduction}

As was envisioned by Stonebraker and Cetintemel \cite{StCe05}, during the last
ten years a diversity of new database (DB) architectures and data models has
emerged, driven by the goal of better addressing the widely varying demands of
modern data-intensive applications.  Notably, many of these new systems do not
rely on the relational model but instead adopt a semi-structured data format,
and alternative query mechanisms, which combine an increased flexibility in
handling data, with a higher efficiency (at least for some types of common
operations).  These systems are generally categorized under the terms
\emph{NoSQL} (for ``not only SQL'') \cite{Catt11,Leav10}.

A large portion of the so-called \emph{non-relational} systems (e.g., \mongodb,
CouchDB, and
% Hadoop
DocumentDB) organize data in collections of semi-structured, tree-shaped
documents in the JavaScript Object Notation (JSON) format, which is commonly
viewed as a lightweight alternative to XML.
Such documents can be seen as complex values
\cite{GrVi91,AbBe95,VaPa95,DaVo97}, in particular due to the presence of nested
arrays.
Consider, e.g., the document in Figure~\ref{fig:mongodb-document}, containing
standard personal information about Kristen Nygaard (such as name and
birth-date), and information about the awards he received, the latter being
stored inside an array.

\begin{figure}[t]
  \centering
\begin{lstlisting}
{ "_id": 4,
  "awards": [
    { "award": "Rosing Prize", "year": 1999, "by": "Norwegian Data Association" },
    { "award": "Turing Award", "year": 2001, "by": "ACM" },
    { "award": "IEEE John von Neumann Medal", "year": 2001, "by": "IEEE" } ],
  "birth": "1926-08-27",
  "contribs": [ "OOP", "Simula" ],
  "death": "2002-08-10",
  "name": { "first": "Kristen", "last": "Nygaard" } }
\end{lstlisting}
  \caption{A sample \mongodb document in the \texttt{bios} collection}
  \label{fig:mongodb-document}
\end{figure}

It is not surprising that among the non-relational languages that have been
proposed for querying JSON collections (see, e.g.,
\cite{beyer11jaql,olston08pig,thusoo09hive} and the \emph{\mongodb aggregation
 framework}\footnote{\scriptsize\url{https://docs.mongodb.com/core/aggregation-pipeline/}}),
languages with rich capabilities have many similarities with well-known query
languages for complex values, such as monad algebra (MA) \cite{BNTW95,Koch06},
nested relational algebra (NRA) \cite{thomas86,Vand01} and Core XQuery
\cite{Koch06}.  For instance, Jaql \cite{beyer11jaql}, one of the most
prominent query languages targeting map-reduce frameworks \cite{dean08},
supports higher-order functions, which have their roots in MA, and the group
and unwind operators of \mongodb
% the \mongodb aggregation framework
are similar to the nest and unnest operators of NRA.  While some of these
languages have been widely used in large-scale applications, there have been
only few attempts at capturing their formal semantics, e.g., through a calculus
for Jaql \cite{BCNS13}.  Only very recently abstract frameworks have been
proposed, with the aim of understanding the formal and computational properties
of query languages over JSON documents \cite{hidders17jlogic,bourhis17json}.

In this paper, we consider the case of \mongodb, a widespread JSON-based
document database, and conduct the first major investigation into the formal
foundations and computational properties of its data model and query language.
\mongodb provides rich querying capabilities by means of the \emph{aggregation
 framework}, which is modeled on the flexible notion of data processing
pipeline.  In this framework, a query is composed of multiple stages, where
each stage defines a transformation, using a \mongodb-specific operator,
applied to the set of documents produced by the previous stage.  The \mongodb
model is at the basis of systems provided by different vendors, such as the
DocumentDB system on Microsoft
Azure\footnote{\scriptsize\url{https://docs.microsoft.com/en-us/azure/documentdb/documentdb-protocol-mongodb}}.

Our first contribution is a formalization of the \mongodb data model and of a
fragment of the aggregation framework query language, which we call
\emph{\mquery}.  We deliberately abstract away some low-level features, which
appear to be motivated by implementation aspects, rather than by the objective
of designing an elegant language for nested structures.  On the other hand, our
objective still is to capture as precisely as possible the actual behavior of
\mongodb, rather than developing top-down a clean theoretical framework that is
distant from the actual system.  We see this as essential for our work to be of
practical relevance, and possibly help ``cleaning up'' some of the debatable
choices made for \mongodb.  \mquery includes the \emph{\textbf{m}atch},
\emph{\textbf{u}nwind}, \emph{\textbf{p}roject}, \emph{\textbf{g}roup}, and
\emph{\textbf{l}ookup} operators, roughly corresponding to the NRA operators
select, unnest, project, nest, and left join, respectively.
% each of which roughly corresponds to an
% operator of NRA: match to select, project to project, lookup to
% left join, and as mentioned above, group to nest, and unwind to unnest.
As a useful side-effect of our formalization effort, we point out different
``features'' exhibited by \mongodb's query language that are somewhat
counter-intuitive, and that might need to be reconsidered by the \mongodb
developers for future versions of the system.
In our investigation, we consider various fragments of \mquery, which we denote
by \MQ{$\alpha$}, where $\alpha$ consists of the initials of the stages that
can be used in the fragment.

Our second contribution is a characterization of the expressive power of
\mquery obtained by comparing it with NRA.  We define the relational view of
JSON documents, and devise translations in both directions between \mquery and
NRA, showing that the two languages are equivalent in expressive power.  We
also consider the \mupg fragment, where we rule out the lookup operator, which
allows for joining a given document collection with external ones.  Actually,
we establish that already \mupg is equivalent to NRA over a single relation,
and hence is capable of expressing arbitrary joins (within one collection),
contrary to what is believed in the community of \mongodb practitioners and
users.  Interestingly, all our translations are compact (i.e., polynomial),
hence they allow us also to carry over complexity results between \mquery and
NRA.

Finally, we carry out an investigation of the computational complexity of
\mupgl and its fragments.  In particular, we establish that what we consider
the minimal fragment, which allows only for match, is \LOGSPACE-complete (in
combined complexity).  Projection and grouping allow one to create
exponentially large objects, but by representing intermediate results compactly
as DAGs, one can still evaluate \mpgl queries in \PTIME.  The use of unwind
alone causes loss of tractability in combined complexity, specifically it leads
to \NP-completeness, but remains \LOGSPACE-complete in query complexity.
Adding also project or lookup leads again to intractability even in query
complexity, although \mupl stays \NP-complete in combined complexity.  In the
presence of unwind, grouping provides another source of complexity, since it
allows one to create doubly-exponentially large objects; indeed we show
\PSPACE-hardness of \mug.  Finally, we establish that the full language and
also the \mupg fragment are complete for exponential time with a polynomial
number of alternations (in combined complexity).
As mentioned, our polynomial translations between \mquery and NRA, allow us to
carry over the complexity results also to NRA (and its fragments).  In
particular, we establish a tight \TAexppoly result for the combined complexity
of Boolean query evaluation in NRA, for which the lower bound was known, but
the best upper bound was \EXPSPACE \cite{Koch06}.

% We further discuss some notable features of \mongodb that we have encountered
% in our investigation as a result of our attempts to understand the semantics of
% its query language. Since such features are to some degree counterintuitive,
% and show even some inconsistent behaviors of the current version of \mongodb
% (\currentversion), we consider it important to make the \mongodb community
% aware of them, so that users can properly make use of the query language.

% The rest of the paper is structured as follows.  In
% Section~\ref{sec:preliminaries} we introduce NRA. In
% Section~\ref{sec:documents} we provide our formalization of \mongodb
% documents, and in Section~\ref{sec:queries} of the \mongodb query language.
% In Section~\ref{sec:expressivity} we study the expressivity of such
% language by providing translations to and from NRA, and in
% Section~\ref{sec:complexity} we study its computational complexity.  We
% conclude the paper in Section~\ref{sec:conclusions}.  Selected proofs are
% given in the Appendix.

%%% Local Variables:
%%% mode: latex
%%% TeX-master: "icdt18"
%%% fill-column: 79
%%% End:

\section{Preliminaries}
\label{sec:preliminaries}

We recap the basics of nested relational algebra (NRA)
\cite{JaeschkeS82,Vand01}, mainly to fix the notation. % of the operators.

Let $\A$ be a countably infinite set of attribute names and relation schema
names.  A \emph{relation schema} has the form $R(S)$, where $R\in\A$ is a
relation schema name and $S$ is a finite set of attributes, each of which is an
atomic attribute (i.e., an attribute name in~$\A$) or a schema of a
sub-relation.  A relation schema can also be obtained through an NRA operation
(see below).  We use the function $\att$ to retrieve the attributes from a
relation schema name, i.e., $\att(R)=S$.
%
% \nb{we don't use this definition}A relation schema $R(S)$ is called
% \emph{flat} if $S$ consists of atomic attributes only.
%
Let $\Delta$ be the domain of all atomic attributes in $\A$.  An
\emph{instance} $\R$ of a relation schema $R(S)$ is a finite set of tuples over
$R(S)$.  A \emph{tuple $t$ over $R(S)$} is a finite set
$\{a_1{:}v_1,\dots,a_n{:}v_n\}$ such that if $a_i$ is an atomic attribute, then
$v_i\in\Delta$, and if $a_i$ is a relation schema, then $v_i$ is an instance of
$a_i$.
In the following, when convenient, we refer to relation schemas by their name
only.

A \emph{filter} $\psi$ over a set $A\subseteq\A$ is a Boolean formula
constructed from atoms of the form $(a=v)$ or $(a=a')$, where
$\{a,a'\}\subseteq A$, and $v$ is an atomic value or a relation.
Let $R$ and $R'$ be relation schemas.  We use the following
% relational algebra
operators:
\begin{inparaenum}[\itshape (1)]
\item \emph{set union} $R\cup R'$ and \emph{set difference}
$R\setminus R'$,
%$S\setminus S'$,
  for $\att(R)=\att(R')$;
\item \emph{cross-product} $R\times R'$, resulting in a relation schema with
  attributes
  $\{\text{rel1}.a \mid a\in\att(R)\} \cup
  \{\text{rel2}.a \mid a\in\att(R')\}$;
% \item \emph{join} $R\Join_{\psi} R'$, where $\psi$ is a filter over
%   $\att(R)\cup\att(R')$; if $\psi$ is missing, then it denotes \emph{natural
%     join};
% \item \emph{left outer join} $R \leftouterjoin_{\psi} R'$;
\item \emph{selection} $\sigma_\psi(R)$, where $\psi$ is a filter over
  $\att(R)$;
\item \emph{projection} $\pi_{P}(R)$, for $P\subseteq\att(R)$;
\item \emph{extended projection} $\pi_{P}(R)$, where $P$ may also contain
  elements of the form $b/e(a_1,\ldots,a_n)$, for an expression $e$ computable
  in \ACz in data complexity, $b$ a fresh attribute name, and
  $\{a_1,\ldots,a_n\}\subseteq\att(R)$;
%\end{inparaenum}
%
%We also use the two operators of NRA:
%\begin{inparaenum}[\itshape (1)] \setcounter{enumi}{5}
\item \emph{nest} $\nest_{\{a_1, \ldots, a_n\}\rightarrow b}(R)$, resulting in
  a schema with attributes
  $(\att(R)\setminus\{a_1,\ldots,a_n\}) \cup \{b(a_1,\ldots,a_n)\}$; and
\item \emph{unnest} $\unnest_a(R)$, resulting in a schema with attributes
  $(\att(R)\setminus\{a\}) \cup \att(a)$.
\end{inparaenum}
For more details on~(5)--(7), we refer to Appendix~\ref{sec:semantics-nra}.
%
% We note that $\unnest_a(R)$ will not preserve tuples $t$ if
% $\pi_a(\{t\})=\{\}$.
%
% We also assume that in NRA the project operator $\pi$ can access
% sub-relations at all levels as in~\cite{Colby89}.
%
% A (NRA) database $\D$ is a set of instances of relational schemas.
Given an NRA query $Q$ and a (relational) database $\D$, the result of
evaluating $Q$ over $\D$ is denoted by $\ansra(Q,\D)$.

% We also adopt the \textit{recursive nested relational algebra} which allows
% for accessing (checking ) subrelations at all levels in a nested relation,
% and is equivalent in expressive power to the non-recursive nested relational
% algebra~\cite{Colby89}.
% %
% In recursive NRA we assume that we allow for equality between a relational
% value and a nested attribute in filter expressions.  We also allow for
% comparing relations.
%
% For instance, the following are valid queries: $\sigma_{awards.year=2001}$,
% $\sigma_{2001 \in awards.year}$, and $\pi_{y/awards.year}$ .

%\nb{TODO: function composition?}

\section{\mongodb Documents}
\label{sec:documents}

In this section, we propose a formalization of the syntax and the semantics of
\mongodb documents.  In our formalization, we make two simplifying assumptions
with respect to the way such documents are treated by the \mongodb system:
\begin{inparaenum}[\itshape (i)]
\item we abstract away document order, i.e., we view documents as expressed in
  JSON, as opposed to
  BSON\footnote{\url{https://docs.mongodb.org/manual/reference/bson-types/}},
  and
\item we consider set-semantics as opposed to bag-semantics.
\end{inparaenum}

\begin{figure}
  \centering
  \begin{tabular}{r@{~~}c@{~~}l}
    \meta{Value} &\DEF& \meta{Literal} \CHOICE~ \meta{Object} \CHOICE~ \meta{Array} \\[-0.5mm]
    \meta{Object} &\DEF& \term{\lobject} \meta{List<Key} \term{:} \meta{Value>} \term{\robject}\\[-0.5mm]
    \meta{Array} &\DEF& \term{[} \meta{List<Value>} \term{]}
  \end{tabular}~
  \begin{tabular}{r@{~~}c@{~~}l}
    \meta{List<T>} &\DEF& \meta{$\varepsilon$} \CHOICE~ \meta{List$^+$<T>}\\
    \meta{List$^+$<T>} &\DEF& \meta{T} \CHOICE~ \meta{T} \term{,} \meta{List$^+$<T>}
  \end{tabular}
  \caption{Syntax of JSON objects.  We use double curly brackets to distinguish
   objects from sets}
  \label{fig:syntax-bson}
\end{figure}

A \mongodb database stores collections of documents, where
% each document is an object consisting of \emph{key-value pairs},
% \footnote{Here we adopt the same terminology as \mongodb, where the term `key'
% is used with the meaning of `attribute' in relational databases, not to be
% confused with the traditional notion in `key constraints'.},
% and a value can itself be a nested object.
a collection corresponds to a table in a (nested) relational database, and a
document to a row in a table.
We define the syntax of \mongodb documents.
% in the JSON format.
%
\emph{Literals} are atomic values, such as strings, numbers, and Booleans.
A \emph{JSON object} is a finite set of key-value pairs, where a \emph{key} is a
string and a \emph{value} can be a literal, an object, or an array of values,
constructed inductively according to the grammar in Figure~\ref{fig:syntax-bson}
(where terminals are written in black, and non-terminals in blue).
We require that the set of key-value pairs constituting a JSON object does not
contain the same key twice.
A \emph{(\mongodb) document} is a JSON object not nested within any other
object, with a special key `\id', which is used to identify the document.
Figure~\ref{fig:mongodb-document} shows a \mongodb document in which, apart from
\id, the keys are \valuefont{birth}, \valuefont{name}, \valuefont{awards}, etc.
% Notice that the value of \valuefont{name} is an object consisting of two
% key-value pairs, and the value of \valuefont{awards} is an array of objects,
% each describing an award.
%
Given a collection name~$C$, a \emph{(\mongodb) collection for $C$} is a finite
set $F_C$ of documents, such that each document is identified by its value of
\id, i.e., each value of \id is unique in $F_C$. Given a set $\mathbb{C}$ of
collection names, a \emph{\mongodb database instance $D$ (over $\mathbb{C}$)} is
a set of collections, one for each name $C \in \mathbb{C}$.  We write $D.C$ to
denote the collection for name $C$.

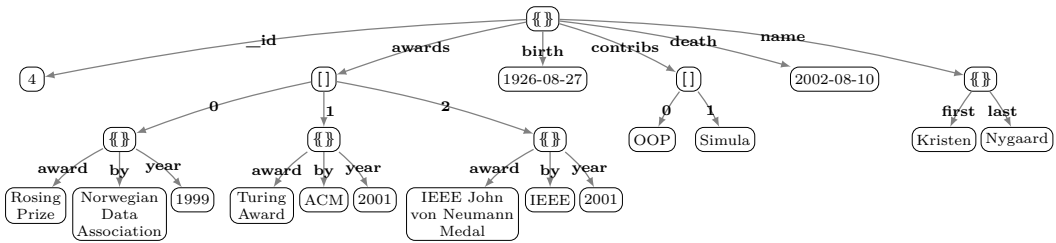
\begin{figure*}[b]
  \centering
  \scalebox{0.8}{\begin{tikzpicture}[xscale=1.2]\scriptsize
  \foreach \al/\x/\y/\lab/\wh in {%
    root/3/1/{\object{}}/above, %
    _id/-4/-0/{4}/below,%
    awards/0/0/{\bf [\,]}/above,%
    birth/3/0/{1926-08-27}/below,%
    contribs/5/0/{\bf [\,]}/above,%
    death/7/0/{2002-08-10}/below,%
    name/9/0/{\object{}}/above,%
    first/8.5/-1/{Kristen}/below,%
    last/9.5/-1/{Nygaard}/below,%
    award0/-2.8/-1/{\object{}}/above,%
    award1/0/-1/{\object{}}/above left,%
    award2/3.1/-1/{\object{}}/above,%
    contrib0/4.5/-1/{OOP}/below,%
    contrib1/5.5/-1/{Simula}/below,%
    a0a/-3.95/-2/{\begin{minipage}{0.85cm}\centering Rosing Prize
      \end{minipage}}/below,%
    a0b/-2.8/-2/{\begin{minipage}{1.4cm}\centering Norwegian Data
        Association\end{minipage}}/below,%
    a0y/-1.8/-2/{1999}/below,%
    a1a/-0.85/-2/{\begin{minipage}{0.9cm}\centering Turing
        Award\end{minipage}}/below,%
    a1b/0/-2/{ACM}/below,%
    a1y/0.7/-2/{2001}/below,%
    a2a/1.9/-2/{\begin{minipage}{1.7cm}\centering IEEE John\\von
        Neumann\\Medal\end{minipage}}/below,%
    a2b/3.1/-2/{IEEE}/below,%
    a2y/3.8/-2/{2001}/below%
  }{ \node[draw,rectangle,rounded corners, minimum height=0.4cm,
    minimum width=0.4cm, inner sep=2,
    anchor=north]%, fill, gray]%, label={[inner sep=1]\wh:{\lab}}]
    (\al) at (\x,\y) {\lab};}

  \foreach \from/\to/\lab/\b in {%
    root/_id/\_id/5, root/name/name/-7,%
    root/awards/awards/5, root/birth/birth/0,
    root/contribs/contribs/-7, root/death/~~death/-5,
    name/first/first/5, name/last/last/-5, awards/award0/0/5,
    awards/award1/~~1/0, awards/award2/2/-5, %
    contribs/contrib0/0/5, contribs/contrib1/1/-5, %
    award0/a0a/award~~~~/10, award0/a0b/by/0, award0/a0y/year/-10,
    award1/a1a/award~~~~/10, award1/a1b/by/0, award1/a1y/~~year/-10,
    award2/a2a/award~~~~/10, award2/a2b/by/0, award2/a2y/~~year/-10%
  }{\draw[role, gray] (\from) to[bend right=\b]
    node[text=black,pos=0.55] {\bf\lab} (\to);}
\end{tikzpicture}}
  \caption{The tree representation of the \mongodb document in
    Figure~\ref{fig:mongodb-document}}
\label{fig:mongodb-document-tree}
\end{figure*}

We formalize \mongodb documents as finite \emph{unordered, unranked,
 node-labeled, and edge-labeled trees}.  We assume three disjoint sets of
labels: the sets $K$ of \emph{keys} and $I$ of \emph{indexes} (non-negative
integers), used as edge-labels, and the set $V$ of \emph{literals}, containing
the special elements \nullvalue, \truevalue, and \falsevalue, and used as node
labels.
A \emph{tree} is a tuple $(N,E,\lnode,\ledge)$, where $N$ is a set of nodes, $E$ is a
successor relation, $\lnode: N \to V \cup \big\{\objectlabel, \arraylabel\big\}$ is
a node labeling function, and $\ledge: E \to K\cup I$ is an edge labeling function,
such that
\begin{inparaenum}[\it (i)]
\item $(N,E)$ forms a tree,
\item a node labeled by a literal must be a leaf,
\item all outgoing edges of a node labeled by \objectlabel must be labeled by
  keys, and
\item all outgoing edges of a node labeled by \arraylabel must be labeled by
  distinct indexes.
  % starting from 0, and respecting the sibling order $\prec$.
\end{inparaenum}
Given a tree $t$ and a node $x$, the \emph{type} of $x$ in $t$, denoted
$\type(x,t)$, is \tliteral if $\lnode(x) \in V$, \tobject if
$\lnode(x) =\objectlabel$, and \tarray if $\lnode(x) =\arraylabel$.
The root of $t$ is denoted by $\treeroot(t)$.
% If $\treeroot(t)$ has an outgoing edge labeled with $\id$, we call the tree
% $t$ a \emph{document}.
A \emph{forest} is a set of trees.

We define inductively the \emph{value represented by} a node $x$ in a
tree $t$, denoted $\mathsf{value}(x,t)$:
\begin{inparaenum}[\it (i)]
\item $\mathsf{value}(x,t) = \lnode(x)$, if $x$ is a leaf in $t$;
\item let $x_1,\dots,x_m$, be all children of $x$ with $\ledge(x,x_i)=k_i$.
  % the corresponding edges labeled by $k_1,\dots,k_m$.
  Then $\mathsf{value}(x,t)$ is
  $\object{k_1{:}\mathsf{value}(x_1,t),\ldots,k_m{:}\mathsf{value}(x_m,t)}$ if
  $\type(x,t)=\tobject$, and $[\mathsf{value}(x_1,t),\ldots,
  \mathsf{value}(x_m,t)]$, if $\type(x,t)=\tarray$.
\end{inparaenum}
The \emph{JSON value represented by~$t$} is then
$\mathsf{value}(\treeroot(t),t)$.
Conversely, the \emph{tree corresponding to a value $u$}, denoted $\tree(u)$, is
defined as $(N,E,\lnode,\ledge)$, where $N$ is the set of $x_v$ such that $v$ is
an object, array, or literal value appearing in $u$, and for $x_v \in N$:
\begin{inparaenum}[\itshape (i)]
\item if $v$ is a literal, then $\lnode(x_{v}) = v$ and $x_v$ is a leaf;
\item if $v = \object{k_1{:}v_1, \dots, k_m{:}v_m}$ for $m \geq 0$, then
  $\lnode(x_v) = \objectlabel$, \mbox{$x_v$ has $m$ children $x_{v_1},\dots,x_{v_m}$}
  with $\ledge(x_v,x_{v_i}) = k_i$; % and $x_{v_1} \prec \cdots \prec x_{v_m}$;
\item if $v = [v_1,\dots,v_m]$ for $m \geq 0$, then $\lnode(x_v) =
  \arraylabel$, $x_v$ has $m$ children $x_{v_1}, \dots, x_{v_m}$ with
  $\ledge(x_v,x_{v_i}) = i-1$. % and $x_{v_1} \prec \cdots \prec x_{v_m}$.
\end{inparaenum}
%
% Note that a literal $v$ can be seen as a tree consisting of a single node
% labeled $v$.
%
% The tree corresponding to a JSON document $d$ is defined as $\tree(d)$.
The tree corresponding to the document in Figure~\ref{fig:mongodb-document} is
depicted in Figure~\ref{fig:mongodb-document-tree}.

%%% Local Variables:
%%% mode: latex
%%% TeX-master: "icdt18"
%%% fill-column: 80
%%% End:

\section{\mongodb Queries}
\label{sec:queries}

% The most popular query mechanism of \mongodb are so-called \emph{find}
% queries, which are essentially based on pattern matching.  However,
\mongodb is equipped with an expressive query mechanism provided by the
\emph{aggregation framework}, and a first contribution of this paper is to
provide a formalization of its core aspects.
% In doing so,
We deliberately abstract away (in the algebra and semantics) some
low-level features%
% of the language
\footnote{We provide in Appendix~\ref{sec:mquery-vs-mongodb} a more detailed
 discussion on the difference in syntax and semantics between our formal
 language and the one implemented by \mongodb.},
% which appear to be motivated by implementation aspects, rather than by the
% objective of designing a formally clean language for nested structures.  On
% the other hand, our objective still is to capture as precisely as possible the
% actual behavior of \mongodb, rather than developing top-down a clean
% theoretical framework that is distant from the actual system.  We see this as
% essential for our work to be of practical relevance, and possibly help
% ``cleaning up'' some of the debatable choices made for \mongodb.\nb{Possibly
% move (also) to intro.}
and we use set (as opposed to bag) semantics.  We call the resulting
language~\emph{\mquery}.
%  and we consider also different fragments of it.

% In our effort we also abstract away low-level features of the language that
% either are not relevant for understanding its expressive power and its
% computational properties, or that are procedural and make the language
% Turing-complete (such as the possibility to embed general Java-script code).

%\subsection{Syntax of \mquery}

An \emph{\mquery} is a sequence of stages $s$, also called a \emph{pipeline},
applied to a collection name~$C$, where each stage
% roughly corresponds to a relational algebra operation, and
transforms a forest into another forest.  Here we are not concerned with
syntactic aspects of \mquery (which are described in detail in
Figure~\ref{fig:mongodb-query-syntax} in the Appendix), and instead propose for
it an algebra, shown in Figure~\ref{fig:algebra-syntax}.
% which we then exploit for the further development in the paper.
% \footnote{For simplicity, the only comparison operator that we kept in the
% algebra is equality.  Adding also order comparison would not affect any of the
% results on expressivity and complexity presented later.}
%
\begin{figure}
\centering
{ $\begin{array}{@{}r@{~}r@{~}l}
    % \mathbf{op} &\DEF& {=} \mid {\neq} \mid {<} \mid {\leq} \mid {>} \mid {\geq}
    %\\
    \varphi &\DEF& p = v
    \mid \exists p
    \mid \neg \varphi
    \mid \varphi \lor\varphi
    \mid \varphi \land \varphi
    \\
    d &\DEF& v \mid p \mid [d,\dots,d] \mid \beta \mid \cond{\beta}{d}{d}
    \\
    \beta &\DEF& \truevalue
    \mid \falsevalue
    \mid p = p
    \mid p = v
    \mid{}\\
    && \hspace{1.2cm}
    \exists p
    \mid \neg \beta
    \mid \beta \lor \beta
    \mid \beta \land \beta
    %\mid \cond{\beta}{\beta}{\beta}
    \\
    % \beta &\DEF& p = p
    % \mid p = v
    % \mid v = v
    % \mid \exists p
    % \mid d \lor d
    % \mid d \land d
    % \mid \neg d
    % \\
    %\kappa &\DEF& \cond{\beta}{d}{d}\\
   \end{array}\qquad
   \begin{array}{r@{~}r@{~}l}
    P &\DEF& p \,\mid\, p/d \,\mid\, p,P \,\mid\, p/d,P
    \\
    G,A &\DEF& p/p' \,\mid\, p/p',G
    \\
    s &\DEF& \match{\varphi}
    \mid \unwind{p}
    \mid \unwind[+]{p}
    \mid \project{P}
    \mid \project[\noid]{P}
    \mid \group{G}{A}
    \mid \lookup{p_1=C.p_2}{p} \\
    \multicolumn{3}{l}{\mquery ~\DEF~ C \pipeline s \pipeline \cdots \pipeline
     s}
  \end{array}$
}
\caption{Algebra for \mquery.  Here, $p$ denotes a path, $v$ a value, and $C$ a
 collection name }
\label{fig:algebra-syntax}
\end{figure}
% (for readability, we use single curly brackets in queries)
%
In an \mquery, \emph{paths}, which are (possibly empty) concatenations of keys,
are used to access actual values in a tree, similarly to how attributes are used
in relational algebra.
We use $\varepsilon$ to denote the empty path.
For two paths $p$ and $p'$, we say that $p'$ is a \emph{(strict) prefix} of $p$,
if $p=p'.p''$, for some (non-empty) path $p''$.
% Also, $p'$ is a \emph{prefix} of $p$ if $p'$ is either a strict prefix of $p$
% or equal to $p$.
%
\mquery allows for five types of \emph{stages} (below, we use the tree $t$ in
Figure~\ref{fig:mongodb-document-tree}):
\begin{compactitem}
\item \emph{\textbf{m}atch} $\match{\varphi}$, selecting trees according to
  criterion $\varphi$, which is a Boolean combination of atomic conditions
  expressing the equality of a path $p$ to a value $v$, or the existence of a
  path $p$. E.g., for $\varphi_1 = (\text{\id}\text{\small =4})$,
  $\varphi_2 = (\text{\small awards.award="Turing Award"})$ and
  $\varphi_3 = (${\small awards $=$ \object{award: "Rosing
    Prize", year: 2001, by: "ACM"}}$)$, $\mu_{\varphi_1}$ and $\mu_{\varphi_2}$
  select~$t$, but $\mu_{\varphi_3}$ does not.

\item \emph{\textbf{u}nwind} $\unwind{p}$ and $\unwind[+]{p}$, which flatten an
  array reached through a path $p$ in the input tree, and output a tree for each
  element of the array; $\unwind[+]{p}$ preserves a tree even when the array
  does not exist or is empty.  For instance, $\unwind{\text{awards}}$ produces
  three trees from $t$, which coincide on all key-value pairs, except for the
  \valuefont{awards} key, whose values are nested objects such as, e.g.,
  {\small\object{award: "Turing Award", year: 2001, by: "ACM"}}.

\item \emph{\textbf{p}roject} $\project{P}$ and $\project[\noid]{P}$, which
  modify trees by projecting away paths, renaming paths, or introducing new
  paths; $\project[\noid]{P}$ projects away \id, while $\project{P}$ keeps it by
  default.  Here $P$ is a sequence of elements of the form $p$ or $q/d$, where
  $p$ is a path to be kept, $q$ is a new path whose value is defined by $d$, and
  among all such paths $p$ and $q$, there is no pair $p$, $p'$ where $p$ is a
  prefix of $p'$.  A \emph{value definition} $d$ can provide for $q$ a constant
  $v$, the value reached through a path $p$ (i.e., \emph{renaming} path $p$ to
  $q$), a new array defined through its values, the value of a Boolean
  expression $\beta$, or a value computed through a conditional expression
  $\cond{\beta}{d_1}{d_2}$.  Note that, in a Boolean expression~$\beta$, one can
  also compare the values of two paths, while in a match criterion $\varphi$ one
  can only compare the value of a path to a constant value.
  % Also note that each value definition can be evaluated to a Boolean value.
  For example,
  $\project{\text{bool}/(\text{birth=death}),\, \text{cond}/\cond{(\exists
    \text{awards})}{\text{contribs}}{\text{\id}},\, \text{newArray}/[0,1]}$
  applied to $t$ produces {\small $\tree${\small(\lobject \id:\,4,
    bool:\,\falsevalue, cond:\,["OOP", "Simula"], newArray:\,[0,1]\robject)}}.

\item \emph{\textbf{g}roup} $\group{G}{A}$, which groups trees according to a
  grouping condition $G$ and collects values of interest according to an
  aggregation condition $A$. Both $G$ and $A$ are (possibly empty) sequences of
  elements of the form $p/p'$, where $p'$ is a path in the input trees, and
  $p$ a path in the output trees.
  % which must also be a key in $A$.
  In these sequences, if $p$ coincides with $p'$, then we simply write $p$
  instead of $p/p$.  Each group in the output will have an \id whose value is
  given by the values of $p'$ in $G$ for that group.  Consider, e.g., the trees
  $\tree${\small(\lobject\id:\,1, a:\,"a1"\robject)} and $\tree${\small(\lobject
    \id:\,2, a:\,"a2", d:\,"d2"\robject)}.  Then $\group{}{\text{ids}/\id}$
  groups them in \mbox{$\tree${\small(\lobject\id: \nullvalue, ids: [1,
      2]\robject)},} while $\group{\text{d}}{\text{a}}$ produces two groups
  $\tree${\small(\lobject\id:\,\lobject \robject, a:\,["a1"]\robject)} and
  $\tree${\small(\lobject\id: \lobject d:\,"d2"\robject, a:\,["a2"]\robject)}.

\item \emph{\textbf{l}ookup} $\lookup{p_1=C.p_2}{p}$, which joins input trees
  with trees in an external collection $C$, using a local path $p_1$ and a path
  $p_2$ in $C$ to express the join condition, and stores the matching trees in
  an array under a path $p$.  E.g., let $C$ consist of
  $\tree${\small(\lobject\id:\,1, a:\,3\robject)} and
  \mbox{$\tree${\small(\lobject \id:\,2, a:\,4\robject)}}.  Then
  $\lookup{\id=C.\text{a}}{\text{docs}}$ evaluated over $t$ adds to it the
  key-value pair {\small docs:\,[\lobject \id:\,2, a:\,4\robject]}.
\end{compactitem}
We consider also various fragments of \mquery, and we denote each fragment by
\MQ{$\alpha$}, where $\alpha$ consists of the initials of the stages that can be
used in queries in the fragment.  Hence, \mupgl denotes \mquery itself, and,
e.g., \mupg denotes \mupgl without lookup.
% the fragment of \mupgl that does not use lookup, and \mup the fragment that
% additionally does not use group.
%
% We observe that \mongodb find queries correspond to a simple case of \mp,
% where a match stage is followed by a project stage in which each projection
% element is of the form \meta{Path:}~\truevalue and \id:~\falsevalue.

% \subsection{Semantics of \mquery}
% \label{sec:semanticsAlgebra}

\smallskip

To define the semantics of \mquery, we first show how to interpret paths over
trees.

\begin{definition}
  Given a tree $t=(N,E,\lnode,\ledge)$, we interpret a (possibly empty) path
  $p$, and its concatenation $p.i_1...i_m$ with indexes $i_1,\dots,i_m$, as sets
  of nodes (below $k$ is a key):\\[1mm]
  \mbox{\quad}$\begin{array}[b]{r@{~}c@{~}l}
    \eval{\varepsilon} &=& \{\treeroot(t)\}
    \\
    \eval{p.k} &=& \{y \in N \mid
    \text{there are $i_1,\dots,i_m$, $m\geq 0$, }
    \text{and }x\in\eval{p.i_1...i_m}\\
    &&\hfill \text{ s.t. } (x,y)\in E \text{ and } \ledge(x,y)=k\}
    \\
    \eval{p.i_1...i_m} &=& \{y \in N \mid
    \text{there is $x\in\eval{p.i_1...i_{m-1}}$ s.t.\ }
    (x,y)\in E \text{ and } \ledge(x,y)=i_m\}
  \end{array}$\\[1mm]
  When $\eval{p}=\emptyset$, we say that the path $p$ is \emph{missing} in $t$.
\end{definition}
Observe that, in the above definition, the semantics of paths allows for
skipping over intermediate arrays at every step in the path.

Given a tree $t$ and a path $p$, when $\type(x,t)=\mathsf{ty}$, for each
$x \in \eval{p}$, where $\mathsf{ty}\in\{\tarray,\tliteral,\tobject\}$, we
define the \emph{type of $p$ in $t$}, denoted $\type(p,t)$, to be $\mathsf{ty}$.
Also, when $\type(p,t)=\tarray$ and $\type(x,t)=\mathsf{ty}$ for each
$x\in \eval{p.i}$ for $i\in I$, we write $\type(p[\,],t)=\mathsf{ty}$.

\begin{figure}[t]
  \centering
  \scalebox{0.87}{\begin{tabular}{@{}l@{~}|@{~}l@{}}
    %%% MATCH
    Match &
    \begin{tabular}[t]{@{}ll@{\hspace{4cm}}ll}
      $t \models (p = v)$,
      & \multicolumn{3}{l}{if there is $x$ in $\eval{p}$ or $\eval{p.i}$ for $i\in I$
      such that $\mathsf{value}(x,t) = v$ holds}
      \\
      $t \models (\exists p)$, &if $\eval{p}\neq \emptyset$
      &
      $t \models \varphi_1 \land \varphi_2$, &if $t \models \varphi_1$ and
      $t \models \varphi_2$
      \\
      $t \models \neg \varphi$, &if $t \not\models \varphi$
      &
      $t \models \varphi_1 \lor \varphi_2$, &if $t \models \varphi_1$ or
      $t \models \varphi_2$
    \end{tabular}
    \\
    & $F\pipeline \match{\varphi} = \{ t \mid t \in F \text{ and }t \models \varphi\}$
    \\
    \hline\hline
    %%% PROJECT
    Project
    & $t \models (p_1=p_2)$ if there is a value $v$ such that $t\models (p_1=v)
    \land (p_2=v)$, or $t \models \neg(\exists p_1)\land \lnot(\exists p_2)$
    \phantom{$!^{!^!}$}
    \\
    % &
    % \begin{tabular}{@{}ll}
    %   $t \models (v_1=v_2)$ if $\models (v_1 = v_2)$.
    %   &
    %   $t \models v$, for a value $v$, if $v \notin \{\nullvalue,\falsevalue,0\}$.
    %   \\
    %   $t \models [d_1,\dots,d_n]$ always.
    %   &
    %   $t \models p$, for a path $p$, if $t\models (\exists p) \land \neg (p=v)$ for $v \in \{\nullvalue, \falsevalue, 0\}$.
    %   \\
    %   &
    % $t \models \cond{c}{\beta_1}{\beta_2}$ if $t\models (c \land \beta_1) \lor (\neg c \land \beta_2)$.
    % \end{tabular}
    % \\
    & $\project{p}(t) =\; \subtree(t,N_p)$, where $N_p$ are the nodes in $t$ on
    a path from $\treeroot(t)$ to a leaf via some $x \in \eval{p}$
    \\
    &
    $\project{q/p}(t)=\;
    \attach(q,\tree(v_{p}))$ if $t \models \exists p$, and
    $\emptyset$ otherwise
    \\
    &
    $\project{q/\cond{c}{d_1}{d_2}}(t)=\;
    \project{q/d_1}(t)$ if $t \models c$, and
    $\project{q/d_2}(t)$ otherwise
    \qquad~
    $\project{q/d}(t)=\;\attach(q,\tree(v_d))$, for other $d$
    \\[1mm]
    & $\project[\noid]{P}(t)=\;\bigoplus_{p \in P}\project{p}(t)~\oplus~\bigoplus_{(q/d)\in P}\project{q/d}(t)$
    \hspace{3cm} $\project{P}(t)=\; \project[\noid]{P}(t) \oplus
    \project{\id}(t)$
    \\
    & $F\pipeline \project[(\noid)]{P} = \{\project[(\noid)]{P}(t) \mid t\in F\}$
    \\\hline\hline
    %%% UNWIND
    Unwind
    & $\unwind{p}(t) =
    \{ (t \setminus \subtree(t, p)) \oplus \attach(p,\subtree(t, p.i))
    \}_{\eval{p.i}\neq\emptyset,\, i \in I}$ if $p$ is a first array, and
    $\emptyset$  otherwise
    \\
    & $\unwind[+]{p}(t) =
    \unwind{p}(t)$, if $\unwind{p}(t)\neq\emptyset$, and
    $\{t\}$ otherwise
    \\
    % & $\unwind{p}(t) =
    % \begin{cases}
    %   \{ (t \setminus \subtree(t, p)) \oplus \attach(p,\subtree(t, p.i)) \}_{i \in
    %     I,~ \eval{p.i}\neq\emptyset}, & \text{if $p$ is a first level array},\\
    %   \emptyset, & \text{otherwise}.
    % \end{cases}$
    % \\
    % & $\unwind[+]{p}(t) =
    % \begin{cases}
    %   \unwind{p}(t), & \text{if }\unwind{p}(t)\neq\emptyset,\\
    %   \{t\}, & \text{otherwise}.
    % \end{cases}$
    % \\
    & $F\pipeline \unwind[(+)]{p} = \bigcup_{t\in F} \unwind[(+)]{p}(t)$
    \\[0.5mm]
    \hline\hline
    %%% GROUP
    Group
    & $F \pipeline \group{}{a_1/b_1,..,a_m/b_m} =
     \big\{\attach(\id,\nullvalue) \oplus \bigoplus_{i=1}^m \attach(a_i, \colltoarray(F, b_i))\big\}$
    \\
    & $F \pipeline \group{g_1/y_1,..,g_n/y_n}{a_1/b_1,..,a_m/b_m} = $\\
    & ~~$\begin{array}[t]{@{}r@{~}l}
      \Big\{ \attach(\id, \object{}) \oplus %
      \bigoplus_{i=1}^m \attach\big(a_i,\colltoarray(F\pipeline\match{\varphi}, b_i) \big) \mid & %
      \varphi = \bigwedge_{j=1}^n(\neg\exists y_j),~ (F \pipeline \match\varphi) \neq \emptyset ~\Big\} \cup {}\\

      \Big\{ \bigoplus_{j \in J}\attach(\id.g_j,t_j) \oplus %
      \bigoplus_{i=1}^m \attach\big(a_i,\colltoarray(F\pipeline\match{\varphi}, b_i)  \big) \mid & %
      J \in 2^{\{1,..,n\}}\setminus\emptyset,\\
      \multicolumn{2}{r}{t_j\in \forest(F,y_j) \text{ for }j \in J,~
      \varphi = \bigwedge_{j \in J} ((y_j=t_j) \land \exists y_j)\land
      \bigwedge_{j\notin J}(\neg\exists y_j),~
      (F\pipeline\match{\varphi})\neq\emptyset~ \Big\}\phantom{{}\cup{}}}
    \end{array}$
    \\
    \hline\hline
    %%% LOOKUP
    Lookup
    & $\lookup{p_1 = C.p_2}{p}[F'](t) = t \oplus \attach(p,\colltoarray(F'\pipeline\match{\varphi},\varepsilon))$, for $\varphi = (p_2 = v_{p_1})$ if $t\models \exists p_1$, and $\varphi = \neg \exists p_2$ otherwise\\
    & $F \pipeline \lookup{p_1 = C.p_2}{p}[F'] =
    \{\lookup{p_1=C.p_2}{p}[F'](t) \mid t \in F\}$
  \end{tabular}
}
\caption{The semantics of \mquery stages. Here, for a value definition $d$ (and
  a tree $t$), we denote by $v_d$ the \emph{value associated} to $d$ in $t$,
  defined as $d$ if $d\in V$, as $\mathsf{value}(\subtree(t,d))$ if $d$ is a
  path, as the value of $(t\models d)$ if $d$ is a Boolean value definition, and
  as $[v_{d_1}, \dots, v_{d_m}]$ if $d=[d_1,\dots,d_m]$. A path $p$ is a
  \emph{first array} in $t$ if $\type(p,t)=\tarray$ and
  $\type(p',t) \neq \tarray$, for each strict prefix $p'$ of $p$.  }
  \label{fig:mquery-semantics1}
\end{figure}

In Figure~\ref{fig:mquery-semantics1}, we define the semantics of the \mquery
stages: specifically, given a forest $F$ and a stage
$s$, we define the forest $F\pipeline
s$ (for a lookup stage, we also require an additional forest $F'$ as parameter).
For the match and project stages, we define when a tree
$t$ satisfies a criterion or a Boolean value definition $\varphi$, denoted $t
\models \varphi$.
%
% In this definition
We employ the classical semantics for ``deep'' equality of non-literal values,
and assume that $(v=\nullvalue)$ holds iff $v$ is \nullvalue.

To define the semantics of the \emph{unwind}, \emph{project}, \emph{group}, and
\emph{lookup} operators, we use auxiliary operators over trees,
% to remove/insert subtrees and paths, etc,
informally introduced here (for a formal definition, see
Appendix~\ref{sec:tree-operations}).
Let $t$, $t_1$, $t_2$ be trees, $F$ a forest, $p$ a path, $N$ a set of nodes,
and $x$ a node. Then:
\begin{inparaenum}[\it (i)]
\item $\subtree(t,N)$ returns the subtree of $t$ induced by~$N$;
\item $\subtree(t,p)$ returns the subtree of $t$ hanging from $p$.  In the case
  where $|\eval{p}|>1$, it returns the array of single subtrees, and in the case
  where $\eval{p}=\emptyset$, it returns \nullvalue;
\item $\attach(p,t)$ constructs a new tree by attaching $p$ on top of the root
  of~$t$;
\item $t_1 \setminus t_2$ returns the tree resulting from removing the subtree
  $t_2$ from $t_1$;
\item $t_1 \oplus t_2$ constructs a new tree resulting from merging $t_1$ and
  $t_2$ by identifying nodes reachable via identical paths; and
\item $\colltoarray(F, p)$ constructs a new tree that is the array of all
  $\subtree(t,p)$ for $t\in F$, while $\forest(F,p)$ keeps all $\subtree(t,p)$
  in a set.
  % \nb{Check if needed}If $p=\varepsilon$, we write $\colltoarray(F)$.
\end{inparaenum}

We provide some comments on the semantics of \mquery. Let $t$ be the tree in
Figure~\ref{fig:mongodb-document-tree}.
\begin{compactitem}
\item Match can check both the value of an array, and the value (of a path)
  inside an array.  E.g., $t \models (\text{\small contribs=["OOP", "Simula"]})$
  and $t \models (\text{\small contribs="OOP"})$. Note that the values of
  several paths inside an array can come from different array elements.  E.g.,
  $t \models \text{\small (awards.award="Rosing Prize")} \land \text{\small
   (awards.year=2001)}$.
  % \nb{Possibly remove the following.}Selecting
  % only those persons (say, their names) who have got a particular award in a
  % specific year can be achieved by more complex queries:
  % $\project{\text{name},\, \text{awards}} \pipeline \unwind{\text{awards}}
  % \pipeline \match{\text{(awards.award="Rosing Prize")} \land
  % \text{(awards.year=2001)}} \pipeline \project{\text{name}}$, or
  % $\project{\text{name},\, \text{awards.award},\, \text{awards.year}}
  % \pipeline
  % \match{\text{(awards = \{award: "Rosing Prize",\, year: 2001\})}} \pipeline
  % \project{\text{name}}$.  % one can either duplicate the array
  % % \valuefont{awards} in \valuefont{aws}, unwind \valuefont{aws}, match using
  % % the
  % % conjunction of two conditions on \valuefont{aws}, and finally project only
  % % the original paths.

\item For $P=q/p$, when $p$ is missing in the input tree, then also $q$ is
  missing in the output tree.  E.g.,
  $\project{\text{newPath}/\text{nonExistingPath}}(t) =
  \tree(\text{\lobject\id:\,4\robject})$. Note, however, the difference when
  $P = q/[p]$:
%
% \item However, in some cases a non existing path gives rise to a null value,
%   e.g.:
%
  $\project{\text{newPath}/[\text{nonExistingPath}]}(t) =
  \tree${\small(\lobject\id:\,4, newPath:\,[\nullvalue{}]\robject)}.

\item When renaming paths inside an array, the array gets ``disassembled''.
  E.g., the result of
  $\project{\text{awsName}/\text{awards.award},\,
   \text{awsYear}/\text{awards.year}}(t)$ is
\begin{lstlisting}
{ _id: 4,
  awsName: [ "Rosing Prize", "Turing Award", "IEEE John von Neumann Medal" ],
  awsYear: [ 1999, 2001, 2001 ] }
\end{lstlisting}

\item In $\group{G}{A}$, when $G$ is empty all input trees are grouped in one
  output tree where the value of \id is \nullvalue, and when $G$ is non-empty,
  i.e., $G = g_1/y_1,\dots,g_n/y_n$, then, each group of input trees corresponds
  to a (possibly empty) subset $Y$ of $\{y_1,\dots, y_n\}$, so that the trees
  agree not only on the respective values reached through all the paths
  $y_j\in Y$, but also on the non-existence of paths not in $Y$.  The group for
  $Y=\emptyset$ has $\object{}$ as the value of \id.

\item The stage $\group{G}{\text{name.first},\, \text{name.last}, \dots}$
  produces two independent arrays of first and last names, so it is impossible
  to reconstruct the original first and last name pairs. In order to keep this
  connection, one should aggregate the common prefix, in this case
  $\group{G}{\text{name},\dots}$.
\end{compactitem}

\smallskip
\noindent
The semantics of an \mquery is obtained by composing (via $\pipeline$) the
answers of its stages.
\begin{definition}
  Let $\q = C \pipeline s_1 \pipeline \cdots \pipeline s_n$ be an \mquery.  The
  \emph{result of evaluating $\q$ over} a \mongodb instance $D$, denoted
  $\ansmongo(\q,D)$, is defined as $F_n$, where $F_0 = D.C$, and for $i\in
  \{1,\dots,n\}$, $F_i = (F_{i-1} \pipeline s_i)$ if $s_i$ is not a lookup
  stage, and $F_i = (F_{i-1} \pipeline s_i[D.C'])$ if $s_i$ is a lookup stage
  referring to an external collection name $C'$.
\end{definition}

\section{Expressivity of \mquery}
\label{sec:expressivity}

In this section we characterize the expressivity of \mquery in terms of
nested relational algebra (NRA), and we do so by by developing translations
between the two languages.

\subsection{Nested Relational View of \mongodb}

% To know what is the relational database corresponding to a \mongodb instance,
We start by defining a nested relational view of \mongodb instances.
In the case of a \mongodb instance with an irregular structure, there is no
natural way to define such a relational view.  This happens either when the
type of a path in a tree is not defined, or when a path has different types in
two trees in the instance.  Therefore, in order to define a schema for the
relational view, which is also independent of the actual \mongodb instances, we
impose on them some form of regularity.  We start by introducing the notion of
\emph{type} of a tree, which is analogous to complex object
types~\cite{Koch06}, and similar to JSON schema~\cite{json-schema-2016}.

\begin{definition}
  Consider JSON values constructed according to the following grammar:\\[1mm]
  \mbox{\qquad}\textup{\begin{tabular}{rcl}
      \meta{Type} &\DEF& \tliteral{} ~\CHOICE~
      \term{\lobject}\meta{List<Key}\term{:}\meta{Type>}\term{\robject}
      ~\CHOICE~ \term{[}\meta{Type}\term{]}
    \end{tabular}}\\[1mm]
   Given such a JSON value $d$, we call the tree $\tree(d)$ a \emph{type}.
  % We say that a tree $t$ \emph{complies with} a type $\tau$, if
  % $\type(p,\tau) = \type(p, t)$ and $\type(p[], \tau) = \type(p[], t)$, for
  % all paths~$p$ such that $\tau\models\exists p$ and $t\models\exists p$.
  %
  We say that a tree $t$ is \emph{of type} $\tau$ if for every path $p$ we have
  that $t\models \exists p$ implies
  \begin{inparaenum}[\itshape (i)]
  \item $\tau\models \exists p$,
  \item $\type(p,t)=\type(p,\tau)$, and
  \item $\type(p\arrayelement,t)=\type(p\arrayelement,\tau)$.
  \end{inparaenum}
  A forest $F$ is \emph{of type} $\tau$ if all trees in $F$ are of type $\tau$.
  A forest (resp., tree) is \emph{well-typed} if it is of some type.
\end{definition}

We now associate to each type $\tau$ a relation schema $\rschema(\tau)$ in
which, intuitively, attributes correspond to paths, and each nested relation
corresponds to an array in $\tau$.  In the following definition, given paths
$p$ and $q$, we say that $p.q$ is a \textit{simple extension} of $p$ if there
is no strict prefix $q'$ of $q$ such that $\type(p.q',\tau)=\tarray$.

\begin{definition}
  For a type $\tau$, the \emph{relation schema} $\rschema(\tau)$, is defined as
  $R_\tau(\ratt_\tau(\varepsilon))$, where, for a path $p$ in $\tau$,
  $\ratt_\tau(p)$ is the set of simple extensions $p'$ of $p$ such that $p'$ is
  an \emph{atomic attribute} if $\type(p', \tau)= \tliteral$, and %
  $p'$ is a \emph{sub-relation} if $\type(p', \tau)= \tarray$.  In the latter
  case, $p'$ has attributes $\{p'.\literalattr\}$ if
  $\type(p'\arrayelement, \tau)= \tliteral$, and $\ratt_\tau(p')$ otherwise.
  %
  % When $\tau$ is the type for a collection named $C$, we set $R_\tau = C$.
\end{definition}
Observe that the names of sub-relations and of atomic attributes in
$\rschema(\tau)$ are given by paths from the root in $\tau$, and therefore are
unique.

Next, we define the relational view of a well-typed forest.  In this view, to
capture the semantics of the missing paths, we introduce the new constant
\missingvalue.

\begin{definition}
  The \emph{relational view} of a well-typed forest $F$, denoted $\rel(F)$, is
  defined as $\{\row_\tau(R_\tau,\varepsilon, t) \mid t \in F\}$, where $\tau$
  is the type of $F$.
  For a relation name $R$ in $\rschema(\tau)$ and a path $p$,
  $\row_\tau(R, p, t)$ is the $R$-tuple
  $\{p.q:\relvalue(p.q,t)\}_{p.q\in\ratt_\tau(p)}$, where when
  $\eval{q}=\emptyset$, $\relvalue(p.q, t)$ is defined as $\missingvalue$,
  otherwise $\relvalue(p.q, t)$ is defined as\\
  \mbox{\quad}\begin{tabular}{ll}
   ~$\mathsf{value}(\subtree(t,q))$,& if $p.q$ is atomic;\\
   $\big\{ (p.q.\literalattr : \mathsf{value}(\subtree(t,q.i))) \mid
   \eval{q.i}\neq\emptyset, \text{ for } i \in I \big\}$,
   & if $\att_\tau(p.q)=\{p.q.\literalattr\}$;\\
   $\big\{ \row_\tau\left(p.q, p.q, \subtree(t,q.i)\right) \mid
   \eval{q.i}\neq\emptyset, \text{ for } i \in I \big\}$,
   & otherwise.
  \end{tabular}
\end{definition}

\begin{example} \label{ex:rel-schema}
  Consider the type $\tau_{\bios}$ for \bios:
\begin{lstlisting}
{ "_id": "literal",
  "awards": [ { "award": "literal", "year": "literal" } ],
  "birth": "literal",
  "contribs": [ "literal" ],
  "name": { "first": "literal", "last": "literal" } }
\end{lstlisting}
Then,  % $\rschema(\tau_\bios) = \bios(\id,
%   \valuefont{awards}(\valuefont{award}, \valuefont{year})$,
% $\valuefont{birth}$, $\valuefont{contribs}(\valuefont{\$literal}),
% \valuefont{name.first}, \valuefont{name.last})$.
$\rschema(\tau_\bios)$ is defined as \bios\/\big(\id, awards(awards.award,
awards.year), birth, contribs(contribs.\literalattr), name.first,
name.last\big).
Moreover, for the tree $t$ in Figure~\ref{fig:mongodb-document-tree}, the
relational view $\rel(\{t\})$ is illustrated in Figure~\ref{fig:rel-view}.
\begin{figure*}
  \centering
  \scalebox{0.8}{
    \begin{tabular}{|@{~}c@{~}|@{~}c@{~}|@{~}c@{~}|@{~}c@{~}|@{~}c@{~}|@{~}c@{~}|}
    \hline
    \id &
    \begin{tabular}{@{}l@{}}
      \\[-3.5mm]
      \begin{tabular}{@{}|@{~}P{4.6cm}@{~}|@{~}P{1.6cm}@{~}|@{}}
        \hline
        \multicolumn{2}{|c|}{awards} \\
        \hline
        awards.award & awards.year \\\hline
      \end{tabular}
      \\[-3.5mm]~
    \end{tabular}
    & birth &
    \begin{tabular}{@{}c@{}}
      contribs\\\hline
      contribs.\literalattr
    \end{tabular}
    & name.first & name.last
    \\\hline
    4
    &
    \begin{tabular}{@{}l@{}}
      \\[-3.5mm]
      \begin{tabular}{@{}|@{~}P{4.6cm}@{~}|@{~}P{1.6cm}@{~}|@{}}
        \hline
        Rosing Prize & 1999\\ \hline
        Turing Award & 2001 \\ \hline
        IEEE John von Neumann Medal & 2001\\ \hline
      \end{tabular}
      \\[-3.5mm]~
    \end{tabular}
    &
    1926-08-27
    &
    \begin{tabular}{@{}|c|@{}}
      \hline
      OOP\\ \hline
      Simula\\
      \hline
    \end{tabular}
    &
    Kristen
    &
    Nygaard
    \\\hline
  \end{tabular}}
  \caption{Relational view of the document about Kristen Nygaard}
  \label{fig:rel-view}
\end{figure*}
\qed
\end{example}

To define the relational view of \mongodb instances, we introduce the notion of
\emph{(\mongodb) type constraints}, which are given by a set $\S$ of pairs
$(C,\tau)$, one for each collection name $C$, where $\tau$ is a type.  We say
that a database $D$ \emph{satisfies} the constraints $\S$ if $D.C$ is of type
$\tau$, for each $(C,\tau) \in \S$. For a given $\S$, for each $(C,\tau)\in\S$,
we refer to $\tau$ by $\tau_C$.  Moreover, we assume that in
$\rschema(\tau_C)$, the relation name $R_{\tau_C}$ is actually $C$.

\begin{definition}
  Let $\S$ be a set of type constraints, and $D$ a \mongodb instance satisfying
  $\S$.  The \emph{relational view} $\rdb_{\S}(D)$ of $D$ with respect to $\S$
  is the instance $\{\rel(D.C) \mid (C,\tau)\in\S\}$.
\end{definition}

Finally, we define equivalence between \mqueries and NRA queries.  To this
purpose, we also define equivalence between two kinds of answers: well-typed
forests and nested relations.

\begin{definition}
  A well-typed forest $F$ is \emph{equivalent} to a nested relation $\R$,
  denoted $F \simeq \R$, if $\rel(F)=\R$.
  % \end{definition}
  %
  % \begin{definition}
  An \mquery~$\q$ is \emph{equivalent to} an NRA query $Q$ w.r.t.\ type
  constraints~$\S$, denoted $\q\equiv_\S Q$, if
  $\ansmongo(\q,D) \simeq \ansra(Q,\rdb_{\S}(D))$, for each \mongodb instance
  $D$ satisfying $\S$.
\end{definition}

Notice that, the above definition of equivalence between well-typed forests and
nested relations appears to be asymmetric,
% since it refers only the transformation $\rel_\tau$ in one direction, and
% therefore
since it would in principle allow for nested relations that are not equivalent
to any well-typed forest.  We notice, however, that the \mongodb view of a
nested relation always exists, is well-typed, and can be defined in a
straightforward way.  Therefore, we can consider both translations (between NRA
and \mquery, and vice-versa), as defined on well-typed forests and their
relational views.

\subsection{From NRA to \mquery}
\label{sec:nra2mquery}

We now show that \mupgl captures NRA, while \mupg captures NRA over a single
collection.

In our translation from NRA to \mquery, we have to deal with the fact that an
NRA query in general has a \emph{tree} structure where the leaves are relation
names, while an \mquery contains \emph{one sequence} of stages.  So, we first
show how to ``linearize'' tree-shaped NRA expressions into a \mongodb pipeline.
More precisely, we show that it is possible to combine two \mupg sequences
$\q_1$ and $\q_2$ of stages into a single \mupg sequence $\pipe(\q_1,\q_2)$, so
that the results of $\q_1$ and $\q_2$ can be accessed from the result of
$\pipe(\q_1,\q_2)$ for further processing.
We define $\pipe(\q_1,\q_2)$ as
$\spec \pipeline \subq_1(\q_1) \pipeline \subq_2(\q_2)$.  The idea of $\spec$
is to duplicate each tree $t$ in the input forest to $t_1$ and $t_2$, and to
specialize them with the aim that $t_j \models (\actRel=j)$ and the copy of $t$
is stored in $t_j$ under the key \valuefont{rel$j$}, for $j\in\{1,2\}$.  The
idea of $\subq_j(\q_j)$ is to execute $\q_j$ so that it affects only the trees
with~{\small$(\actRel=j)$}, hence does not interfere with $\q_{3-j}$, and
stores its result in the trees under the key \valuefont{rel$j$}.

We set
$\spec=\project[\noid]{\text{origDoc}/\varepsilon,\,\actRel/[1,2]}
\pipeline \unwind{\actRel} \pipeline
\project{\actRel,\,\{\text{rel}i/(\cond{\actRel=i)}{\text{origDoc}}{\text{dummy}}\}_{i=1,2}}$,
where \valuefont{dummy} is a path that does not exist in any collection.
In this way, we obtain that
$
  \{t\} \pipeline \spec =
  \{\ \tree(\text{\small \lobject\actRel:\,1, rel1: $t$\robject}), \
  \tree(\text{\small \lobject\actRel:\,2, rel2: $t$\robject})\ \}
$, for each tree $t$.

\begin{figure}
  \centering
\scalebox{0.93}{$
  \begin{array}{@{}c@{~}|@{~}l@{~}|@{}c@{~}|@{~}l@{}}
    s&\subq_j(s) & s&\subq_j(s)
    \\\hline
    \match{\varphi} & \match{(\actRel=3-j) \lor \varphi_{[p\to\text{rel}j.p]}}
    &
    ~\group{g/y}{a/b}
    &
    \multirow{4}{*}{\begin{array}[t]{@{}l@{}}
        \group{\text{rel}j.g/\text{rel}j.y,\, \actRel}{\text{rel}j.a/\text{rel}j.b,\, \text{rel}(3-j)} \pipeline{} \\
        \project[\noid]{\substack{\text{rel}(3-j),\, \actRel/\id.\actRel,\, \text{rel}j.a,\qquad~~\\\qquad\qquad\qquad\qquad\text{rel}j.\id.g/\id.\text{rel}j.g}} \pipeline{} \\
        \project{\actRel,\, \{\text{rel}i/\cond{(\actRel=i)}{\text{rel}i}{\text{dummy}}\}_{i=1,2}}  \pipeline{}\\
        \unwind[+]{\text{rel}(3-j)}
      \end{array}}
    \\
    % \cline{1-2}
    \unwind[+]{p} & \unwind[+]{\text{rel}j.p} &&
    \\
    % \cline{1-2}
    \unwind{p} & \match{(\actRel=3-j) \lor ((\exists \text{rel}j.p) \land \neg(\text{rel}j.p=[]))} \pipeline \unwind[+]{\text{rel}j.p} &&
    \\
    % \cline{1-2}
    \project{p,\,q/d} & \project{\substack{\text{rel}(3-j),\, \actRel,\, \text{rel}j.\id,\, \text{rel}j.p,\qquad\qquad\quad~\\\qquad\text{rel}j.q/\cond{(\actRel=j)}{d_{[q'\to\text{rel}j.q']}}{\text{dummy}}}} &&
  \end{array}
$
}
\caption{Subquery $\subq_j(s)$ for stage $s$, where we have detailed only the
 short forms for project and group stages. We use $e_{[p\to q]}$ to denote the
 expression $e$ in which every occurrence of the path $p$ is replaced by the
 path $q$.}
  \label{fig:subquery}
\end{figure}

As for $\subq_j(\q_j)$, $j\in\{1,2\}$, it is defined as
$\subq_j(s_1) \pipeline\cdots\pipeline \subq_j(s_n)$, for
$\q_j = s_1 \pipeline \cdot\cdot\cdot \pipeline s_n$, where $\subq_j$ for
single stages is defined in Figure~\ref{fig:subquery}.
Since the idea of $\subq_j(s)$ is to affect only the trees with
{\small$(\actRel=j)$}, $\subq_j(\mu_\varphi)$, selects all trees with
{\small$(\actRel=3-j)$}, while among the trees with {\small$(\actRel=j)$} it
selects only those that satisfy $\varphi$, in which all original paths $p$ are
replaced by \valuefont{rel$j.p$}.
The unwind stage $\unwind{p}$ cannot be implemented simply by
$\unwind{\text{rel}j.p}$, since all trees with {\small$(\actRel=3-j)$} would be
lost (they do not contain the path \valuefont{rel$j.p$}).  Therefore we rely on
$\unwind[+]{\text{rel}j.p}$, selecting among the trees with
{\small$(\actRel=j)$} only those where the path $\text{rel}j.p$ is present and
its value is not the empty array.
The encoding of the project stage $\project{p,\,q/d}$ needs to make sure that
\valuefont{rel$(3-j)$} and \valuefont{\actRel} are not lost, and that the path
\valuefont{rel$j.q$} is not created in the trees with {\small$(\actRel=3-j)$}
(guaranteed by the conditional expression for $q/d$).
The encoding of the group stage $\group{g/y}{a/b}$
\begin{inparaenum}[\itshape (i)]
\item adds \valuefont{\actRel} to the grouping condition so as to group all
  trees with {\small$(\actRel=3-j)$} in one tree,
\item renames the paths \valuefont{\id.\actRel} and \valuefont{\id.rel$j.g$},
\item normalizes the trees by making sure that the trees with
  {\small$(\actRel=i)$} contain only \valuefont{rel$i$} but not
  \valuefont{rel$(3-i)$}, and finally
\item unwinds the array \valuefont{rel$(3-j)$} where all original trees with
  {\small$(\actRel=3-j)$} have been aggregated.
\end{inparaenum}
% The above can immediately be generalized to encode project and group stages
% with arbitrary conditions.

\begin{example}
  Consider the sequences of stages
  $\q_1=\match{\text{name.first}=\text{"Kristen"}}\pipeline\project{\text{name}}$
  and $\q_2=\match{\exists\text{awards}}\pipeline\project{\text{awards}}$. Then
  $\pipe(\q_1,\q_2)$ is the following sequence of stages:\\[1mm]
  \mbox{\qquad}$\begin{array}{@{}r@{~}l}
      \spec ~\pipeline~ & \match{(\actRel=2)\lor(\text{rel1.name.first}=\text{"Kristen"})} ~\pipeline~ %
        \project{\text{rel2},\,\actRel,\,\text{rel1}.\id,\,\text{rel1.name}}\pipeline{}\\
      & \match{(\actRel=1)\lor(\exists\text{rel2.awards})} ~\pipeline~ %
        \project{\text{rel1},\,\actRel,\,\text{rel2}.\id,\,\text{rel2.awards}}\\
  \end{array}$\\[1mm]
  Let $t$ be the tree in Figure~\ref{fig:mongodb-document-tree}. The result
  of $\{t\} \pipeline\pipe(\q_1,\q_2)$ consists of two trees:
\begin{lstlisting}
{ "actRel": 1,
  "rel1": {"_id": 4,  "name": { "first": "Kristen", "last": "Nygaard" } } },
{ "actRel": 2,
  "rel2": {"_id": 4,  "awards": [
      { "award": "Rosing Prize", "year": 1999, "by": "Norwegian Data Association" },
      { "award": "Turing Award", "year": 2001, "by": "ACM" },
      { "award": "IEEE John von Neumann Medal", "year": 2001, "by": "IEEE" } ] } }
\end{lstlisting}
% So, the result of $\{t\}\pipeline \q_1$ is found in the first tree under the
% key \valuefont{rel1}, and the result of $\{t\}\pipeline \q_2$ is found in the
% second tree under the key \valuefont{rel2}.
%
\vspace{-0.7cm}
\qed
\end{example}

\begin{figure*}
  \centering
\scalebox{0.9}{$%\begin{array}{ll}
    \begin{array}[t]{c|l}
      Q & \raTomaq(Q)
      \\\hline
      C & \project{\att(C)}
      \\%[2.5mm]
      \sigma_\psi(Q) &
      \raTomaq(Q) \pipeline
      \project{\att(Q),\, \text{cond}/\psi} \pipeline
      \match{\text{cond}=\mathbf{true}} \pipeline \project{\att(Q)}
      \\%[2.5mm]
      \pi_{S}(Q) &
      \raTomaq(Q) \pipeline
      \project{S}
      \\%[2.5mm]
      \nest_{S \to b}(Q) & \raTomaq(Q) ~\pipeline %
      \project{(\att(Q)\setminus S),\, \{b.p/p \,\mid\, p \in S\}}\pipeline
      \group{(\att(Q)\setminus S)}{b} ~\pipeline~
      \project[\noid]{b,\, \{p/\id.p \,\mid\, p \in\att(Q)\setminus S\}}
      \\%[2.5mm]
      \unnest_a(Q) & \raTomaq(Q) ~\pipeline~
      \unwind{a} % ~\pipeline~
      % \project{\att(\unnest_a(Q))}
      \\
    % \end{array}
    % &
    % \begin{array}[t]{c|l}
      % Q & \raTomaq(Q)
      % \\\hline
      Q_1 \times Q_2 & \pipe(\raTomaq(Q_1),\raTomaq(Q_2)) \pipeline
      \group{}{\text{rel1},\, \text{rel2}} ~\pipeline~
      \unwind{\text{rel1}} ~\pipeline~
      \unwind{\text{rel2}}
      \\%[2mm]
      Q_1\cup Q_2 & \pipe(\raTomaq(Q_1),\raTomaq(Q_2)) \pipeline
      \project{\text{rel1},\,\text{rel2},\,\{\text{p}i/\cond{(\actRel=1)}{\text{rel}1.\text{p}i}{\text{rel}2.\text{p}i}\}_{i=1}^n}\pipeline{}\\[-1mm]
      & \group{\text{p1},\dots,\text{p}n}{} ~\pipeline~
      \project[\noid]{\{\text{p}i/ \id.\text{p}i\}_{i=1}^n}
      \\[1mm]
      Q_1\setminus Q_2 & \pipe(\raTomaq(Q_1),\raTomaq(Q_2)) \pipeline
      \project{\text{rel1},\,\text{rel2},\, \{\text{p}i/\cond{(\actRel=1)}{\text{rel}1.\text{p}i}{\text{rel}2.\text{p}i}\}_{i=1}^n}\pipeline{}\\[-1mm]
      &\group{\text{p1},\dots,\text{p}n}{\text{rel2}} ~\pipeline~
      \match{\text{rel2}=\arrayelement} ~\pipeline~
      \project[\noid]{\{\text{p}i/\id.\text{p}i\}_{i=1}^n}
      \\
    \end{array}
    % \end{array}
      $
}
  \caption{Translation from NRA to \mupg}
  \label{fig:nra2mq}
\end{figure*}

We start with a singleton set $\S=\{(C,\tau_C)\}$ of type constraints for a
collection name~$C$, and consider an NRA query $Q$ over the relation name $C$
(with schema $\rschema(\tau_C)$).  The translation of $Q$ is the \mupg query
$C \pipeline \raTomaq(Q)$, where $\raTomaq(Q)$ is defined recursively in
Figure~\ref{fig:nra2mq}, where we overload the function $\att$ and assume that
for an NRA query $Q'$, $\att(Q')$ is the attribute set of the schema implied by
$Q'$.
The translation of $Q_1 \times Q_2$ first groups all input trees in one tree,
where all trees $t_i$ that are the answers to $Q_i$ are aggregated in arrays
\valuefont{reli}, and then unwinds these arrays, thus producing all possible
pairs $(t_1,t_2)$.  The translations of $Q_1 \cup Q_2$ and $Q_1 \setminus Q_2$,
where we assume that $\att(Q_i) = \{p_1,\dots,p_n\}$, first create fresh paths
\valuefont{p$i$} in each tree to be used in the grouping condition. Then, in
the case of union it only remains to rename the paths \valuefont{\id.p$i$} back
to \valuefont{p$i$}, while in the case of difference, we also select only those
``tuples'' $(p_1,\dots,p_n)$ that were not present in the answer to $Q_2$.

\begin{theorem}
  \label{thm:nra-to-mupg-correct}
  Let $Q$ be a NRA query over $C$.  Then $C \pipeline \raTomaq(Q) \equiv_\S Q$.
\end{theorem}

Next, we consider NRA queries across several collections, and show how to
translate them to \mupgl.  Let $\S$ be a set of type constraints, and $Q$ an
NRA query over the schemas for collections named $C_1,\ldots,C_n$, with $n\geq
2$.  Let us take $C_1$ to be the collection over which we evaluate the
generated \mquery.  Then, we first need to ``bring in'' the trees from the
collections $C_2,\ldots,C_n$, which we do in a preparatory phase
$\mathsf{bring}(C_2,\ldots,C_n)$, defined as:
\[
  \begin{array}{@{}l}
  \group{}{\text{coll}1/\epsilon} \pipeline
  \lambda^{\text{dummy} = C_2.\text{dummy}}_{\text{coll}2} \pipeline\cdots\pipeline
  \lambda^{\text{dummy} = C_n.\text{dummy}}_{\text{coll}n} \pipeline%{}\\
  \project{\text{coll}1,..,\text{coll}n,\, \text{actColl}/[1..n]} \pipeline
  \unwind{\text{actColl}} \pipeline{} \\
  \project{\text{actColl},\,
   \{\text{coll}i/\cond{(\text{actColl}=i)}{\text{coll}i}{\text{dummy}}\}_{i=1}^n}
  \pipeline \unwind[+]{\text{coll}1} \pipeline \cdots \pipeline
  \unwind[+]{\text{coll}n}
\end{array}
\]
Second, we define a function $\raTomaq^\star(Q)$ that differs from
$\raTomaq(Q)$ in the translation of the collection names:
$\raTomaq^\star(C_i) = \match{\text{actColl}=i} \pipeline
\project{\{p/\text{coll}i.p \,\mid\, p \in \att(C_i)\}}$.
Finally, the translation of $Q$ is the \mupgl query
$C_1 \pipeline \mathsf{bring}(C_2,\dots,C_n) \pipeline \raTomaq^\star(Q)$.

\begin{theorem}\label{thm:nra-to-mupgl-correct}
  Let $Q$ be an NRA query over $C_1,\dots,C_n$, and
  $\q=C_1 \pipeline \mathsf{bring}(C_2,\ldots,C_n) \pipeline
  \raTomaq^\star(Q)$.  Then $\q\equiv_\S Q$.  Moreover, the size of $\q$ is
  polynomial in the size of $Q$.
\end{theorem}

Thus, we obtain that \mupgl captures full NRA, and that \mupg captures NRA over
a single collection.
We observe that the above translation serves the purpose of understanding the
expressive power of \mquery, but is likely to produce queries that \mongodb
will not be able to efficiently execute in practice, even on relatively small
database instances.
We also note that the translation from NRA to \mquery works even if we allow
for database instances $D$ such that $D.C$ is not strictly of type $\tau_C$,
but may also contain other paths not in $\tau_C$.
% However, the translation can be implemented differently in practice, and we
% have developed also an alternative, more involved, translation of binary RA
% constructs.  We have also devised optimization techniques that allow us to
% produce queries that execute more efficiently than the ones obtained with the
% more direct translation.  These techniques, and evidence about their
% effectiveness are reported in Section~\ref{sec:optimization}.

\subsection{From \mquery to NRA}
\label{sec:mquery2nra}

In this section, we aim at defining a translation from \mquery to NRA, and for
this we want to exploit the structure, i.e., the stages of \mqueries.  Hence,
we define a translation $\maqTonra{s}$ from stages $s$ to NRA expressions
% possibly composed of several operators.
such that, for an \mquery $C \pipeline s_1 \pipeline \cdots \pipeline s_n$, the
corresponding NRA query is defined as
$C\circ \maqTonra{s_1} \circ \cdots \circ \maqTonra{s_n}$\footnote{We follow
 the convention that $(f \circ g)(x) = g(f(x))$.}, where we identify the
collection name $C$ with the corresponding relation schema in the relational
view.
However, such translation might not always be possible, since \mquery is
capable of producing non well-typed forests, for which the relational view is
not defined.  This capability is due to value definitions in a project
operator: already a query as simple as
$\project{\text{a}/\cond{\id=1}{[0,1]}{\text{"s"}}}$ produces from the
well-typed forest
$\{\tree(\text{\object{\id:\,1}}),\ \tree(\text{\object{\id:\,2}})\}$ a non
well-typed one:
$\{\tree(\text{\object{\id:\,1, a:\,[0,1]}}), \tree(\text{\object{\id:\,2,
  a:\,"s"}})\}$. Therefore, in order to derive such a translation
$\maqTonra{s}$, we restrict our attention to \mqueries with stages preserving~well-typedness.

\begin{definition}
  % Given a set $\S$ of type constraints, \mquery $\q$ is \emph{well-typed} for
  % $\S$, if for each \mongodb instance $D$ satisfying $\S$, $\ansmongo(\q,D)$ is
  % a well-typed forest.
  % 
  Given a type $\tau$ (and a type $\tau'$), a stage $s$ is \emph{well-typed}
  for $\tau$ (and $\tau'$), if for each forest $F$ of type $\tau$ (and each
  forest $F'$ of type $\tau'$), $F \pipeline s$ (resp., $F \pipeline s[F']$
  when $s$ is a lookup stage) is a well-typed forest.
  % Given a set $\S$ of type constraints, \mquery $\q$ is \emph{well-typed} for
  % $\S$, if for each \mongodb instance $D$ satisfying $\S$, $\ansmongo(\q,D)$ is
  % a well-typed forest. A stage $s$ is \emph{well-typed} for $\S$ if the \mquery
  % of single stage $s$ is well-typed for $\S$. \nb{GX: added}
\end{definition}

We observe that the match, unwind, group and lookup stages are always
well-typed, and, given such a stage $s$ and input types $\tau$, $\tau'$, we can
compute the output type $\tau_o$ of~$s$:
\begin{inparaenum}[\itshape (i)]
\item match does not change the input type, i.e., $\tau_o = \tau$,
\item for unwind and group stages $s$ it is obtained by evaluating $s$ over
  $\{\tau\}$, i.e., $\{\tau_o\}=\{\tau\} \pipeline s$, and
\item similarly, the output type for a lookup stage is the single tree in
  $(\{\tau\} \pipeline \lookup{p_1 = C.p_2}{p} [\{\tau'\}])$.
\end{inparaenum}
As for a project stage $s = \project{P}$ and an input type $\tau$, we can check
whether $s$ is well-typed for $\tau$, and if yes, we can compute the output
type $\tau_o$ of $s$, as follows.  For each $p/d \in P$, we compute the type
$\tau_d$ of $d$ with respect to $\tau$; if all $\tau_d$ are defined, then $s$
is well-typed and $\tau_o$ is the type where $\subtree(\tau_o,p)$ coincides
with $\tau_d$ for each $p/d\in P$, and that agrees with $\tau$ on all $p\in P$;
% $\tau_o=\tree(\object{\{p: \tau_d\}_{p/d\in P}, \{p: \subtree(\tau,p)\}_{p \in  P}})$,
otherwise $s$ is not well-typed. The \emph{type} $\tau_d$ of a value definition
$d$ with respect to a type $\tau$ is defined inductively as follows:
\begin{inparablank}
\item $\tau_v = \tau'$ for a value $v$, if $v$ is of type $\tau'$, and
  undefined otherwise;
\item $\tau_\beta = \tree(\tliteral)$ for a Boolean value definition~$\beta$;
\item $\tau_p = \subtree(\tau, p)$, for a path~$p$;
\item $\tau_{[d_1,\dots,d_n]} = \tree([\tau_{d_1}])$ if
  $\tau_{d_1}=\cdots=\tau_{d_n}$, and undefined otherwise;
\item $\tau_{\cond{c}{d_1}{d_2}}$ is $\tau_{d_1}$ if $c$ is valid, $\tau_{d_2}$
  if $c$ is unsatisfiable, $\tau_{d_1}$ if $c$ is satisfiable and not valid and
  $\tau_{d_1}=\tau_{d_2}$, and undefined otherwise.
\end{inparablank}

Then, given a set $\S$ of type constraints and an \mquery
$\q=C \pipeline s_1 \pipeline \cdots \pipeline s_n$, we can check whether each
stage in $\q$ is well-typed for its input type determined by $\q$ and $\S$.  To
do so, we take the input type for $s_1$ to be $\tau_0$, where
$(C,\tau_0) \in \S$, and we compute sequentially the input type for each stage
$s_i$, as long as this is possible, i.e., all stages preceding it are
well-typed.

The translation $\maqTonra{s}$, for well-typed stages $s$, is quite natural,
although it requires some attention to properly capture
% the flexibility and
the semantics of \mquery.
% Due to space limitations
It is reported in Appendix~\ref{app:mquery2nra}.

\begin{theorem}
  \label{thm:maq-to-nra-pipeline}
  Let $\S$ be a set of type constraints, $\q$ an \mquery
  $C \pipeline s_1 \pipeline \cdots \pipeline s_m$ in which each stage is
  well-typed for its input type, and
  $Q = C\circ \maqTonra{s_1}\circ \cdots \circ\maqTonra{s_m}$.  Then
  $\q \equiv_\S Q$, moreover, the size of $Q$ is polynomial in the size of $\q$
  and $\S$.
\end{theorem}

A natural question that comes up is in which cases an \mquery can be translated
to NRA even if contains stages that are not well-typed.  E.g., in the example
% at the beginning of the subsection,
above, this can happen when the path \valuefont{a} is projected away in the
subsequent stages without being actually used.  We leave this problem for
future work.

%%% Local Variables:
%%% mode: latex
%%% TeX-master: "icdt18"
%%% fill-column: 79
%%% End:

\section{Complexity of \mquery}
\label{sec:complexity}

In this section we report results on the complexity of different fragments of
\mquery.
Specifically, we are concerned with the combined and query complexity of the
\emph{Boolean query evaluation} problem, which is the problem of checking
whether the answer to a given query over a given database instance is
non-empty.

Our first result establishes that the full \mupgl and also \mupg are complete
for exponential time with a polynomial number of alternations under \LOGSPACE
reductions \cite{ChKS81,John90}.  That is, have the same complexity as
monad algebra with atomic equality and negation \cite{Koch06}, which however is
strictly less expressive than NRA.

\begin{theorem}\label{lem:mupg-nexptime-complete}
  \mupg and \mupgl are \TAexppoly-complete in combined complexity, and in \ACz
  in data complexity.
\end{theorem}
\begin{proof}[Proof Sketch]
The proof of the lower bound follows the line of the \TAexppoly-hardness proof
in \cite{Koch06}.
As for the upper bound, we provide an algorithm that follows a strategy based
on starting the alternating computation from the last stage, inspired by a
similar one in \cite{Koch06}.  Let $\q$ be an \mupgl query and $D$ a database
instance.  We check whether there is a tree in $\ansmongo(\q,D)$ using an
alternating Turing machine running in exponential time with polynomially many
alternations.

Intuitively, for a forest $F'$ resulting from applying a stage $s$ in $\q$ to a
previous result $F$, i.e., $F'=F\pipeline s$, in general we need to check
whether there is a tree and/or all trees in $F'$ that satisfy some conditions
(such as, the value of a path $p$ in such a tree should/should not be $v$, or a
path $p$ should/should not exist), without explicitly constructing $F'$.  To do
so, we derive from the conditions on $F'$ suitable conditions to be checked on
$F$.
Such conditions are obtained/guessed from the criteria in match stages, and
Boolean value definitions and conditional value definitions in project stages.
Both branching and alternations happen because of the group stage. For
instance, if $s=\group{}{a_1/b_1,\,a_2/b_2}$ and the conditions on $F'$ contain
$a_1=[]$, then we need to check that there is no tree in $F$ satisfying
$\exists b_1$.  If $s=\group{g/y}{a_1/b_1,\,a_2/b_2}$ and the conditions on
$F'$ contain $\id.g=v$, $a_1\neq[]$ and $a_2\neq[]$, then we need to check
whether in $F$ there is a tree satisfying $y=v$ and $\exists b_1$, and %there is
a tree satisfying $y=v$ and $\exists b_2$.

The overall computation starts from $F'=\ansmongo(\q,D)$, and propagates the
constraints on the intermediate forests to the previous stages.
The ``depth'' of the checks is given by the number of stages, the branching and
the number of alternations are bounded by the size of $\q$, which give us
\TAexppoly upper bound.
The bound in data complexity can be shown as for NRA, known to be in \ACz
\cite{SuTa97}.
\end{proof}

As a corollary, we obtain a tight bound for the combined complexity of NRA.

\begin{corollary}\label{cor:NRA-complexity}
  NRA is \TAexppoly-complete in combined complexity.
\end{corollary}

Next, we study some of the less expressive fragments of \mquery.  We consider
match to be an essential operator, and we start with the minimal fragment \mq,
for which we show that query answering is tractable and very efficient.
\begin{theorem}\label{lem:mq-logspace-complete}
  \mq is \LOGSPACE-complete in combined complexity.
\end{theorem}
\begin{proof}[Proof Sketch]
The lower-bound can be shown by a reduction from the directed forest
accessibility problem, known to be complete for \LOGSPACE under \NCone
reducibility \cite{CoMc87}, to the problem whether $t\models \exists p$, for a
tree $t$ and a path $p$.
The upper-bound follows from the following facts:
\begin{inparaenum}[\itshape (i)]
\item we can check in \LOGSPACE whether $t \models (p=v)$ and whether
  $t \models \exists p$, for a tree $t$, a path $p$, and a value $v$;
\item tree-isomorphism, needed to check equality between the sub-tree reached
  through a path $p$ and a complex value $v$ is in \LOGSPACE \cite{Lind92};
\item the Boolean formula value problem is A\LOGTIME-complete \cite{Buss87},
  and hence in \LOGSPACE.
\end{inparaenum}
\end{proof}

Next, we observe that the project and group operators allow one to create
exponentially large values by duplicating the existing ones. For instance, the
result of $\{\tree(\lobject a{:}1\robject)\} \pipeline s_1
\pipeline\cdots\pipeline s_n$, for $s_1=\cdots=s_n=\project{a.\ell/a,\,a.r/a}$,
is a set consisting of a full binary tree of depth $n$. Nevertheless, without
the unwind operator it is still possible to maintain tractability.

\begin{theorem}\label{lem:mpg-ptime-complete}
  \mp is \PTIME-hard in query complexity and \mpgl is in \PTIME in combined
  complexity.
\end{theorem}
\begin{proof}[Proof Sketch]
  The lower-bound follows from the fact that we can compute the value of a
  monotone Boolean circuit consisting of assignments to $n$ variables in $n$
  project stages, and in the final match stage we can check whether the output
  variable evaluates to 1.
  For the upper-bound, we notice that it is not necessary to materialize the
  exponentially large trees, instead we can work on their compact
  representations in the form of directed acyclic graphs (DAGs). Thus,
  % without the unwind operator
  we can devise an algorithm for which the result of each stage grows at most
  linearly in the size of the stage and its input set of DAGs.  Hence, we can
  evaluate each stage on a structure that is at most polynomial.
\end{proof}

We can identify the unwind operator as one of the sources of complexity, as it
allows one to multiply the number of trees each time it is used in the
pipeline.  Indeed, adding the unwind operator alone causes already loss of
tractability, provided the input tree contains multiple arrays (hence in
combined complexity).

\begin{theorem}\label{thm:mu-complexity}
  \muq is \LOGSPACE-complete in query complexity and \NP-complete in combined
  complexity.
\end{theorem}

\begin{proof}[Proof Sketch]
For the \LOGSPACE upper-bound, we observe that the number of times the unwind
operation can actually multiply the number of trees is bounded by the number of
arrays that are present in the input tree, and hence by a constant.  Hence, we
can both compute the result of the unwind stages, and evaluate the match
conditions in \LOGSPACE in the size of the query.
The \NP lower-bound results from a straightforward encoding of the Boolean
satisfiability problem: we start from an input forest containing $n$ arrays
$[0,1]$, then we generate with $n$ unwind stages all $2^n$ assignments, and
finally we check with a match stage whether there is a satisfying one.
The \NP upper-bound follows from the next theorem.
\end{proof}

Adding project and lookup does not increase the combined complexity, but does
increase the query complexity, since they allow for creating multiple arrays
from a fixed input tree.

\begin{theorem}\label{lem:mupl-np-complete}
  \mup and \mul are \NP-hard in query complexity, and \mupl is in \NP in
  combined complexity.
\end{theorem}
\begin{proof}[Proof Sketch]
The proof of the lower-bound is analogous to the one for the \NP lower-bound in
Theorem~\ref{thm:mu-complexity}, except that now we can use either project or
lookup to generate the forest with $n$ arrays $[0,1]$.
For the upper-bound, we extend the idea of using DAGs as compact
representations of trees. We only specify how to evaluate an unwind stage:
instead of creating a separate DAG for each element of the array, we guess an
element of the array and produce at most one DAG for each input DAG.  This is
sufficient, since without group, we can evaluate each original tree
independently of the other ones.
\end{proof}

In the presence of unwind, group provides another source of complexity, since
in \mug we can generate doubly exponentially large trees, analogously to monad
algebra \cite{Koch06}. Let $t_0 = \tree(\lobject \id:\lobject
x:0\robject\robject)$ and $t_1 = \tree(\lobject \id:\lobject
x:1\robject\robject)$. The result of applying the \mug query $s_1 \pipeline
\cdots \pipeline s_n$, where %
$s_i = \group{}{x/\id.x} \pipeline%
\group{x.l/x,\, x.r/x}{} \pipeline %
\unwind{\id.x.l} \pipeline \unwind{\id.x.r}$, to $\{t_0,t_1\}$ is a forest
containing $2^{2^n}$ trees, each encoding one $2^n$-bit value.
    
Below we show that already \mug queries are \PSPACE-hard.
% and group can be used to encode alternations.

\begin{theorem}\label{lem:mug-pspace-hard}
  \mug is \PSPACE-hard in query complexity.
\end{theorem}
\begin{proof}
  Proof by reduction from the validity problem of QBF.  Let $\varphi$ be a
  quantified Boolean formula over the variables $x_1, \dots, x_n$ of the form
  $\mathsf{Q}_1 x_1 \mathsf{Q}_2 x_2 \dots \mathsf{Q}_n x_n.\psi$, for
  $\mathsf{Q}_i \in \{\exists,\forall\}$. We construct a forest $F$ and an
  \mug query $\q$ such that $F \pipeline \q$ is non-empty iff $\varphi$ is
  valid.

  $F$ contains a single tree $d$ of the form \object{x:\,[0,1]},
  and $\q$ is as follows:\\[1mm]
  \renewcommand{\id}{\texttt{\_\!\!\;id}\xspace}
  \mbox{\qquad}$\begin{array}{l}
    \group{\text{x1}/\text{x},~\dots,~\text{x}n/\text{x}}{}\pipeline
     \unwind{\id.\text{x1}} \pipeline\dots\pipeline \unwind{\id.\text{x}n}
    \pipeline \match{\psi'}\pipeline{}
    \\[-1mm]
    \group{\text{x1}/\text{\id.x1},~\dots,~\text{x}(n-1)/\text{\id.x}(n-1)}{\text{val}/\text{x}n}\pipeline
    \match{\textit{qua}_n(\text{val})}\pipeline{}
    \\[-1mm]
    %&\group{\text{x1}/\text{\id.x1},~\dots,~\text{x}(n-2)/\text{\id.x}(n-2)}{\text{val}/\text{x}(n-1)}\pipeline
    % \match{\textit{qua}_{n-1}(\text{val})}\pipeline{}
    % \\
     \cdots \\[-1mm]
    \group{\text{x1}/\text{\id.x1}}{\text{val}/\text{x}2}\pipeline{}
    \match{\textit{qua}_2(\text{val})}\pipeline{}
    \\[-1mm]
    \group{}{\text{val}/\text{x}1}\pipeline{}
    \match{\textit{qua}_1(\text{val})}
  \end{array}$\\[1mm]
  \noindent where $\psi'$ is the criterion with occurrences of a variable $x_i$
  in $\psi$ encoded by the path \id.x$i$, $\textit{qua}_i(\text{val})$ is the
  expression $(\text{val}=[0,1])$ if $\mathsf{Q}_i$ is $\forall$, and the
  expression $(\text{val} \neq [])$ if $\mathsf{Q}_i$ is $\exists$.

  The query $\q$ consists of two logical parts. In the first one we create $n$
  arrays [0,1], unwind each of them, thus creating all possible $2^n$ variable
  assignments and then filter only the satisfying ones. In the second part, for
  each quantifier $\mathsf{Q}_ix_i$, we filter the assignments to the variables
  $x_1,\dots,x_{i-1}$ satisfying the formula
  $\mathsf{Q}_ix_i\ldots \mathsf{Q}_nx_n.\psi$ by using group.
\end{proof}

% Next, we show that evaluation of \mp queries with additional array operators
% \emph{filter} and \emph{map}, which allow for filtering out and for
% transformation of the elements inside an array, respectively, is \NP-hard
% already in query complexity.

% \begin{lemma}\label{lem:mp-np-hard-with-filter-and-map}
%   Boolean query evaluation for \mp queries with filter and map operators is
%   NP-hard in query complexity.
% \end{lemma}

%%% Local Variables:
%%% mode: latex
%%% TeX-master: "icdt18"
%%% fill-column: 79
%%% End:

\section{Conclusions and Future Work}
\label{sec:conclusions}

We carried out a first formal investigation of \mongodb, a widely used NoSQL
database system, for which different vendors now claim compatibility.
% with the aim of understanding its query expressivity and complexity.
We provided a formalization of the \mongodb data model, and of \mquery, a core
fragment of the \mongodb query language.  We studied the expressivity of
\mquery, showing the equivalence between its well-typed fragment (i.e., queries
composed of well-typed stages) and NRA, by developing compact translations in
both directions.  We further investigated the computational complexity of
significant fragments of \mquery, obtaining several (tight) bounds in combined
complexity.
% , which range from \LOGSPACE to alternating exponential-time with a
% polynomial number of alternations.
As a byproduct, we have also established a tight complexity bound for NRA.

We have carried out our investigation on a real-world data model and query
language, of a widely adopted document database still lacking a proper
formalization, as opposed to studying a possibly abstract formalism not derived
from a real system, as done in recent proposals
\cite{hidders17jlogic,bourhis17json}.  Our work provides a better understanding
of the semantic and computational properties of \mongodb, hence, we believe
that
% it is of great practical value and
it will have a strong impact, both on the design and implementation and
on the usage of the system, since \mongodb is still under active development
% some of the insights provided by our work might help the
% developers in tuning the system, and possibly backtrack on some choices that
% appear difficult to justify from a formal point of view
(cf.\ the discussion in Appendix~\ref{sec:mquery-vs-mongodb}).
% \item Hence, a thorough comparison with these languages would require a full
%   investigation of their formal and computational properties, not much
%   different in scope than what is provided here, and left for future work.
% \item \nb{Clean up and expand.}The \mquery aggregation framework is now being
%   adopted by other systems as well \cite{systems-adopting-mquery}, and
%   therefore our effort of providing a clean formalization and corresponding
%   semantics is even more important.

Our work still leaves several interesting theoretical questions for
investigation.  Various complexity results are still open, including the
precise complexity of \mug, and an understanding of \mquery under bag
semantics, which is the one adopted by \mongodb, and in the presence of lists
that represent arrays.  We have not addressed the problem of when an \mquery is
translatable to NRA even when not every stage is well typed, which could lead
to a more relaxed notion of well-typedeness.  With the latest v3.4, \mongodb
has been extended with a \emph{graph-lookup} stage in a pipeline, allowing for
a recursive search on a collection, and it is of interest to understand how
extending \mquery with this feature affects its formal and computational
properties.

% In our formalization, we have abstracted away some features present in the
% \mongodb system, such as bag semantics and document order (which in \mongodb
% is xnot even well characterized).  We plan to extend our results on
% expressivity and complexity, to take into account such additional features.

We are currently working on applying the results presented here, to provide
high-level access to \mongodb data sources by relying on the standard
ontology-based data access (OBDA) paradigm \cite{PLCD*08}.  Specifically, we
are connecting the intermediate layer of an OBDA system \cite{CCKK*17}, which
generates relational queries, with a \mongodb backend, and we exploit for this
the translation from NRA to \mquery \cite{BCCRX16}.  In this context, it is of
particular importance to generate queries that can be efficiently executed by
the backend, and hence optimize the translation techniques, so as to ensure
scalability for complex queries over large collections.

\iftechreport
\else
\clearpage
\fi

\medskip
\noindent
\textbf{Acknowledgements.}  ~ We thank Christoph Koch and Dan Suciu for helpful
clarifications on nested relational algebra, and Henrik Ingo for information
about \mongodb.  We also thank Martin Rezk for his participation to initial
work on the topic of the paper.

\bibliographystyle{plain}
\bibliography{main-bib}
%\bibliography{string-medium,references}

\clearpage
\appendix
\newenvironment{theoremnum}[1]{\par\smallskip\noindent\textbf{\textcolor{darkgray}{\ensuremath{\blacktriangleright}}~\textsf{Theorem~#1}.}\hspace*{0.3em}\em}{\smallskip}

\newenvironment{lemmanum}[1]{\par\smallskip\noindent\textbf{\textcolor{darkgray}{\ensuremath{\blacktriangleright}}~\textsf{Lemma~#1}.}\hspace*{0.3em}\em}{\smallskip}

\newenvironment{corollarynum}[1]{\par\smallskip\noindent\textbf{\textcolor{darkgray}{\ensuremath{\blacktriangleright}}~\textsf{Corollary~#1}.}\hspace*{0.3em}\em}{\smallskip}

\section{Semantics of Nested Relational Algebra}
\label{sec:semantics-nra}

\subsection{Extended projection}

Let $R$ be a relation schema. We define extended projection $\pi_P(R)$ in
detail, where $P$ is a set containing attributes of $R$ and elements of the
form $b/e$, where $b$ is a fresh attribute name and $e$ is an expression
defined according to the grammar:
\def\subrel{\mathsf{subrel}}
\def\evalu{\mathsf{eval}}
\def\tup{\mathsf{tup}}
$$
\begin{array}{r@{~}c@{~}l}
  e &\DEF& a \mid c \mid f \mid \cond{f}{e}{e} \mid \subrel(t,\dots,t)
  \\
  f &\DEF& \truevalue
  \mid \falsevalue
  \mid a = a
  \mid a = c
  \mid \neg f
  \mid f \land f
  \mid f \lor f
  %\mid \cond{f}{f}{f}
  \\
  t &\DEF& \{b{:}e,\dots,b{:}e\}\\
\end{array}
$$
Here, $a \in \att(R)$, $c$ is a constant atomic value, $f$ is an expression
that evaluates to a Boolean value, $b$ is a fresh attribute name, $t$ is a
tuple definition, and $\subrel(t_1,\dots,t_n)$ is a relation definition, which
constructs a relation from the tuples $t_1,\dots,t_n$, where all $t_i$ are
required to be of the same schema.

Let $r$ be an $R$-tuple. We define the evaluation $\evalu(e, r)$ of $e$ over
$r$ inductively as follows:
\begin{itemize}
\item $\evalu(a, r) = v$ where $a{:}v \in r$.
\item $\evalu(c, r) = c$.
\item $\evalu((e_1=e_2), r)$ is \truevalue if $\evalu(e_1, r)=\evalu(e_2,r)$, and
  \falsevalue otherwise.
\item $\evalu((e_1=c), r)$ is \truevalue if $\evalu(e_1, r)=c$, and
  \falsevalue otherwise.
\item $\evalu((\neg f), r) = \neg \evalu(f, r)$.
\item $\evalu((f_1 \land f_2), r) = \evalu(f_1, r) \land \evalu(f_2,r)$.
\item $\evalu((f_1 \lor f_2), r) = \evalu(f_1, r) \lor \evalu(f_2,r)$.
\item $\evalu(\cond{f}{e_1}{e_2}, r)$ is $\evalu(e_1,r)$ if $\evalu(f,r) =
  \truevalue$, and $\evalu(e_2,r)$ otherwise.
\item $\evalu(\subrel(t_1,\dots,t_n), r) = \{\tup(t_1,r), \dots, \tup(t_n,r)\}$
  where $\tup(\{b_1{:}e_1,\dots,b_n{:}e_n\}, r) =
  \{b_1{:}\evalu(e_1,r),\dots,b_n{:}\evalu(e_n,r)\}$.
\end{itemize}

Then, given a relation instance $\R$ of schema $R$,
$\pi_{a_1,\dots,a_n,b_1/e_1,\dots,b_m/e_m}(\R)$ is the relation
$$\begin{array}{l}
  \Big\{
  \{a_{1}{:}\evalu(a_1,r),\ldots,a_n{:}\evalu(a_n,r), b_1{:}\evalu(e_1,r),\dots,b_m{:}\evalu(e_m,r) \} \mid r \in \R
  \Big\}.
\end{array}$$

We observe that the result of extended projection can be computed in \LOGSPACE.\nb{for Diego: Check!}

\subsection{Nest}

The \emph{nest} operator $\nest_{\{a_1, \ldots, a_n\}\rightarrow b}(R)$ results
in a schema with attributes $(\att(R)\setminus\{a_1,\ldots,a_n\}) \cup
\{b(a_1,\ldots,a_n)\}$.  Let $\R$ be a relation instance of schema $R(\{ a_1,
\ldots, a_m \})$ and $n\leq m$.  Then $\nest_{\{a_1, \ldots, a_n\}\rightarrow
  b}(\R)$ is the relation
$$\begin{array}{l}
  \Big\{
  \{a_{n+1}{:}v_{n+1},\ldots,a_m{:}v_m, b{:}\big(\pi_{a_1, \ldots, a_n}
  (\sigma_{a_{n+1}=v_{n+1},\ldots,a_m=v_m} (\R))\big) \} \mid{} \\
  \hspace{7cm} \{ a_{n+1}{:}v_{n+1},\ldots, a_m{:}v_m \} \in \pi_{a_{n+1},\ldots,a_m}(\R)
  \Big\}.
\end{array}$$

\subsection{Unnest}

The \emph{unnest} operator $\unnest_a(R)$ results in a schema with attributes
$(\att(R)\setminus\{a\}) \cup \att(a)$.  Let $\R$ be a relation instance of
schema $R(\{ a_1, \ldots, a_n, a \})$.  Then $\unnest_{a}(\R)$ is the relation
$$\Big\{%
\{a_1{:}v_1,\ldots,a_{n}{:}v_{n}, b_1{:}u_1,\ldots,b_k{:}u_k\} \mid \{
a_1{:}v_1,\ldots,a_n{:}v_n, a{:}v \} \in \R, \{ b_1{:}u_1,\ldots,b_k{:}u_k \}
\in v%
\Big\}.$$

\section{Examples of \mongodb Queries}
\label{sec:mongodb-queries-examples}

\mongodb provides two main query mechanisms.  The basic form of query is a
\emph{find} query, which allows one to filter out documents according to some
(Boolean) criteria and to return, for each document passing the filter, a tree
containing a subset of the key-value pairs in the document.  Specifically, a
find query has two components, where the first one is a \emph{criterion} for
selecting documents, and the second one is a \emph{projection condition}.

\begin{example}
  The following \mongodb find query selects from the \texttt{bios} collection
  the documents talking about scientists whose first name is Kristen, and for
  each document only returns the full name and the date of birth.
\begin{lstlisting}
  db.bios.find(
      {"name.first": {$eq: "Kristen"}},
      {"name": true, "birth": true}
  )
\end{lstlisting}
% $
\noindent%
When applied to the document in Figure~\ref{fig:mongodb-document}, it returns the
following tree:
\begin{lstlisting}
  {   "_id": 4,
      "birth": "1926-08-27",
      "name": { "first": "Kristen", "last": "Nygaard" }
  }
\end{lstlisting}
Observe that by default the document identifier is included in the answer of
the query.
\qedempty
\end{example}
Note that with a find query we can either obtain the original documents as they
are, or we can modify them by specifying in the projection condition only a
subset of the keys, thus retaining in the answer only the corresponding
key-value pairs.  However, we cannot change the shape of the individual pairs.

A more powerful querying mechanism is provided by the \emph{aggregation
  framework}, in which a query consists of a pipeline of \emph{stages}, each
transforming a forest into a new forest.  We call this transformation pipeline
an \emph{\mquery}.  One of the main differences with find queries is that
\mquery can manipulate the shape of the trees.

\begin{example}
  The following \mquery essentially does the same as the
  previous find query, but now it flattens the complex object \texttt{name}
  into two key-value pairs: \\
$ \bios \pipeline \match{\text{name.first=``Kristen''}} \pipeline
\project{\text{birth, firstName}/\text{name.first, lastName}/\text{name.last}}$
\begin{lstlisting}
db.bios.aggregate([
    {$match: {"name.first": {$eq: "Kristen"}}},
    {$project: {
        "birth": true, "firstName": "$name.first", "lastName": "$name.last" } }
])
\end{lstlisting}
% $
So the document from our running example will be transformed into the following
tree:
\begin{lstlisting}
{   "_id" : 4,
    "birth": "1926-08-27",
    "firstName": "Kristen",
    "lastName": "Nygaard"
}
\end{lstlisting}
\qedempty
\end{example}

We note that the unwind operator creates a new document for every element in
the array. Thus, unwinding \valuefont{awards} (once) in the document in our
running example will output 3 documents, only one of which satisfies the
subsequent selection stages.
In the example below we illustrate how the match operator interacts with arrays.
\begin{example}
  Consider the following query consisting of a match stage with two conditions
  on keys inside the \valuefont{awards} array: \\
$\bios \pipeline \match{\text{awards.year}=1999 \wedge \text{awards.award}=\text{``Turing Award''}}$
\begin{lstlisting}
db.bios.aggregate([
    {$match: {"awards.year": {$eq: 1999},
              "awards.award": {$eq: "Turing Award"} }}
])
\end{lstlisting}
% $
The query returns all the persons that have received an award in 1999, and the
Turing award in a possibly different year. Observe that it does not impose that
one array element must satisfy all the conditions.  This query retrieves the
document of our running example because Kristen Nygaard received an award (the
Rosing Prize) in 1999 in addition to the Turing Award (in 2001).
\qedempty
\end{example}

\begin{example}
To ensure that the two previous conditions are satisfied by the same array
element, the standard solution in the \mquery fragment consists in flattening
the \valuefont{awards} array before the match stage, as follows: \\
$\bios \pipeline \unwind{\text{awards}} \pipeline \match{\text{awards.year}=1999
  \wedge \text{awards.award}=\text{``Turing Award''}}$
\begin{lstlisting}
db.bios.aggregate([
   {$unwind: "$awards"},
   {$match: {"awards.year": {$eq: 1999},
             "awards.award": {$eq: "Turing Award"} }}
])
\end{lstlisting}
% $

An alternative solution to flattening is merging the two conditions into an
object equality:\\
$\bios \pipeline \match{\text{awards}=\{\text{"year": }1999,
  \text{ "award": "Turing Award"} \}} $
\begin{lstlisting}
db.bios.aggregate([
   {$project: {"awards.award": true,
               "awards.year": true }}, 
   {$match: {"awards": {$eq: {"award" : "Turing Award",
                              "year" : 1999}} }}
])
\end{lstlisting}
% $
Note that the object equality requires to remove non-compared keys from the
array elements before the match stage.  These two queries return no result in
our running example.
\qedempty
\end{example}

In the following examples, we illustrate two cases of non well-typed project stages.
\begin{example}
  Consider the following query which creates the array \valuefont{fields} out
  of an object, a literal and an array: \\
$\bios \pipeline \project{\text{fields/}[\text{name, birth, awards}]}$
\begin{lstlisting}
db.bios.aggregate([
    {$project: {"fields": ["$name", "$birth", "$awards"] }}, 
])
\end{lstlisting}
When applied to our running example, the resulting document contains a
non well-typed array.
\begin{lstlisting}
{ "_id" : 4,
  "fields" : [ 
      {"first": "Kristen", "last": "Nygaard"}, 
      "1926-08-27", 
      [ {"award": "Rosing Prize", "year": 1999, "by": "Norwegian Data Association"}, 
        {"award": "Turing Award", "year": 2001, "by": "ACM"}, 
        {"award": "IEEE John von Neumann Medal", "year": 2001, "by": "IEEE"} ] ]
}
\end{lstlisting}
\qedempty
\end{example}

In the remaining of this section, we introduce a second document in the
\valuefont{bios} collection, as depicted by Figure \ref{fig:second-document}.
\begin{figure}[t]
  \centering
\begin{lstlisting}
{ "_id": 6,
  "awards": [
    { "award": "Award for the Advancement of Free Software", "year": 2001, "by": "FSF" },
    { "award": "NLUUG Award", "year": 2003, "by": "NLUUG" } ],
  "birth": "1956-01-31",
  "contribs": [ "Python" ],
  "name": { "first": "Guido", "last": "van Rossum" } }
\end{lstlisting}
  \caption{Second \mongodb document in the \texttt{bios} collection}
  \label{fig:second-document}
\end{figure}

\begin{example}
In the query below, the project operator assigns an array or an object
to the \valuefont{value} field according to the document id:\\
$\bios \pipeline \project{\text{value/}\cond{\text{\_id}=4}{\text{awards}}{\text{name}}}$
\begin{lstlisting}
db.bios.aggregate([
    {$project: {"value": {$cond: {if: {$eq: ["$_id", 4]},
                                  then: "$awards",
                                  else: "$name"}} }}, 
])  
\end{lstlisting}
When applied to the two documents of the \valuefont{bios} collection, it
produces a non well-typed forest.
\begin{lstlisting}
{   "_id": 4,
    "value": [ 
        { "award": "Rosing Prize", "year": 1999, "by": "Norwegian Data Association" },
        { "award": "Turing Award", "year": 2001, "by": "ACM" }, 
        { "award": "IEEE John von Neumann Medal", "year": 2001, "by": "IEEE" } ]
}

{  "_id": 6,
    "value": {"first" : "Guido", "last" : "van Rossum" }
}  
\end{lstlisting}
\qedempty
\end{example}
Finally, in the examples below we illustrate the group stage, which combines different
documents into one.
\begin{example}
  The following query returns for each year all scientists that received an
  award in that year: \\
$\bios \pipeline \unwind{\text{awards}} \pipeline \group{\text{year/awards.year}}{\text{names/name}} $
\begin{lstlisting}
db.bios.aggregate([
    {$unwind: "$awards"},
    {$group: {
        _id: {"year": "$awards.year"}, "names": {$addToSet: "$name"} }},
])
\end{lstlisting}
Running this query over the \valuefont{bios} collection produces the following output:
\begin{lstlisting}
{   "_id": { "year": 2003 },
    "names": [
        { "first": "Guido", "last": "van Rossum" } ]
},
{   "_id": { "year": 2001 },
    "names": [
        { "first": "Kristen", "last": "Nygaard" },
        { "first": "Guido", "last": "van Rossum" } ]
},
{   "_id": { "year": 1999 },
    "names": [
        { "first": "Kristen", "last": "Nygaard" } ]
}
\end{lstlisting}
\qedempty
\end{example}

\begin{example}
 Consider the following query which groups the persons according to their date
 of death:\\
$\bios \pipeline \group{\text{death}}{\text{names/name}} $
  \begin{lstlisting}
db.bios.aggregate([
    {$group: { 
        "_id": "$death",
        "names": {$addToSet: "$name"} }}
])
\end{lstlisting}
When executing over the \valuefont{bios} collection, it produces the following output:
  \begin{lstlisting}
{   "_id": "2002-08-10",
    "names" : [ 
        { "first": "Kristen", "last": "Nygaard" } ]
}

{   "_id": null,
    "names": [ 
        { "first": "Guido", "last": "van Rossum" } ]
}    
  \end{lstlisting}
Since the \valuefont{death} path is not present in the document about Guido van
Rossum, the latter is grouped in the document where \valuefont{\_id} is \nullvalue.
\qedempty
\end{example}

\begin{example} 
The query below considers two grouping paths: \valuefont{death} and
\valuefont{citizenship}. Note that the latter path is absent in the
\valuefont{bios} collection.
\\
$\bios \pipeline \group{\text{death,\,citizenship}}{\text{names/name}} $
  \begin{lstlisting}
db.bios.aggregate([
    {$group: { 
        "_id": {"death": "$death", "citizenship": "$citizenship"},
        "names": {$addToSet: "$name"} }}
])
  \end{lstlisting}
% $
Executing this query over the \valuefont{bios} collection produces the
following result:
  \begin{lstlisting}
{   "_id": { "death": "2002-08-10" },
    "names" : [ 
        { "first": "Kristen", "last": "Nygaard" } ]
}

{   "_id": {},
    "names" : [ 
        { "first": "Guido", "last": "van Rossum" } ]
}
\end{lstlisting}
In this case, missing grouping paths do not appear in the resulting
  documents. Consequently, the entry about Guido van Rossum is
  grouped in the document where \valuefont{\_id} equals \valuefont{\{\}}.
\qedempty
\end{example}

We conclude with a complex query performing a join within a document. 
\begin{example}
  Consider the \mquery \\
$
\begin{array}{@{}l}
 \bios \pipeline \project{\text{name, award1/awards, award2/awards}}
         \pipeline \unwind{\text{award1}} \pipeline \unwind{\text{award2}} \pipeline \\
         \project{\text{name, award1, award2,
  twoInOneYear/(award1.year=award2.year}\,\wedge\,\text{award1.award}\neq
  \text{award2.award})} \pipeline \\
\match{\text{twoInOneYear=}\truevalue} \pipeline \\
  \project{\text{firstName/name.first,\,lastName/name.last,\,awardName1/award1.award,\,awardName2/award2.award,\,year/award1.year}}
\end{array}
$
\begin{lstlisting}
db.bios.aggregate([
  {$project: { "name": true,
    "award1": "$awards", "award2": "$awards" } },
  {$unwind: "$award1"},
  {$unwind: "$award2"},
  {$project: {
    "name": true, "award1": true, "award2": true,
    "twoInOneYear": { $and: [
      {$eq: ["$award1.year", "$award2.year"]},
      {$ne: ["$award1.award", "$award2.award"]} ]} }},
  {$match: { "twoInOneYear": true } },
  {$project: { "firstName": "$name.first",
    "lastName": "$name.last" ,
    "awardName1": "$award1.award",
    "awardName2": "$award2.award",
    "year": "$award1.year" } },
])
\end{lstlisting}%

It consists of 6 stages and retrieves all persons who received two awards in
one year.
The first stage keeps the complex object \valuefont{name}, creates two copies
of the array \valuefont{awards}, and projects away all other paths.  The
second and third stages flatten (unwind) the two copies (\valuefont{award1} and
\valuefont{award2}) of the array of awards (which intuitively creates a
cross-product).  The fourth step compares awards pairwise and creates a new key
(\valuefont{twoInOneYear}) whose value is true if the scientist has two awards
in one year.  The fifth one selects the documents of interest (those where
\valuefont{twoInOneYear} is true), and the final stage renames the selected keys.

By applying the query to the document in Figure~\ref{fig:mongodb-document}, we
obtain:
\begin{lstlisting}
{   "_id": 4,
    "firstName": "Kristen",
    "lastName": "Nygaard",
    "awardName1": "IEEE John von Neumann Medal",
    "awardName2": "Turing Award",
    "year": 2001
}
\end{lstlisting}
\qedempty
\end{example}

\section{Syntax and Semantics of \mquery}
\label{sec:mquery-vs-mongodb}

\begin{figure*}[htb]
  \centering
  \scalebox{0.8}{%\bf
    \begin{tabular}{l@{\quad}l}
      \multicolumn{2}{@{}c}{\begin{tabular}{@{}r@{~}c@{~}l@{}}
          %\meta{MFQ} &\DEF& \meta{Collection}.find(\{\meta{Criterion}\},  \{ \meta{List<Path}:\:\truevalue\meta{>} \}) \quad
          \meta{MAQ} &\DEF& \meta{Collection}\vf{.aggregate([} \meta{List$^+$<}\meta{Stage}\meta{>} \vf{])}
          \\
          \meta{Stage} &\DEF&  \vf{\{\$match:~\{}\meta{Criterion}\vf{\}\}}
          ~\CHOICE~ \vf{\{\$unwind:~\{}\meta{UnwindExpr}\vf{\}\}}
          ~\CHOICE~ \vf{\{\$project:~\{}\meta{Projection}\vf{\}\}}\\
          &\CHOICE& \vf{\{\$group:~\{}\meta{GroupExpr}\vf{\}\}}
          ~\CHOICE~ \vf{\{\$lookup:~\{}\meta{LookupExpr}\vf{\}\}}\\
        \end{tabular}}\\
      \hline
      \begin{tabular}{@{}r@{~}c@{~}l@{}}
        \meta{Path} &\DEF& \meta{Key} \CHOICE~ \meta{Key}.\meta{Path}
        \\
        \meta{PathRef} &\DEF& \vf{\$}\meta{Path} ~\CHOICE~ \vf{\$\$ROOT}\\
        \meta{Lop} &\DEF& \vf{\$and} ~\CHOICE~ \vf{\$or} ~\CHOICE~ \vf{\$nor}\\
        \meta{Cop} &\DEF& \vf{\$eq} ~\CHOICE~ \vf{\$gt} ~\CHOICE~ \vf{\$lt}\\
        &\CHOICE& \vf{\$ne} ~\CHOICE~ \vf{\$gte} ~\CHOICE~ \vf{\$lte} \\
        \meta{Bop} &\DEF& \meta{Lop} \CHOICE~ \meta{Cop} \\
        \meta{Boolean} &\DEF& \truevalue ~\CHOICE~ \falsevalue\\
        \meta{Criterion} &\DEF& \meta{Path}\vf{:} \meta{Condition} \\
        &\CHOICE& \meta{Lop}\vf{:} \vf{[} \meta{List$^+$<}\vf{\{}\meta{Criterion}\vf{\}}\meta{>} \vf{]}
        \\
        \meta{Condition} &\DEF& \vf{\{}\meta{Cop}\vf{:} \meta{Value}\vf{\}} \\
        &\CHOICE& \vf{\{\$not:} \meta{Condition}\vf{\}}\\
        &\CHOICE& \vf{\{\$exists:} \meta{Boolean}\vf{\}}\\
        \meta{Projection} &\DEF& \meta{List$^+$<ProjectionElem>} \\
        \meta{ProjectionElem} &\DEF& \id\vf{:} \falsevalue\\
        &\CHOICE& \meta{Path}\vf{:} \truevalue \\
        &\CHOICE&\meta{Path}\vf{:} \meta{ValueDef}
        \\
      \end{tabular}
      &
      \begin{tabular}{@{}r@{~}c@{~}l@{}}
        \meta{ValueDef} &\DEF& \meta{PathRef} \\
        &\CHOICE& \vf{\{\$literal:} \meta{Value}\vf{\}}\\
        &\CHOICE& \vf{[}\meta{List<ValueDef>}\vf{]}\\
        &\CHOICE& \vf{\{}\meta{Bop}\vf{:} \vf{[}\meta{List<ValueDef>}\vf{]\}}\\
        &\CHOICE& \vf{\{\$not:} \meta{ValueDef}\vf{\}}
        \\
        &\CHOICE& \vf{\{\$cond:} \vf{\{}
        \begin{tabular}[t]{@{}l}
          \vf{if:} \meta{ValueDef}\vf{,}\\
          \vf{then:} \meta{ValueDef}\vf{,} \\
          \vf{else:} \meta{ValueDef} \vf{\}\}}
        \end{tabular}
        \\
        \meta{GroupExpr} &\DEF& \id\vf{:} \meta{GroupCondition}\vf{,} \\
        &&\meta{List<\meta{Key}}\vf{:} \vf{\{\$addToSet:} \meta{\meta{PathRef}}\vf{\}}\meta{>}\\
        \meta{GroupCondition} &\DEF& \nullvalue
          ~~\CHOICE~ \vf{\{}\meta{List<Path}\vf{:} \meta{PathRef>}\vf{\}}\\
        \meta{UnwindExpr} &\DEF& \vf{path:} \meta{PathRef}\vf{,} \\
        &&\vf{\footnotesize preserveNullAndEmptyArrays:} \meta{Boolean}
        \\
        \meta{LookupExpr} &\DEF& \vf{from:} \meta{Collection}\vf{,} \\
        && \vf{localField:} \meta{Path}\vf{,} \\
        && \vf{foreignField:} \meta{Path}\vf{,} \\
        && \vf{as:} \meta{Path}\\
      \end{tabular}
    \end{tabular}
  }
  \caption{The \mquery grammar}
  \label{fig:mongodb-query-syntax}
\end{figure*}

We provide the actual syntax of \mquery in
Figure~\ref{fig:mongodb-query-syntax}.
A \meta{Path} (which in \mongodb terminology is actually called a ``field''), is
a non-empty concatenation of \meta{Key}s, where elements for \meta{Key} are from
the set~$K$.  Elements for \meta{Value} are defined according to the grammar in
Figure~\ref{fig:syntax-bson}. \meta{Collection} is a collection name, that is, a
non-empty string.
The empty path, which can be used in a path reference, is denoted in \mongodb by
the string \vf{\$\$ROOT}.  In the following, a \emph{path} is either the
empty path or an element constructed according to \meta{Path}.
We assume that a projection $p_1{:}d_1,\dots,p_n{:}d_n$ is such that there are
no $i\neq j$ where $p_i$ is a prefix of $p_j$.  By default the \id key is kept
in a projection, and to project it away the projection must contain an element
\id\vf{:} \falsevalue.
The comparison operators used in a value definition \meta{ValueDef} accept only
arrays of length~2.
We observe that, with respect to the official \mongodb syntax, we have
removed/introduced some syntactic sugar.  In particular, for \meta{Criterion} we
disallow expressions of the form \vf{"name.first":\,"john"}.  Instead
they can be expressed as \vf{"name.first":\,\{\$eq:\,"john"\}}.
Moreover, we allow for the use of \vf{\$nor} in \meta{ValueDef}, as it
can be expressed using \vf{\$not} and \vf{\$and}.

\subsection{Notes on our \mquery algebra}
\begin{itemize}
\item The grouping condition \nullvalue in the grammar is given by the empty
  sequence $G$ in the algebra.
\item \mongodb can interpret any value definition as a Boolean expression, in
  particular, one can use $p$, $v$, and $[d_1,\ldots,d_n]$ as atomic Boolean
  value definitions.  Specifically, $t\models p$ for a path $p$, and
  $t\models v$ for a value $v$, hold whenever $v$ (resp., the ``value'' of $p$
  in $t$) is not \nullvalue, \falsevalue, or 0, while
  $t\models[d_1,\ldots,d_n]$ always holds.  In our algebra instead, we consider
  as atomic Boolean value definitions only $p=p$, $p=v$, and $\exists p$.
  % To make sure that the atomic expressions in such value definitions are of
  % this form, we replace atomic expressions $p$, for a path $p$, with
  % $\exists p \land (p\neq\nullvalue) \land (p\neq\falsevalue) \land (p\neq0)$,
  % and get rid of the ``constant'' expressions of the form $v$ for a value $v$,
  % and $[d_1,\dots,d_k]$ as follows. First, replace the former with \truevalue
  % when $v\neq \nullvalue$, $v\neq\falsevalue$ and $v\neq 0$, and with
  % \falsevalue otherwise, and replace the latter by \truevalue. Second, simplify
  % the Boolean combinations with \truevalue and \falsevalue in the standard way
  % (e.g., $\truevalue \lor e$ becomes \truevalue, and $\truevalue \land e$
  % becomes $e$), and simplify conditional value definitions that have \truevalue
  % or \falsevalue as the condition in the obvious way.

\item We observe that in the \mongodb grammar there is no explicit operator to
  check the existence of a path in a Boolean value definition. Nevertheless, we
  included $\exists p$ as an atomic Boolean value definition $\beta$ in our
  algebra since it can be expressed using a conditional value definition as
  follows:
  \[\cond{p}{\truevalue}{\cond{(\neg(p=\nullvalue)\land \neg(p=\falsevalue)
      \land \neg(p=0))}{\truevalue}{\falsevalue}}.\]
\item For simplicity, the only comparison operator that we kept in the algebra
  is equality.  Adding also order comparison would not affect any of the
  results on expressivity and complexity presented in the paper.
\end{itemize}

\subsection{Semantics: Tree operations}
\label{sec:tree-operations}

In the following, let $t=(N,E,\lnode,\ledge)$ be a tree.
Below, when we mention reachability, we mean reachability along the edge
relation.
\begin{description}
\item[subtree] \emph{the subtree of $t$ rooted at $x$ and induced by $M$}, for
  $n \in M$ and $M \subseteq N$, denoted $\subtree(t,x,M)$, is defined as
  $(N',E|_{N' \times N'},\lnode|_{N'},\ledge|_{E'})$ where $N'$ is the subset of
  nodes in $M$ reachable from $x$ through nodes in $M$.
  We write $\subtree(t,M)$ as abbreviation for $\subtree(t,\treeroot(t),M)$.

  \smallskip%
  For a path $p$ with $|\eval{p}|=1$, \emph{the subtree $\subtree(t,p)$ of $t$
    hanging from $p$} is defined as $\subtree(t,r_p,N')$ where $\{r_p\} =
  \eval{p}$, and $N'$ are the nodes reachable from $r_p$ via~$E$.
  For a path $p$ with $|\eval{p}|=0$, $\subtree(t,p)$ is defined as
  $\tree(\nullvalue)$.

\item[attach] The \emph{tree $\attach(k_1\ldots k_n, t)$ constructed by inserting the
    path $k_1 \ldots k_n$ on top of the tree $t$}, for $n \geq 1$, is defined as
  $(N',E',\lnode',\ledge')$, where
  \begin{itemize}
  \item $N' = N \cup \{x_0, x_1,\dots,x_{n-1}\} $, for fresh
    $x_0,\dots,x_{n-1}$.
  \item $E' = E \cup \{(x_0,x_1),(x_1,x_2),\dots,(x_{n-1},\treeroot(t))\}$,
  \item $\lnode' = \lnode \cup \{(x_0,\objectlabel),\dots,(x_{n-1},\objectlabel)\}$,
  \item $\ledge' = \ledge \cup \{((x_0,x_1),k_1), \dots, ((x_{n-2},x_{n-1}),k_{n-1}),
    ((x_{n-1},\treeroot(t)),k_n)\}$.
  \end{itemize}

\item[intersection] Let $t_1$ and $t_2$ be trees.  The function
  $\commonpath{t_1}{t_2}$ returns the set of pairs of nodes $(x_n,y_n) \in N^1
  \times N^2$ reachable along identical paths in $t_1$ and $t_2$, that is, such
  that there exist $(x_0,x_1), \dots, (x_{n-1},x_n)$ in $E^1$, for
  $x_0=\treeroot(t_1)$, and $(y_0,y_1), \dots, (y_{n-1},y_n)$ in $E^2$, for
  $y_0=\treeroot(t_2)$, with $\lnode^1(x_i) = \lnode^2(y_i)$ and
  $\ledge^1(x_{i-1},x_i) = \ledge^2(y_{i-1},y_i)$, for $1\leq i\leq n$.

\item[merge] Let $t_1, t_2$ be trees $(N^j,E^j,\lnode^j,\ledge^j)$, $j=1,2$, such
  that $N^1 \cap N^2 = \emptyset$, and for each path $p$ leading to a leaf in
  $t_2$, i.e., $t_2 \models (p = v)$ for some literal value $v$, we have that
  $t_1 \not\models \exists p$ and the other way around. Then the \emph{tree
    $t_1 \oplus t_2$ resulting from merging $t_1$ and~$t_2$} is defined as
  $(N,E,\lnode,\ledge)$, where
  \begin{itemize}
  \item $N = N^1 \cup {N^2}'$, for ${N^2}' = N^2\setminus \{x_2 \mid (x_1,x_2)
    \in \commonpath{t_1}{t_2}\}$
  \item $E=E^1 \cup (E^2 \cap ({N^2}' \times {N^2}')) \cup
    ((\commonpath{t_1}{t_2}) \circ E^2)$
  % \item ${\prec} = {\prec}^1 \cup {\prec}^2 \cup \{(x_1,x_2) \mid x_1 \in N^1,
  %   x_2 \in {N^2}', ~x_1,x_2\text{ are siblings in } (N,E)\}$
  \item $\lnode = \lnode^1 \cup \lnode^2|_{{N^2}'}$
  \item $\ledge = \ledge^1 \cup \ledge^2|_{{N^2}' \times {N^2}'} \cup %
    \{((x_1,y_2),\ell) \mid \ledge^2(y_1,y_2)=\ell, (x_1,y_1) \in
    \commonpath{t_1}{t_2}\}$
  \end{itemize}

\item[minus] $t_1 \setminus t_2$ is $\subtree(t_1,N')$ where $N' = N_1
  \setminus N_2$.

\item[array] Let $\{t_1, \dots, t_n\}$, $n\geq 0$, be a forest and $p$ a
  path. The operator $\colltoarray(\{t_1, \dots,t_n\}, p)$ creates the tree
  encoding the array of the values of the path $p$ in the trees $t_1, \dots,
  t_n$. Let $t_j^p = \subtree(t_j, p)$ with $(N^j,E^j,\lnode^j,\ledge^j)$ where all
  $N^j$ are mutually disjoint, and $r_j = \treeroot(t_j^p)$. Then,
  $\colltoarray(\{t_1,\dots,t_n\}, p)$ is the tree $(N,E,\lnode,\ledge)$ where
  \begin{itemize}
  \item $N = \left(\bigcup_{j=1}^n N^j\right) \cup \{v_0\}$,
  \item $E = \left(\bigcup_{j=1}^n E^j\right) \cup \{(v_0,r_1),\dots,(v_0,r_n)\}$,
  % \item $\prec = \left(\bigcup_{j=1}^n \prec^j\right) \cup \{(r_1,r_2),
  %   (r_2,r_3), \dots, (r_{n-1},r_n)\}^\star$, where $\star$ is the transitive closure,
  \item $\lnode = \big(\bigcup_{j=1}^n \lnode^j\big) \cup \{(v_0, \arraylabel)\}$,
  \item $\ledge = \left(\bigcup_{j=1}^n \ledge^j\right) \cup
    \{((v_0,r_1),0),\dots,((v_0,r_n),n-1)\}$.
  \end{itemize}
\end{description}

We also define $\subtree(t,p)$ for paths $p$ such that $|\eval{p}| > 1$.  In
this case it returns the tree encoding the array of all subtrees hanging from
$p$. Formally, $\subtree(t,p) = \colltoarray(\{t_1, \dots,t_n\}, \varepsilon)$,
where $\{r_1, \dots, r_n\} = \eval{p}$, $N_j$ the set of nodes reachable from
$r_j$ via $E$, and $t_j = \subtree(t,r_j,N_j)$.
We observe that the definition of the $\colltoarray$ operator is recursive as
it uses the generalized $\subtree$ operator.

\subsection{Notes on our Semantics}

We conclude this section by discussing some of the features in which our
semantics differs from the current version of the \mongodb system.  The reason
for this divergence is that with respect to these features, the behavior of
\mongodb might be considered counterintuitive,
% , at least to users familiar with relational databases and SQL,
or even as an inconsistency in the semantics of operators.

\begin{description}
% \item[Comparison of null values.] SQL employs three-valued semantics, where
%   each occurrence of \texttt{NULL} is treated as a \emph{fresh} unknown
%   value, and the expression $(\texttt{NULL}=\texttt{NULL})$ evaluates to
%   \texttt{NULL} (hence is not true).  On the other hand, \mongodb works
%   under two-valued semantics, where \nullvalue is treated as a constant and
%   $(\nullvalue=\nullvalue)$ evaluates to true.  Strangely, in comparisons
%   done within \textbf{\$project} (but not within \textbf{\$match}),
%   \nullvalue is considered less than any constant, in particular
%   $(\nullvalue<-\infty)$.  Since there is no rationale for this, we consider
%   this as a bug.
\item[Comparison of values.]  In our semantics, we employ the classical
  semantics for ``deep'' equality of non-literal values, which differs from the
  actual semantics exhibited by \mongodb based on comparing the binary
  representation of
  values\footnote{\url{https://docs.mongodb.org/manual/reference/bson-types/\#comparison-sort-order}}.

\item[Group.] In \mongodb, the group operator behaves
  differently when grouping by one path and when grouping by multiple paths.  In
  the former case \missingvalue is treated as \nullvalue, while in the latter
  case it is treated differently. More specifically, when grouping by one path
  (e.g., $\group{g/y}{...}$), \mongodb puts the trees with $y=\nullvalue$ and
  those where $y$ is missing into the same group with
  $\id=\lobject g:\nullvalue\robject$. On the contrary, when grouping with
  multiple paths (e.g., $\group{g_1/y_1,\ldots,g_2/y_2}{\ldots}$), the trees
  with all $y_i$ missing are put into a separate group with
  $\id = \lobject\robject$.

\item[Comparing value and path.] The criteria in \emph{match} and Boolean value
  definitions in \emph{project} behave differently in \mongodb.  For instance,
  when comparing a path $p$ of type \tarray with a value $v$ using equality,
  match checks
  \begin{inparaenum}[\itshape (1)]
  \item whether $v$ is exactly the array value of $p$, or
  \item whether $v$ is an element inside the array value of $p$.
  \end{inparaenum}
  Instead, \emph{project} only checks condition~\textit{(1)}. More generally,
  for \emph{match}, $t\models (p=v)$ if there is a node $x$ in $\eval{p}$ or in
  $\eval{p.i}$ for some $i\in I$ such that $\mathsf{value}(x,t)=v$, but for
  \emph{project}, $t\models (p=v)$ if $\tree(v)$ coincides with $\subtree(t,p)$.

  In our semantics \emph{project} checks both conditions~\textit{(1)} and
  \textit{(2)}, also when comparing the values of two paths.

\item[Null and missing values.]  In \mongodb for \emph{match}, $(p=\nullvalue)$
  holds
  \begin{inparaenum}[\itshape (a)]
  \item when $p$ exists and its value is \nullvalue, or
  \item when $p$ is missing.
  \end{inparaenum}
  Instead, for \emph{project}, $(p=\nullvalue)$ holds only for~(a).

  In our semantics, we systematically distinguish the cases (a) and (b).
\end{description}

\section{Nested Relational Algebra to \mquery}
\label{sec:RA2MDB}

\begin{lemma}\label{lem:pipeline-properties}
  The result of $\pipe(\q_1, \q_2)$ contains the result of $\q_i$ in the trees
  with $\text{actRel}=i$ under the key \valuefont{rel$i$}.
\end{lemma}
\begin{proof}
  Let $F$ be a forest, and $F_0$ the result of evaluating of the first 3 stages
  in $\pipe(\q_1, \q_2)$ over $F$.  Then $F_0$ satisfies the
  property:
  \begin{itemize}
  \item[$(\star)$] for each tree $t$ in $F_0$, if $t\models(\text{actRel}=1)$,
    then $t\models \exists\text{rel}1 \land \neg\exists\text{rel}2$, and if
    $t\models(\text{actRel}=2)$, then
    $t\models \exists\text{rel}2 \land \neg\exists\text{rel}1$.
  \end{itemize}
  Moreover, for each tree $t \in F$, there are exactly two trees $t_1$ and
  $t_2$ in $F_0$ such that $t_1\models(\text{actRel}=1)$,
  $\subtree(t_1,\text{rel}1)$ coincides with $t$, and
  $t_2\models(\text{actRel}=2)$, $\subtree(t_2,\text{rel}2)$ coincides with
  $t$.  These follow from the semantics of conditional value definition and of
  $\project{p/q}$ when $q$ is missing from the input trees.

  Let $F_1 = F_0 \pipeline \subq_1(\q_1)$. We prove that
  \begin{description}
  \item[(clean)] $F_1$ satisfies $(\star)$,
  \item[(own)] $(F_1\pipeline\match{\text{actRel}=1})$, coincides with
    $F\pipeline\q_1$, and
  \item[(other)] $(F_1\pipeline\match{\text{actRel}=2})$ coincides with
    $(F_{0}\pipeline\match{\text{actRel}=2})$, which coincides with $F$ (i.e.,
    the ``other'' trees are not affected).
  \end{description}

  It is sufficient to prove the above for the case of $\q_1$ being a single
  stage pipeline $s$.  Consider the following cases:
  \begin{itemize}
  \item $s$ is a match stage $\match\varphi$. Then $\subq_1(\q_1) =
    \match{(\text{actRel}=2) \lor \varphi_{[p/\text{rel}1.p]}}$. Since match
    does not alter the structure of the trees, $F_1$ satisfies $(\star)$.

    Let $t \in (F_{0} \pipeline \match{(\text{actRel}=2) \lor
      \varphi_{[p/\text{rel}1.p]}} \pipeline\match{(\text{actRel}=1)})$. Then by
    the properties of match, it follows that $t \in (F_{0} \pipeline
    \match{(\text{actRel}=1)} \pipeline \match{\varphi_{[p/\text{rel}1.p]}})$. By
    assumption, $(F_{0} \pipeline \match{(\text{actRel}=1)})$ coincides with $F$,
    therefore we obtain that $t$ is in $F \pipeline \q_1$ (up to proper
    renaming). Similarly, in the other direction, when $t \in (F
    \pipeline\q_1)$, we derive that $t \in
    (F_1\pipeline\match{(\text{actRel}=1)})$.

    Since the query $\match{(\text{actRel}=2) \lor
      \varphi_{[p/\text{rel}1.p]}} \pipeline\match{(\text{actRel}=2)}$ is
    equivalent to the query $\match{(\text{actRel}=2)}$, we obtain that the
    forest $(F_{0} \pipeline \match{(\text{actRel}=2) \lor
      \varphi_{[p/\text{rel}1.p]}} \pipeline\match{(\text{actRel}=2)})$ coincides
    with $(F_{0} \pipeline \match{(\text{actRel}=2)})$.

  \item $s$ is an unwind stage $\unwind[+]{p}$. Then $\subq_1(\q_1) =
    \unwind[+]{\text{rel}1.p}$. First, $\subq_1(\q_1)$ does not affect the trees
    with $\text{actRel}=2$ because there does not exist the path
    $\valuefont{rel}1.p$, and $\subq_1(\q_1)$ will preserve all such trees as
    they are. Second, the trees that contain the path $\valuefont{rel}1.p$
    (hence, with $\text{actRel}=1$), will be affected in exactly the same way
    as the trees in $F$ would be affected by $\q_1$. Finally, since unwind does
    not affect other paths than $p$, we have that $F_1$ satisfies the clean
    specialization property.

  \item $s$ is an unwind stage $\unwind{p}$.  Then
    $\subq_1(\q_1) = \match{(\text{actRel}=2) \lor ((\exists \text{rel}1.p)
      \land (\text{rel}1.p\neq []))} \pipeline \unwind[+]{\text{rel}1.p}$.
    Again, $\subq_1(\q_1)$ does not affect the trees with $\text{actRel}=2$
    because they will all pass the match stage and the subsequent unwind will
    preserve them as they are. Second, we note that evaluating $\q_1$ over $F$
    will remove trees where path $p$ does not exist, or $p$ exists and its
    value is empty array. This is done by $\subq_1(\q_1)$ in the match
    stage. The subsequent unwind acts as the unwind above. Again, we have that
    $F_1$ satisfies the clean specialization property.

  \item $s$ is a project stage $\project{p,\, q/d}$. Then,

    $\subq_1(\q_1) = \project{\text{rel}2, ~\text{actRel},\, \text{rel}1.\id,
      ~\text{rel}1.p,
      ~\text{rel}1.q=(\text{actRel}=1)/d_{[p'/\text{rel}1.p']}/\text{dummy}}$.

    It is easy to see that (clean) and (other) are satisfied. As for (own), the
    trees with $\text{actRel}=1$ will keep the paths $\text{rel}1.\id$,
    $\text{rel}1.p$ and the value of the path $\text{rel}1.q$ will be defined
    by $d$. Hence, (own) also holds.

  \item $s$ is a group stage $\group{g/y}{a/b}$. Then

    $\subq_1(\q_1) =
    \begin{array}[t]{@{}l}
      \group{\text{rel}1.g/\text{rel}1.y,
      ~\text{actRel}}{\text{rel}1.a/\text{rel}1.b,
      ~\text{rel}2} \pipeline{} \\[2mm]
      \project{\text{rel}2,\, \text{actRel}/\id.\text{actRel}, ~\text{rel}1.a, ~\text{rel}1.\id.g/ \id.\text{rel}1.g} \pipeline{} \\[1mm]
      \project{\text{actRel},\,\{\text{rel}i=(\text{actRel}=i)/\text{rel}i/\text{dummy}\}_{i=1,2}}  \pipeline{} \\
      \unwind[+]{\text{rel}2}
    \end{array}$

    The result of the first stage is $n+1$ trees where
    \begin{itemize}
    \item one tree originates from all trees with $\text{actRel}=2$, the value of
      $\text{rel}2$ is the array of all such $\text{rel}2$ and $\text{rel}1.a$ is an
      empty array.
    \item $n$ is the number of different values $v_1, \dots, v_n$ of
      $\text{rel}1.y$ in all trees with $\text{actRel}=1$, and each of the $n$
      trees originates from a subset of the trees with $\text{actRel}=1$ and
      $\text{rel}1.y = v_i$, the value of $\text{rel}2$ is the empty array, the
      value of $\text{rel}1.a$ is all $\text{rel}1.b$ in this subset of trees,
      and the value of $\text{rel}1.g$ is $v_i$.
    \end{itemize}
    The result of the second stage is $n+1$ trees where some paths in \id are
    renamed. The result of the third stage is a forest satisfying the clean
    specialization property. In the forth stage, the array $\text{rel}2$ is unwinded,
    hence the trees with $\text{actRel}=2$ are brought in the original shape. It
    is easy to see that all properties are satisfied.
  \end{itemize}

  Since the translation is symmetric, we have also that $F_2 = F_1 \pipeline
  \subq_2(\q_2)$ satisfies the corresponding properties (clean), (own) and
  (other).
\end{proof}

\begin{theoremnum}{\ref{thm:nra-to-mupg-correct}}
  Let $Q$ be a NRA query over $C$.  Then $C \pipeline \raTomaq(Q) \equiv_\S Q$.
\end{theoremnum}

\begin{proof}[Proof Sketch]
Follows from the definition of $\rschema_\tau(C)$,
Lemma~\ref{lem:pipeline-properties} and the semantics of \mquery stages.
\end{proof}

\begin{example}
  Consider the following NRA queries over $\rschema(\tau_\bios)$, where
  \valuefont{fn} stands for name.first, \valuefont{ln} for name.last,
  \valuefont{an} for awards.award, and \valuefont{ay} for awards.year:
  \[
  \begin{array}{@{}l}
    Q = \pi_{\text{fn},\, \text{ln},\, \text{an},\, \text{ay}}(\unnest_{\text{awards}}(\bios))
    \\
    Q' = \sigma_{(\text{rel}1.\text{ay}=\text{rel}2.\text{ay}) \land ((\text{rel}1.\text{fn}\neq \text{rel}2.\text{fn}) \lor (\text{rel}1.\text{ln}\neq \text{rel}2.\text{ln}))}(Q \times Q)
  \end{array}
  \]
  Thus, $Q'$ asks for a pair computer scientists that received an award in the
  same year. We illustrate some steps of $\raTomaq$:
  \begin{itemize}
  \item $\raTomaq(Q) = \begin{array}[t]{@{}l}
      \project{\id,\, \text{awards},\, \text{birth},\, \text{contribs},\, \text{fn},\, \text{ln}} \pipeline %{}\\
      \unwind{\text{awards}} \pipeline \project{\id,\, \text{an},\, \text{ay},\, \text{birth},\, \text{contribs},\, \text{fn},\, \text{ln}} \pipeline{} \\
      \project{\text{fn},\, \text{ln},\, \text{an},\, \text{ay}}
      \\
    \end{array}$
  \item $\subq_1(\raTomaq(Q)) = \begin{array}[t]{@{}l}
      \project{\substack{\text{rel}2,\, \text{actRel},\, \text{rel}1.\id,\, \text{rel}1.\text{awards},\, \text{rel}1.\text{birth},\, \text{rel}1.\text{contribs},\, \text{rel}1.\text{fn},\, \text{rel}1.\text{ln}}} \pipeline{}\\
      \match{(\text{actRel}\neq 1) \lor ((\exists \text{rel}1.\text{awards}) \land (\text{rel}1.\text{awards}\neq []))} \pipeline
      \unwind[+]{\text{awards}} \pipeline{}\\
      \project{\substack{\text{rel}2,\, \text{actRel},\, \text{rel}1.\id,\, \text{rel}1.\text{an},\, \text{rel}1.\text{ay},\, \text{rel}1.\text{birth},\, \text{rel}1.\text{contribs},\, \text{rel}1.\text{fn},\, \text{rel}1.\text{ln}}} \pipeline{} \\
      \project{\text{rel}2,\, \text{actRel},\, \text{rel}1.\text{fn},\,
        \text{rel}1.\text{ln},\, \text{rel}1.\text{an},\,
        \text{rel}1.\text{ay}}
      \\
    \end{array}$
  \item $\raTomaq(Q') = \begin{array}[t]{@{}l}
      \pipe(\raTomaq(Q),\raTomaq(Q)) \pipeline%{}\\
      \group{}{\text{rel}1,\, \text{rel}2} \pipeline \unwind{\text{rel}1} \pipeline \unwind{\text{rel}2} \pipeline{}\\
      \project{\substack{\text{rel}1.\text{fn},\, \text{rel}1.\text{ln},\, \text{rel}1.\text{an},\, \text{rel}1.\text{ay}, \text{rel}2.\text{fn},\, \text{rel}2.\text{ln},\, \text{rel}2.\text{an},\, \text{rel}2.\text{ay},~~~\\ \text{cond}/((\text{rel}1.\text{ay}=\text{rel}2.\text{ay}) \land ((\text{rel}1.\text{fn}\neq\text{rel}2.\text{fn}) \lor (\text{rel}1.\text{ln}=\text{rel}2.\text{ln})))}} \pipeline%{}\\
      \match{(\text{cond}=\truevalue)} \pipeline{}\\
      \project{\substack{\text{rel}1.\text{fn},\, \text{rel}1.\text{ln},\, \text{rel}1.\text{an},\, \text{rel}1.\text{ay}, \text{rel}2.\text{fn},\, \text{rel}2.\text{ln},\, \text{rel}2.\text{an},\, \text{rel}2.\text{ay}}}\hfill\qedfull \\
    \end{array}$
  \end{itemize}
\end{example}

\begin{theoremnum}{\ref{thm:nra-to-mupgl-correct}}
  Let $Q$ be an NRA query over $C_1,\dots,C_n$, and
  $\q=C_1 \pipeline \mathsf{bring}(C_2,\ldots,C_n) \pipeline
  \raTomaq^\star(Q)$.  Then $\q\equiv_\S Q$.  Moreover, the size of $\q$ is
  polynomial in the size of $Q$.
\end{theoremnum}

\begin{proof}[Proof Sketch]
The correctness of the translation follows from
Theorem~\ref{thm:nra-to-mupg-correct}, considering the form of
$\mathsf{bring}(C_2,\ldots,C_n)$.  As for the size of $\q$, it suffices to
observe that in the inductive definition of the translation, at each step the
translation is called recursively on each subquery at most once, and moreover a
linear number of stages, each of linear size, are added to the resulting
pipeline.
\end{proof}

\section{MQuery to Nested Relational Algebra}
\label{app:mquery2nra}

In this section, we show how to translate \mqueries composed of well-typed
stages to NRAs.  First, given a set $\S$ of constraints, and a well-typed
\mquery stage $s$, we define an NRA query $\maqTonra{s}$.  Then, for an \mquery
$C \pipeline s_1 \pipeline \cdots \pipeline s_n$, the corresponding NRA query
is defined as $(\maqTonra{s_1} \circ \cdots \circ
\maqTonra{s_n})(C)$\footnote{We follow the convention that $(f \circ g)(x) =
  g(f(x))$.}.
Below we assume that the input to $\maqTonra{s}$ is a query $\Rin$ with the
associated attributes $\att(\Rin)$, and $\tau$ is the type corresponding to the
schema of $\Rin$.

\subsection{Match}

We assume match criteria $\varphi$ to be in negation normal form, that is,
negation appears directly in front of the atoms of the form $(p=v)$ and
$\exists p$.
We say that a path $p$ is \emph{nested in $\tau$} if $\type(p',\tau)=\tarray$
for some strict prefix $p'$ of $p$.

Let $\tau$ be the input type. We now define the translation of
$\match{\varphi}$ with respect to $\tau$. It is done in 3 steps, progressing
from quasi-atomic criteria to the most general ones.

\smallskip\noindent\textbf{\textsf{Step 1.}}  We first introduce the
translation for so-called \emph{simple filter} criteria. A criterion $\varphi$
is called a \emph{simple filter} criterion, if $\varphi$ is conjunction- and
disjunction-free, and, for a non-nested path $p$, is either of the form
$\exists p$, $\lnot (\exists p)$, or of the form $(p=v)$, $\neg(p=v)$ and if
$\type(p,\tau)=\tarray$, then $v$ is an array of values.

Next, we define an auxiliary function $f_\tau(\varphi)$, whose goal is to
translate simple filter criteria properly also when $p$ is not an attribute
name in $\rschema(\tau)$.

For $\varphi=(p=v)$, we need to check whether $p$ and $v$ are ``compatible''
with respect to $\tau$, that is, whether $v$ is of type $\subtree(p,\tau)$.
When they are incompatible, we set $f_\tau(p=v) = \falsevalue$; otherwise if
$v$ is of type $\subtree(p,\tau)$, we define $f_\tau(p=v)$ as\\[1mm]
$\begin{cases}
  (p=v), &\text{if }\type(p,\tau)=\tliteral\\
  (p=\rel(F))\text{ where }F=\{\attach(p, \tree(v_i))\}_{i=1}^n\text{ for } v =
  [v_1,\ldots, v_n], &\text{if }\type(p,\tau)=\tarray\\
  \bigwedge\{ (p'=v') \mid (p':v')\in \row_\tau(R_\tau, \epsilon, \attach(p,
  v)),\ v'\neq \missingvalue\},
  &\text{if }\type(p,\tau)=\tobject.
\end{cases}$\\[1mm]

For $\varphi=\exists p$, $f_\tau(\varphi)$ is defined as follows (we write
$q\neq v$ as a shortcut for $\neg(q=v)$):\\[1mm]
$\begin{cases}
  (p\neq\missingvalue), &\text{if }\type(p,\tau)\in\{\tliteral,\tarray\}\\
  \bigvee\{ (p.q\neq\missingvalue) \mid p.q\in \ratt_\tau(\epsilon)\}, &\text{if
  } \type(p,\tau)=\tobject.
\end{cases}$\\[2mm]

For $\varphi=\lnot\varphi'$, we set $f_\tau(\varphi)=\lnot(f_\tau(\varphi'))$.
Now, for a simple filter criterion $\varphi$, we translate $\match{\varphi}$
simply as $\sigma_{f_\tau(\varphi)}$.

\ignore{\nb{old}We assume match criteria $\varphi$ to be in negation normal form, that is,
negation appears directly in front of the atoms of the form $(p=v)$ and $\exists
p$.
We say that a path $p$ is \emph{nested in $\tau$} if $\type(p',\tau)=\tarray$
for some strict prefix $p'$ of $p$.

We first introduce the translation for conjunction- and disjunction-free
criteria $\varphi$ according to the type $\tau$, that is when $\varphi$ is of
the form $(p=v)$, $\exists p$, $\neg(p=v)$, or $\lnot\exists p$. To this
purpose, we define an auxiliary function $f(\varphi)$, whose goal is to
translate the criteria properly also when $p$ is not an attribute name in
$\rschema(\tau)$.

For $\varphi=(p=v)$, we need to check whether $p$ and $v$ are ``compatible''
with respect to $\tau$, that is, whether $v$ is of type $\subtree(p,\tau)$ or
of type $\subtree(p.0,\tau)$.
When they are incompatible, we set $f(p=v) = \falsevalue$; otherwise if $v$ is
of type $\subtree(p,\tau)$, we define $f(p=v)$ as
\begin{itemize}%\itemsep-\parsep
\item $(p=v)$, if $\type(p,\tau)=\tliteral$;
\item $(p=\rel(F))$ where $F=\{\attach(p, \tree(v_i))\}_{i=1}^n$ and $\tau$ is the type of $F$, if
  $\type(p,\tau)=\tarray$ and $v = [v_1,\ldots, v_n]$;
\item
  $\bigwedge_{\substack{(p':v')\in \row_\tau(R_\tau, \epsilon, \attach(p, v)),\
    v'\neq \missingvalue}} (p'=v')$,
  % where $\mathit{tup}=\row_\tau(R_\tau, \epsilon, \attach(p, v))$,
  if \mbox{$\type(p,\tau)=\tobject$};
\end{itemize}
and if $v$ is of type $\subtree(p.0,\tau)$ (hence, $\type(p,\tau)=\tarray$), we
define $f(p=v)$ as
\begin{itemize}%\itemsep-\parsep
\item $(v \in p)$, if $\type(p[], \tau)=\tliteral$;
\item \mbox{$(\row_\tau(R_\tau, \epsilon, \attach(p, v)) {\in} p)$, if
   $\type(p[], \tau)=\tobject$.}
\end{itemize}

For $\varphi=\lnot(p=v)$, we set $f(\varphi)=\lnot(f(p=v))$.

For $\varphi=\exists p$, $f(\varphi)$ is defined as follows (we write $q\neq v$
as a shortcut for $\neg(q=v)$):
\begin{itemize}[$\bullet$]\itemsep-\parsep
\item $(p\neq\missingvalue)$, if $\type(p,\tau)\in\{\tliteral,\tarray\}$;
\item $\bigvee_{p.q\in \ratt_\tau(\epsilon)} (p.q\neq\missingvalue)$, if
  \mbox{$\type(p,\tau)=\tobject$}.
\end{itemize}

For $\varphi=\neg\exists p$, we set $f(\varphi)=\neg ( f(\exists p))$.

If the path $p$ in $\varphi$ is not nested in $\tau$, we translate
$\match{\varphi}$ as $\sigma_{f(\varphi)}$, i.e.,
$\maqTonra{\match{\varphi}}=\sigma_{f(\varphi)}$.\nb{end of old}
}

\begin{lemma}
  Let $\varphi$ be a simple filter and $r{=}\row_\tau(R_\tau, \epsilon, t)$.
  Then for any tree $t$ of $\tau$, we have $t\models \varphi$ if and only if
  $\evalu(f_\tau(\varphi), r) = \truevalue$.
\end{lemma}

\begin{proof}[Proof Sketch]
  We start with a simple filter $\varphi=(p=v)$. When $p$ and $v$ are
  incompatible with respect to $\tau$, it is easy to see that
  $t \not\models (p=v)$, and
  $\sigma_{f_{\tau}(\varphi)} \rel(\{t\}) = \emptyset$. Below we assume $p$ and
  $v$ are compatible.

  \begin{itemize}
  \item If $\type(p,\tau)=\tliteral$, then the path $p$ corresponds to a top
    level atomic attribute in the relation $\R$. When $t\models \varphi$, there
    is a node $x\in \eval{p}$ and $\mathsf{value}(x,t)=v$. According to the
    construction of $r$, $t \models \varphi$ iff $(p:v) \in r$.
  \item If $\type(p,\tau)=\tarray$, then the path $p$ corresponds to a top
    level relational attribute in $r$. Checking $p=v$ amounts to check
    $(p{:}\rel(F)) \in r$.
  \item If $\type(p,\tau)=\tobject$, then the path $p$ leads to a subtree of
    $\tau$. To check $p=v$, we have to check all paths $p'$ with prefix
    $p$ such that
    $\type(p',\tau)=\tliteral$, which corresponds to multiple attributes with
    prefix $p$ in r.
% According to the construction of
%     $r'=\row_\tau(R_\tau, \epsilon, \attach(p, v))$, each condition $(p'=v')$
%     in $f_\tau(p=v)$ does the checking for each $(p':v') \in r'$.
  \end{itemize}

  For $\varphi=\exists p$, evaluating $\varphi$ corresponds to checking whether
  the corresponding value in $r$ is $\missingvalue$. When
  $\type(p,\tau)\in\{\tliteral,\tarray\}$, it checks the single attribute $p$
  in $r$; when $\type(p,\tau)=\tobject$, it checks all attributes with prefix
  $p$.

  Finally, the result holds for $\varphi=\neg \varphi'$ trivially because
  $ t \models \varphi $ iff $ t \not\models \varphi'$ iff
  $ \evalu(f_\tau,\varphi')=\falsevalue$ iff $ \evalu(f_\tau,\varphi)=\truevalue$.
\end{proof}

\smallskip\noindent\textbf{\textsf{Step 2.}}  Consider now the case of a
criterion $\varphi$ of the form $p=v$ or $\lnot(p=v)$ that is not a simple
filter criterion, and for simplicity, assume that the level of nesting of $p$
is 1, i.e., there is only one (possibly non-strict) prefix $q$ of $p$ with
$\type(q,\tau) = \tarray$. We call $q$ the \emph{parent relation} of $p$. Then,
$q$ is a sub-relation of $R_{\tau}$, and $p$ is a prefix of some path in
$\ratt_{\tau}(q)$.  To check the condition on $p$ according to the semantics of
match, we need to be able to access the contents of the sub-relation $q$ by
unnesting it, but to return the original (i.e., nested) relation $q$.  So
before actually doing a selection, we apply several preparatory phases.
\begin{itemize}
\item
  $\mathsf{AddID}=\pi_{\att(Q),\, \mathit{id}.\att(Q)/\att(Q)} \circ
  \nest_{\mathit{id}.\att(Q) \to \mathit{ID}}$
  creates an identifier for each tuple (required for negative $\varphi$, for
  which we need to unnest and then to nest back):

\item $\mathsf{AddDup}_{\varphi}=\pi_{\att(Q),\, \mathit{ID},\, q'/q}$ creates
  a copy $q'$ of the sub-relation $q$;

\item $\mathsf{Prep}_{\varphi}$ does proper preprocessing of the new attribute
  $q'$
  \begin{itemize}[--]
  \item $\mathsf{Prep}_{\varphi} = \unnest_{q'} \circ \pi_{\att(Q),\,
      \mathit{ID},\, \{a'/a\mid a\in\att(q')\}}$ for positive $\varphi$,
  \item
    $\mathsf{Prep}_{\varphi} = \unnest_{q'} \circ \pi_{\att(\Rin),\,
      \mathit{ID},\, \mathit{res}/f'(\varphi')} \circ \nest_{\{\mathit{res}\}
      \rightarrow \mathit{cond}}$ for $\varphi$ of the form $\neg\varphi'$.
  \end{itemize}
\end{itemize}

Let $\tau'$ be the type resulting from unwinding $q$ in $\tau$, i.e.,
$\{\tau'\} = \{\tau\}\pipeline \unwind{q}$.
We define $f'(\varphi)$ as $\falsevalue$ when $p$ and $v$ are incompatible with
respect to $\tau$; as $f_{\tau'}(\varphi)[p \to p.\literalattr']$ if $p=q$,
and as $f_{\tau'}(\varphi)[p\to p']$ otherwise, for positive $\varphi$; and as
$(\mathit{cond}=\{(\mathit{res}:\falsevalue)\})$, for negative $\varphi$.
Then, we apply selection with the condition  $f'(\varphi)$, and
finally, project away the auxiliary columns $q'$ and $\mathit{ID}$.
% Then, to translate $\match{\varphi}$, we first clone the column $q$ to $q'$,
% unnest $q'$, apply the filter, and finally project the original columns.
More precisely,
\[
  \maqTonra{\match{\varphi}} = \mathsf{AddID} \circ \mathsf{AddDup}_{\varphi}
  \circ \mathsf{Prep}_{\varphi} \circ \sigma_{f'(\varphi)} \circ \pi_{\att(\Rin)}.
\]
This translation can be extended to the case of multiple levels of nesting, and
we omit the details, which are tedious but straightforward.

\begin{lemma}
  Given a criterion $\varphi$ of the form $p=v$ or $\lnot(p=v)$ that is not a
  simple filter criterion, and such that the level of nesting of $p$ is 1,
  then for any tree $t$ of $\tau$, we have $t\models \varphi$ if and only if
  $\maqTonra{\match{\varphi}}(\rel(\{t\})) = \rel(\{t\})$.
\end{lemma}

\begin{proof}[Proof Sketch]
  We observe that $\mathsf{AddID}$ adds one relational attribute named
  $\mathit{ID}$ which has the value $\{r\}$ where $r=\row_\tau(R_\tau,
  \epsilon, t))$.

  For $\varphi=(p=v)$, the translation follows a ``standard'' approach
  \cite{Colby89} for accessing nested attributes: it unnests the sub-relation
  of the nested attributes, then checks the condition over the flattened
  sub-relation. The final project ensures that we return a subset of the
  original relation.

  For $\varphi=\lnot(p=v)$, we first observe that $t\models \lnot(p=v)$ iff for
  each $x \in \eval{p}$ we have that $\mathsf{value}(x,t)\neq v$. Therefore, we
  have to check the condition over all rows in the flattened sub-relation. To
  this purpose, $\mathsf{Prep}_{\varphi}$ nests the result of evaluating the
  condition in an attribute $\mathit{cond}$, while $\mathit{ID}$ ensures that
  exactly the original tuples have been reconstructed after nesting. If
  $t\models \lnot(p=v)$, then \textit{all} values in $\mathit{res}$ are false,
  and the filter $(\mathit{cond}=\{(\mathit{res}:\falsevalue)\})$ evaluates to
  true. Otherwise if $t\not\models \lnot(p=v)$, then \textit{some} value in
  $\mathit{res}$ is true, and the filter
  $(\mathit{cond}=\{(\mathit{res}:\falsevalue)\})$ evaluates to false.
\end{proof}

\smallskip\noindent\textbf{\textsf{Step 3.}}  Now we deal with arbitrary
criteria $\varphi$. Let $\alpha_1,\dots,\alpha_n$ be all positive non-simple
filter literals in $\varphi$, $\beta_1,\dots,\beta_m$ all negative non-simple
filter literals in $\varphi$, and $\delta_1,\dots,\delta_k$ all simple filter
literals in $\varphi$. For each non-simple filter literal over a path $p$, we
need to create a separate duplicate of the parent relation $q$ of $p$. So below
we assume that
$\mathsf{AddDup}_{\alpha_1,\dots,\alpha_n,\beta_1,\dots,\beta_n}$ creates a new
column named uniquely for each literal $\alpha_i$ and $\beta_j$, and that
$\mathsf{Prep}_{\beta_j}$ projects also all these new columns and gives unique
names to the sub-relations $\mathit{cond}$ (and projects them as well). We set
$f'(\delta_i)=f(\delta_i)$ and let $f'(\varphi)$ be the result of replacing in
$\varphi$ each literal $\ell$ by $f'(\ell)$ (respecting the unique names of the
attributes and sub-relations for each literal over a nested path).
Then $\maqTonra{\match{\varphi}}$ is the query
\[
\begin{array}{l}
  \mathsf{AddID} \circ
  \mathsf{AddDup}_{\alpha_1,\dots,\alpha_n,\beta_1,\dots,\beta_m} \circ%{}\\
  \mathsf{Prep}_{\beta_1} \circ \dots \circ \mathsf{Prep}_{\beta_m} \circ
  \mathsf{Prep}_{\alpha_1} \circ \dots \circ \mathsf{Prep}_{\alpha_n} \circ{}\\
  \sigma_{f'(\varphi)} \circ \pi_{\att(\Rin)}.
\end{array}
\]

\begin{lemma}
  \label{lem:maq2nra:match:3}
  Given an arbitrary criterion $\varphi$, for any tree $t$ of type $\tau$, we
  have that $t\models \varphi$ if and only if
  $\maqTonra{\match{\varphi}}(\rel(\{t\})) = \rel(\{t\})$.
\end{lemma}

\begin{proof}[Proof Sketch]
  The operators $\textsf{Prep}_{\beta_i}$ and $\textsf{Prep}_{\alpha_j}$
  evaluate $\beta_i$ and $\alpha_j$ separately and $f'(\varphi)$ combines the
  results of evaluation of $\alpha_j$, $\beta_i$, and $\delta_k$.
\end{proof}

\begin{lemma}
\label{lem:maq2nra:match}
  Let $F$ be a forest of type $\tau$ and $s=\match{\varphi}$ a match
  stage of \mquery, then
  $F \pipeline s \simeq \maqTonra{s} (\rel(F)) $ .
\end{lemma}

\begin{proof}
  The result follows directly from Lemma~\ref{lem:maq2nra:match:3}.
\end{proof}

\subsection{Unwind}
The unwind operator $\unwind{p}$ can be translated to unnest in NRA:
$\maqTonra{\unwind{p}} = \unnest_{p}$.
To deal with $\unwind[+]{p}$, we first replace empty sub-relations $p$ with the
relation consisting of one tuple
$\{a_1:\missingvalue, \dots, a_n:\missingvalue\}$, where
$\{a_1,\dots,a_n\} = \ratt_\tau(p)$, and then apply the normal unnest. Hence,
$\maqTonra{\unwind[+]{p}}$ is defined as
\[\pi_{\att(\Rin)\setminus\{p\},\,
    p/\cond{(p=\subrel())}{\subrel(\{a_1:\missingvalue, \dots,
      a_n:\missingvalue)\})}{p}} \circ \unnest_{p}.\]

\begin{lemma}\label{lem:maq2nra:unwind}
  Let $F$ be a forest of type $\tau$ and $s=\unwind[n]{p}$ an unwind stage of
  \mquery, then $F \pipeline s \simeq \maqTonra{s} (\rel_\tau(F)) $ .
\end{lemma}

\begin{proof}
  Straightforward considering the semantics of unwind and unnest.
\end{proof}

\subsection{Project}

We consider a well-typed project stage $\project{P}$ with an input type $\tau$.

Let $p$ be a path in $\tau$. We define function $\projattr_\tau(p)$.
\begin{itemize}
\item If $\type(p,\tau)\in\{\tliteral,\tarray\}$, then $\projattr_\tau(p) =
  \{p\}$.
\item If $\type(p, \tau)=\tobject$, then $\projattr_\tau(p)=\{ p.p' \mid p.p'
  \in \ratt_{\tau}(q)\}$, where $q$ is the longest prefix of $p$ such that
  $\type(q,\tau)=\tarray$, when $p$ is nested, and $q=\epsilon$, if $p$ is not
  nested.
\end{itemize}

\def\projval{\mathsf{val}}
\def\projexp{\mathsf{exp}}
\def\bool{\mathsf{bool}}

To define $\projattr_{\tau}(q/d)$, we first define function $\projval_{\tau}(d)$ for value
definitions $d$ returning a set of pairs $(p,v)$ where $p$ is a path, and $v$
is a literal or array value, a path, a conditional expression, or an expression
of the form $\bool(\beta)$ for a Boolean value definition. In an array value,
the leafs of the trees might have paths or $\bool(\beta)$ as values.
Below we assume that for a value definition $d'$,
%\nb{GX: why not $d$?}
 $\tau_{d'}$
is the type of $d'$ with respect to $\tau$, as defined in
Section~\ref{sec:mquery2nra}. 
\begin{itemize}
\item $\projval_{\tau}(p)$ for a path $p$ is defined as
  \begin{itemize}
  \item $\{(\varepsilon,p)\}$, if $\type(p,\tau)\in\{\tliteral,\tarray\}$,
  \item $\{(r,p.r) \mid p.r \in \ratt_\tau(p) \}$, if $\type(p,\tau)=\tobject$.
  \end{itemize}

\item $\projval_{\tau}(v)$ for a constant value $v$ is defined as
  \begin{itemize}
  \item $\{(\varepsilon,v)\}$ if $v$ is of literal type,
  \item $\{(p, \subtree(\tree(v), p)) \mid p \in \ratt_{\tau'}(\varepsilon)\}$
    if $v$ is of object type $\tau'$.
  \end{itemize}
% \item $\projval_{\tau}(\object{k_1{:}d_1,\dots,k_n{:}d_n}) = \{(k_i.q,v) \mid
%   i=1..n,\, (q,v)\in\projval_{\tau_i}(d_i)\}$ where $d_i$ is of type $\tau_i$.

\item $\projval_{\tau}([d_1,\dots,d_n]) = \{(\varepsilon,[t_1,\dots,t_n])\}$,
  % where $t_i$ is the tree such that $\subtree(t_i,p)$ coincides with $v$, for
  % each $(p,v) \in \projval_{\tau}(d_i)$.
  where $t_i = \bigoplus_{(p,v)\in \projval_{\tau}(d_i)} \attach(p,\tree(v))$ .

\item $\projval_{\tau}(\beta)$ for a Boolean value definition $\beta$ is defined as
  $\{(\varepsilon, \bool(\beta))\}$.

\item $\projval_{\tau}(\cond{c}{d_1}{d_2})$ is %\nb{check: suspicious because of no recursion}
  \begin{itemize}
  \item $\{(\varepsilon,\cond{c}{v_1}{v_2})\}$, for $(\varepsilon,v_i) \in
    \projval_{\tau}(d_i)$, if $\tau_{d_1}$ is a literal or array type; and
  \item
    $\{(p,\cond{c}{v_1'}{v_2'}) \mid p \in \ratt_{\tau_{d_1}}(\varepsilon),
    v_i' = \subtree(t_i, p)\}$, where\\
    $t_i = \bigoplus_{(q,v)\in \projval_{\tau}(d_i)}\attach(q, \tree(v))$, if
    $\tau_{d_1}$ is an object type.
  \end{itemize}
\end{itemize}

We extend the function $\projattr_\tau$ to expressions $q/d$ as follows:
\[\projattr_{\tau}(q/d)=\{ q.p/\projexp(v) \mid (p,v) \in \projval_{\tau}(d)\},\]
where
\begin{itemize}
\item $\projexp(p) = p$ for a path $p$,
\item $\projexp(v) = v$ for a literal value $v$,
\item $\projexp(\bool(\beta)) = f'(\beta)$ for a Boolean value definition $\beta$.
\item $\projexp(\cond{c}{v_1}{v_2}) =\cond{f'(c)}{\projexp(v_1)}{\projexp(v_2)}$,
\item $\projexp([t_1,\dots,t_n]) = \subrel(\projexp(t_1),\dots,\projexp(t_n))$,
\item $\projexp(\tree(\object{k_1{:}v_1,\dots,k_n{:}v_n})) = \{k_1{:}v_1,\dots,k_n{:}v_n\}$,
\end{itemize}
Here, for a Boolean value definition $\beta$, we use the function $f'$ defined
in the translation of match. An expression $p_1=p_2$ is called a \emph{simple
  filter} if both $p_1$ and $p_2$ are non-nested paths such that
$\type(p_1,\tau) = \type(p_2,\tau)$.

Now, $\maqTonra{\project{P}}$ is the query
\[
  \mathsf{AddID} \circ
  \mathsf{AddDup}_{\alpha_1,\dots,\alpha_n,\beta_1,\dots,\beta_m} \circ
  \mathsf{Prep}_{\beta_1} \circ \dots \circ \mathsf{Prep}_{\beta_m} \circ
  \mathsf{Prep}_{\alpha_1} \circ \dots \circ \mathsf{Prep}_{\alpha_n} \circ
  \pi_{\projattr_{\tau}(P)},
\]
where $\alpha_1,\dots,\alpha_n$ are all positive non-simple filter literals,
$\beta_1,\dots,\beta_m$ all negative non-simple filter literals, and
$\delta_1,\dots,\delta_k$ all simple filter literals in Boolean value
definitions in $P$, and $\projattr_{\tau}(P) = \bigcup_{p \in P}\projattr_{\tau}(p) \cup
\bigcup_{q/d \in P}\projattr_{\tau}(q/d)$ (respecting the unique names of the
attributes and sub-relation for each literal over a nested path).

\begin{lemma}\label{lem:maq2nra:project}
  Let $F$ be a forest of type $\tau$ and $s=\project{P}$ a project stage
  of \mquery, then $F \pipeline s \simeq \maqTonra{s} (\rel_\tau(F))$.
\end{lemma}

\begin{proof}[Proof Sketch]
  The operators $\mathsf{AddID}, \mathsf{AddDup}, \mathsf{Prep}$ evaluate
  all the non-simple literal filters used in $d$. Then for any $p/d \in P$, the
  function $\projval_{\tau}(d)$ splits $d$ into multiple definitions $p_i:v_i$, which
  guarantees that each $p_i/v_i$ can be translated by the $\projexp$ function
  into an equivalent extended NRA project expression.
\end{proof}

\subsection{Group}

To translate the group operator $\group{G}{a_1/b_1,..,a_m/b_m}$,  we
\begin{inparaenum}[\itshape (i)]
\item rename attributes according to $G$ and project only the attributes $b_i$;
\item nest the attributes $\{b_1,\dots,b_m\}$ into $\mathit{acc}$, and create m
  copies of $\mathit{acc}$ (for each $b_i$);
\item finally, for each $i=1,\dots,m$, we intuitively project the column $b_i$
  from the relation $\mathit{acc}_i$. Since it is a sub-relation, we first
  unnest it, project only $b_i$, and then nest $b_i$ into sub-relation $a_i$.
\end{inparaenum}
For simplicity, we only show the translation when all types of paths in $G$ and
$b_i$'s are either \tliteral or \tarray. In this case,
$\maqTonra{\group{G}{a_1/b_1,..,a_m/b_m}}$ is defined as
\[\begin{array}{l}
  \pi_{\mathsf{id}(G), \, b_1, \ldots, b_m} \circ{}\\ %
  \nest_{\{b_1,\ldots, b_m\} \rightarrow \mathit{acc}} \circ
  \pi_{\mathsf{idAtt}(G), \, \mathit{acc}_1/\mathit{acc}, \ldots, \mathit{acc}_m/\mathit{acc}} \circ{}\\ %
  \unnest_{\mathit{acc}_1} \circ \pi_{\mathsf{idAtt}(G), \, b_1, \, \mathit{acc}_2, \ldots, \mathit{acc}_m} \circ \nest_{\{b_1\} \rightarrow a_1} \circ{}\\
  \unnest_{\mathit{acc}_2} \circ \pi_{\mathsf{idAtt}(G), \, a_1, \, b_2, \, \mathit{acc}_3, \ldots, \mathit{acc}_m} \circ \nest_{\{b_2\} \rightarrow a_2}\circ{}\\
  \cdots \\
  \unnest_{\mathit{acc}_m} \circ \pi_{\mathsf{idAtt}(G), \, a_1, \ldots, a_{m-1}, \, b_m} \circ \nest_{\{b_m\} \rightarrow a_m}\\
\end{array}\]
where $\mathsf{id}(G) = \id.g_1/y_1,\dots,\id.g_m/y_m$ and
$\mathsf{idAtt}(G)=\id.g_1, \ldots, \id.g_n$ if $G = g_1/y_1,\dots,g_m/y_m$, and
$\mathsf{id}(G) = \id/\nullvalue$ and $\mathsf{idAtt}(G)=\id$, if $G$ is
empty. This translation can be extended to the case in which some types are
\tobject by using the $\projattr_\tau$ function defined above.  We omit the
details.

\begin{lemma}\label{lem:maq2nra:group}
  Let $F$ be a forest of type $\tau$ and
  $s=\group{G}{a_1/b_1, \ldots, a_m/b_m}$ a group stage of \mquery,
  then $F \pipeline s \simeq \maqTonra{s} (\rel_\tau(F))$.
\end{lemma}

\begin{proof}
  Straightforward considering the semantics of group and nest.
\end{proof}

% \begin{proof}[Proof Sketch]
%   \begin{itemize}
%   \item If $G=\nullvalue$, after the first project
%     $\pi_{\mathsf{id}(G), \, a_1/b_1, \ldots, a_m/b_m}$,
%     $(\id:\nullvalue)$ is in all resulting tuples. Next
%     $\nest_{\{a_1,\ldots, a_m\} \rightarrow \mathit{acc}}$ collects
%     all tuples into a single tuple with an attribute $\id$
%     (corresponds to $\attach(\id,\nullvalue)$) and another attribute
%     $\mathit{acc}$.  Finally,
%     $\pi_{\mathsf{idAtt}(G),\,a_1/\mathit{acc}.a_1,\ldots,
%       a_m/\mathit{acc}.a_m}$ is taking care of accumulate all $a_i$
%     into their (single) group, corresponding to
%     $\attach(a_j, \colltoarray(F, b_j))$.
%   \item If $G$ is empty.
%   \item If $G=g_1/y_1,\ldots,g_m/y_m$, after the first project
%     $\pi_{\mathsf{id}(G), \, a_1/b_1, \ldots, a_m/b_m}$, the \id
%     attribute corresponds to group condition in \mongodb. Next
%     $\nest_{\{a_1,\ldots, a_m\} \rightarrow \mathit{acc}}$ collects
%     all tuples into tuples with an attribute $\id$ (corresponds to
%     $\attach(\id,\nullvalue)$) and another attribute $\mathit{acc}$.
%     Finally,
%     $\pi_{\mathsf{idAtt}(G),\,a_1/\mathit{acc}.a_1,\ldots,
%       a_m/\mathit{acc}.a_m}$ is taking care of accumulate all $a_i$
%     into their (single) group, corresponding to
%     $\attach(a_j, \colltoarray(F, b_j))$.
%   \end{itemize}

% \end{proof}

% \smallskip%
% \noindent
% \textsc{Lookup.}

\subsection{Lookup}

For a lookup operator $\lambda^{p_1=C'.p_2}_{p}$, we assume that $C'$ is of
type $\tau'$.  To translate lookup, we first compute a subquery $Q_1$ that
extends each tuple $r$ in the input relation with a Boolean attribute noMatch
that encodes whether there exists at least one matching tuple in the relation
$C'$:
\[Q_1 = \begin{array}[t]{@{}l}
  (\Rin \times C') \circ{}\\
  \pi_{\text{rel1}.\att(\Rin),\  \text{cond}/(\text{rel1}.\projattr_\tau(p_1)=\text{rel2}.\projattr_{\tau'}(p_2))} \circ{}\\
  \nest_{\{\text{cond}\}\to \text{er}} \circ{}\\
  \pi_{\att(\Rin)/\text{rel1}.\att(\Rin),\ \text{noMatch}/(\text{er}=\subrel(\{\text{cond}:\falsevalue\})}\end{array}\]
where $\text{rel1}.\projattr(p_1)=\text{rel2}.\projattr(p_2)$ is an
abbreviation for a conjunction of multiple equality conditions, if
$\subtree(\tau, p_1)$ coincides with $\subtree(\tau',p_2)$, and \falsevalue
otherwise.  Then, we
\begin{inparaenum}[\itshape (i)]
\item cross-product the result of $Q_1$ with $C'$ again,
\item from all possible pairs $(r,r') \in Q_1 \times C'$ we select only those
  for which either there is no match, or the joining condition is satisfied,
\item we nest the attributes from $C'$ into $p$ (to capture the behavior of
  lookup that stores all matching trees in an array), and finally
\item for all tuples $r \in Q_1$, for which there are no matching tuples, we
  replace the value of $p$ by the empty relation.
\end{inparaenum}
More precisely, we define $\maqTonra{\lambda^{p_1=C'.p_2}_{p}}(\Rin, C')$ as
\[\begin{array}{@{}l}
  (Q_1 \times C') \circ{}\\
  \sigma_{\text{noMatch}=\truevalue\ \lor\ \text{rel1}.\projattr_\tau(p_1)=\text{rel2}.\projattr_{\tau'}(p_2)}\circ{}\\
  \nest_{\text{rel}2.\att(C') \rightarrow p} \circ{}\\
  \pi_{\att(Q)/\text{rel}1.\att(Q),\, p/\cond{(\text{noMatch}=\truevalue)}{\subrel()}{p}}\\
\end{array}\]
%
% \[\begin{array}{@{}l}
%   \pi_{\att(Q)/\text{rel}1.\att(Q),\, p/\cond{(p=\{(\text{rel2}.a:\nullvalue)\}_{a\in\att(C')})}{\subrel()}{p}}
%   (\nest_{\text{rel}2.\att(C') \rightarrow p} (\Rin
%   \leftouterjoin_{\Rin.\projattr(p_1)=C'.\projattr(p_2)} C'))
% \end{array}\]

\begin{lemma}\label{lem:maq2nra:lookup}
  Let $F$ be a forest of type $\tau$, $F'$ a forest of type $\tau'$,
  and $s=\lambda^{p_1=C'.p_2}_{p}$. Then
  $F \pipeline s[F'] \simeq \maqTonra{s} (\rel_\tau(F),
  \rel_{\tau'}(F'))$.
\end{lemma}

\begin{proof}
  Straightforward considering the semantics of lookup.
% \nb{GX: refine}
%   Given a tree $t\in F$ and a forest $F'$, there are two
%   possibilities:
%   \begin{itemize}
%   \item If $\eval[t]{p_1} = \eval[t']{p_2}$ for some $t'\in
%     F'$. For all such $t'$, in the corresponding left join,
%     $\row(t) \cup \row(t')$ (after some renaming) will be in the
%     result of the left join, and then following our translation,
%     $\row(t')$ is renamed and nested into the sub~relation $p$, and
%     correspond to the subtree in
%     $\attach(p,\colltoarray(F'\pipeline\match{p_2=v_1},\epsilon))$.
%   \item If $\eval[t]{p_1} \neq \eval[t']{p_2}$ for all $t'\in F'$. In the
%     corresponding left join, $\row(t) \cup \text{rel}2:null$ is in the
%     result. After the nesting step, it becomes
%     $\row(t) \cup p:\{\text{rel}2:\nullvalue\}$, and is changed to
%     $\row(t) \cup p: \{\}$, also corresponds to $\attach(p, \{\})$
%   \end{itemize}
\end{proof}

\begin{example}
  Consider a collection of type $\tau_\bios$ in Example~\ref{ex:rel-schema}.
  First, we provide the translation of some atomic criteria:

  {\small\begin{itemize}%[$\bullet$]
    \item $f$(name.first="Kristen") $=$ (name.first="Kristen")
    \item $f$(name=\lobject first: "Kristen"\robject) $=$ \falsevalue since
      according to $\tau_\bios$, the object under the key \valuefont{name}
      should contain also the key \valuefont{last}.
    \item \mbox{$f$(name=\lobject first:\,"Kristen", last:\,"Nygaard"\robject)\,$=$\,((name.first="Kristen")\,$\wedge$\,(name.last="Nygaard"))}
    \item for $\varphi=(\text{contribs=["OOP", "Simula"]})$, $f(\varphi)$
      computes a comparison between a sub-relation name and a relation value:
      (\text{contribs}=\{(contribs.\literalattr:\,"OOP"),
      (contribs.\literalattr:\,"Simula")\})
      % \item $f$(\text{contribs="OOP"}) $=$ (\text{"OOP"} $\in$
      %   \text{contribs.(contribs.\literalattr)})\nb{GX: fix this example or remove it}
    \end{itemize}
  }

  Second, we provide the translation of match stages for a criterion about a
  nested path and for a complex criterion:
  \[
\begin{array}{l}
  \maqTonra{\match{\text{contribs="OOP"}}} =
%(\text{"OOP"}   \text{contribs.(contribs.\literalattr)})\\
       \begin{array}[t]{@{}l}
         \pi_{\att(\bios),\, \text{contribs1/contribs}} \circ
         \unnest_{\text{contribs1}} \circ{}\\ %
         \pi_{\substack{\att(\bios),\,
         \text{contribs.lit1}/\text{contribs.lit}}} \circ{} \\
         \sigma_{(\text{contribs.lit1 = "OOP"})} \circ \pi_{\att(\bios)}
       \end{array}
\\
  \maqTonra{\match{(\text{awards.year}=2001)}} =
       \begin{array}[t]{@{}l}
         \pi_{\att(\bios),\, \text{awards1/awards}} \circ
         \unnest_{\text{awards1}} \circ{}\\ %
         \pi_{\substack{\att(\bios),\, \text{awards.award1}/\text{awards.award},\,% ~~\\\qquad\qquad\quad
         \text{awards.year1}/\text{awards.year}}} \circ{}\\
         \sigma_{(\text{awards.year1} = 2001)} \circ \pi_{\att(\bios)}
       \end{array}
       \\
       \maqTonra{\match{((\text{awards.year} \neq  1999) \,\lor\, (\text{awards.year} = 2000))}} = \\
       \quad\quad \begin{array}[t]{@{}l}
         \mathsf{AddID} \circ
         \pi_{\att(\bios),\, \mathit{ID},\, \text{awards1}/\text{awards},\, \text{awards2}/\text{awards}} \circ{}\\
         \unnest_{\text{awards1}} \circ %
         \pi_{\att(\bios),\, \mathit{ID},\, \text{awards2},\, \text{res}/(\text{awards.year}=1999)} \circ%{}\\ %
         \nest_{\{\text{res}\} \rightarrow \text{cond}} \circ{}\\ %
         \unnest_{\text{awards2}} \circ% {}\\
         \pi_{\substack{\att(\bios),\, \mathit{ID},\, \text{cond},\, \text{awards.award2}/\text{awards.award},\,% ~\\\qquad\qquad\qquad\qquad
                    \text{awards.year2}/\text{awards.year}}} \circ{}\\
         \sigma_{(\text{cond}=\{(\text{res}: \falsevalue)\}) \lor (\text{awards.year2}=2000)} \circ%
         \pi_{\att(\bios)}
       \end{array}
     \end{array}
  \]
  where for the first two stage we omitted the creation of the identifier
  column.  Finally, we provide the translation of a group stage:
  % \nobrackettag
  \[
    \begin{array}[b]{l}
      \maqTonra{\group{\text{year/awards.year}}{\text{persons/name}}} = \\
      \quad\quad
      \begin{array}[b]{@{}l}
        \pi_{\text{\_id.year/awards.year,
          persons.first/name.first, persons.last/name.last}}  \circ{} \\
        \nest_{\{\text{persons.first, persons.last}\} \rightarrow \text{persons}}\\
      \end{array}
      \hspace{3.2cm}\qed
    \end{array}
  \]
\end{example}

\begin{theorem}\label{thm:maq-to-nra-stage}
  Let $F$ be a forest of type $\tau$ and $s$ a stage of \mquery, then
  $F \pipeline s \simeq \maqTonra{s} (\rel_\tau(F)) $ if $s$ is not a lookup
  stage; otherwise
  $F \pipeline s[F'] \simeq \maqTonra{s} (\rel_\tau(F), \rel_{\tau'}(F'))$ for a forest $F'$ of type $\tau'$.
\end{theorem}

\begin{proof}
  This follows from Lemmas~\ref{lem:maq2nra:match} to \ref{lem:maq2nra:lookup}.
\end{proof}

\begin{theoremnum}{\ref{thm:maq-to-nra-pipeline}}
  Let $\S$ be a set of type constraints, $\q$ an \mquery
  $C \pipeline s_1 \pipeline \cdots \pipeline s_m$ in which each stage is
  well-typed for its input type, and
  $Q = C\circ \maqTonra{s_1}\circ \cdots \circ\maqTonra{s_m}$.  Then
  $\q \equiv_\S Q$, moreover, the size of $Q$ is polynomial in the size of $\q$
  and $\S$. 
\end{theoremnum}

\begin{proof}
  Let $D$ be a \mongodb database satisfying a set $\S$ of constraints,
  $ F_i= \ansmongo(C \pipeline s_1 \pipeline \cdots \pipeline s_i, D)$, and
  $\R_i=\ansra(C\circ \maqTonra{s_1}\circ \cdots \circ\maqTonra{s_i},
  \rdb_{\S}(D))$. We can easily show $F_i\simeq \R_i$, for all $i\in[1..m]$ by
  induction on $i$.
  %
  % We prove by induction on $m$, the length of pipeline. The base case ($m=1$)
  % has been already shown in Theorem~\ref{thm:maq-to-nra-pipeline}.
  % Now let's assume the theorem holds for $m-1$, i.e.
  % $F_{m-1} \simeq \R_{m-1}$, and $F_{m-i}$ is of some type
  % $\tau_{m-1}$. Applying Theorem~\ref{thm:maq-to-nra-stage}, we
  % have $F_{m-1} \pipeline s_m \simeq \maqTonra{s_m}(\R_{m-1})$, and
  % therefore, $F_m \simeq \R_m$.
  
  As for the size of $Q$, it is easy to see that for each stage $s$, the size
  of $\maqTonra{s}$ is polynomial in the size of $s$ and $\S$. Moreover, the
  size of $Q$ is linear in the sum of the sizes of $\maqTonra{s_i}$, for
  $i=1,\dots,m$.
\end{proof}

%\endinput

%\input{8-a-mquery-to-nra}

\section{Complexity of MQuery}

\begin{theoremnum}{\ref{lem:mupg-nexptime-complete}}
  \mupg and \mupgl are \TAexppoly-complete in combined complexity
  and in \ACz in data complexity.
\end{theoremnum}

\begin{lemma}
  \label{lem:TAexppoly-lower-bound}
  \mupg is \TAexppoly-hard in combined complexity.
\end{lemma}
\begin{proof}
  We adapt the proof of \TAexppoly-hardness from \cite{Koch06}.

  Let $M = (\Sigma, Q, \delta, q_0, F)$ be an alternating Turing machine that
  runs in time $2^{p_1(n)}$ with $p_2(n)$ alternations on inputs of size $n$,
  where $\Sigma$ the tape alphabet, $Q$ is the set of states partitioned into
  existential $Q_\exists$ and universal $Q_\forall$ states,
  $\delta: Q \times \Sigma \times \{1,2\} \to Q \times \Sigma \times
  \{-1,0,+1\}$
  the transition function, which for a state $q$ and symbol $s$ gives two
  instructions $\delta(q,s,1)$ and $\delta(q,s,2)$, $q_0$ the initial state and
  $F \subseteq Q$ the set of accepting states.

  Following Koch, we simulate the computation of $M$ in \mupg.  Each run of $M$
  is a tree of configurations of depth bounded by $p_2(n) \cdot 2^{p_1(n)}$,
  and each configuration consists of a tape of length bounded by $2^{p_1(n)}$,
  a current state and a position marker on the tape. We construct an \mupg $\q$
  and a forest $F$ such that $F \pipeline \q$ is non-empty iff $M$ accepts its
  input. $F$ consists of a single document containing the key-value pair \id: 1.

  \begin{itemize}
  \item The tape of a configuration is modeled as a nested object of nested
    depth $p_1(n)$ and with $2^{p_1(n)}$ leaves.  The position of the head on
    the tape is represented by an extended tape alphabet
    $\Sigma' = \Sigma \cup \{\bar{s} \mid s \in \Sigma\}$. That is, the symbol
    $\bar{s}$ in a tape cell indicates that the cell stores symbol $s$ and it
    is the current position of the head. The following is a valid tape:
    \begin{center}
      \valuefont{"tape": \{"l": \{"l": "$\bar{\texttt{0}}$", "r": "0"\}, "r":
        \{"l": "\#", "r": "\#"\}\}}.
    \end{center}
    We can compute the set of all $m^{2^{p_1(n)}}$ tapes (including non-valid
    ones) by the query:
    \[\left.
      \begin{array}{@{}r} \mathsf{Tapes} = \project{\text{tape}/[\Sigma']}
        \pipeline \project{l/\text{tape},\, r/\text{tape}} \pipeline %
        \unwind{l} \pipeline \unwind{r} \pipeline%
        \project{\text{tape}.l/l,\, \text{tape}.r/r} \pipeline%
        \group{}{\text{tape}} \pipeline{}\\
        \dots \mbox{\hspace{3cm}}\\
        \project{l/\text{tape},\, r/\text{tape}} \pipeline %
        \unwind{l} \pipeline \unwind{r} \pipeline%
        \project{\text{tape}.l/l,\, \text{tape}.r/r} \pipeline%
        \group{}{\text{tape}}  \mbox{\hspace{0.34cm}}\\
      \end{array}\right\}p_1(n) \text{ times}
    \]
    Where for a set $S$, the value definition $[S]$ means a constant array
    consisting of all elements in $S$ (we view everything as strings). The
    result of this query (on F) is a single document containing an array of
    all possible tapes under the key \valuefont{tape}.

  \item In turn, a configuration is a pair, consisting of a tape and a
    state. We can compute all possible (including non-valid ones)
    configurations by the following \mupg:
    \[\mathsf{Configs} = \mathsf{Tapes} \pipeline %
    \project{\text{tape},\, \text{state}/[Q]} \pipeline%
    \unwind{\text{tape}} \pipeline \unwind{\text{state}} \pipeline%
    \project{\text{c.tape}/\text{tape},\, \text{c.state}/\text{state}}\]
    The result of $\mathsf{Configs}$ is a set of trees, each containing one
    possible configuration under the key \valuefont{c}.

  \item Next, we are going to construct a query that computes the pairs of
    configurations $c_1$ and $c_2$ such that $c_2$ is a possible immediate
    successor of $c_1$ according to $\delta$ (also including pairs of non-valid
    configurations).
    First, we create all pairs of configurations $c_1$ and $c_2$, and make
    working copies $w_1$ and $w_2$ of the tapes.
    \[\mathsf{Prepare\text{-}succ} = \mathsf{Configs} \pipeline %
    \group{}{c1/c,\, c2/c} \pipeline \unwind{c1} \pipeline \unwind{c2}
    \pipeline%
    \project{\text{succ}/\{c1/c1,\, c2/c2\},\, w1/c1.\text{tape},\, w2/c2.\text{tape}}\]

    Second, to check that $c_1$ is a possible successor of $c_2$, we verify
    that $w_1$ and $w_2$ differ at at most two consecutive tape positions. The
    tapes are of exponential length, but we can find these two positions by
    doing a number of checks that is equal to the depth of the value encoding a
    tape minus 1.  Namely, we iteratively compare the halves of the working
    copies, and in the next step the working copies become the halves which are
    not equal (see \cite{Koch06} for more details):
    \begin{itemize}
    \item If $w_1.l = w_2.l$ (the left halves of the tapes are equal), we
      replace $w_1$ by $w_1.r$ and $w_2$ by $w_2.r$
    \item If $w_1.r = w_2.r$ (the right halves of the tapes are equal), we
      replace $w_1$ by $w_1.l$ and $w_2$ by $w_2.l$
    \item If $w_1.l.l = w_2.l.l$ and $w_1.r.r = w_2.r.r$ (the left and the
      right quarters of the tapes are equal, so the difference should be in the
      ``inner'' part of the tree), we replace $w_1.l$ by $w_1.l.r$, $w_1.r$ by
      $w_1.r.l$ and $w_2.l$ by $w_2.l.r$, $w_2.r$ by $w_2.r.l$
    \end{itemize}
    We implement zooming-in by the query $\mathsf{Zoom\text{-}in} ={}$
    \[
    \begin{array}[t]{@{}l}
      \project{\substack{\text{succ},\,
      w1/\cond{(w1.l=w2.l)}{w1.r}{%
      \cond{(w1.r=w2.r)}{w1.l}{%
      \cond{((w1.l.l=w2.l.l) \land
      (w1.r.r=w2.r.r))}{\{l/w1.l.r,r/w1.r.l\}}{\nullvalue}}}\\
      \quad w2/\cond{(w1.l=w2.l)}{w2.r}{%
      \cond{(w1.r=w2.r)}{w2.l}{%
      \cond{((w1.l.l=w2.l.l) \land
      (w1.r.r=w2.r.r))}{\{l/w2.l.r,r/w2.r.l\}}{\nullvalue}}}}}
      \pipeline{} \\
      \match{\neg(w1=\nullvalue)}
    \end{array}
    \]

    After finding the two positions where the tapes differ, we check that the
    head is over one of these positions.
    \[\mathsf{Head} = \match{\bigvee_{s\in\Sigma} ((w1.l=\bar{s}) \lor (w1.r=\bar{s}))}\]

    Then, we check that the difference is according to the transition function
    $\delta$. Let criterion $\varphi_\delta$ be the disjunction of the
    following formulas $\varphi_{q,s,q',z,l}$, for each instruction
    $\delta(q,s,i) = (q',z,l)$:
    \[
    \begin{array}{r@{}l}
      \varphi_{q,s,q',z,0} ={}
      &
        \begin{array}[t]{@{}l}
          (\text{succ}.c1.state=q) \land (\text{succ}.c2.state=q') \land{}\\
          ((w1.l=\bar{s}) \land (w2.l=\bar{z}) \land
          \bigvee_{b \in \Sigma} ((w1.r=b) \land (w2.r=b)) \lor{}\\
          ~(w1.r=\bar{s}) \land (w2.r=\bar{z}) \land
          \bigvee_{b \in \Sigma} ((w1.l=b) \land (w2.l=b)))\\
        \end{array}
      \\
      \varphi_{q,s,q',z,+1} ={}
      &
        \begin{array}[t]{@{}l}
          (\text{succ}.c1.state=q) \land (\text{succ}.c2.state=q') \land{}\\
          (w1.l=\bar{s}) \land (w2.l=z) \land \bigvee_{b \in \Sigma}
          ((w1.r=b) \land (w2.r=\bar{b}))
        \end{array}
      \\
      \varphi_{q,s,q',z,-1} ={}
      &
        \begin{array}[t]{@{}l}
          (\text{succ}.c1.state=q) \land (\text{succ}.c2.state=q') \land{}\\
          (w1.r=\bar{s}) \land (w2.r=z) \land \bigvee_{b \in \Sigma}
          ((w1.l=b) \land (w2.l=\bar{b}))
        \end{array}
      \\
    \end{array}
    \]
    Finally, the query $\mathsf{Succ}$ that computes pairs of successor
    configurations is:
    \[\mathsf{Succ}=\mathsf{Prepare\text{-}succ}\pipeline %
    \underbrace{\mathsf{Zoom\text{-}in}\pipeline\cdots\pipeline
      \mathsf{Zoom\text{-}in}}_{p_1(n)-1 \text{ times}}%
    \pipeline~ \mathsf{Head} \pipeline \match{\varphi_\delta}
    \pipeline{} \project{\text{succ}}
    \]

  \item To encode alternations, we first need to compute computation paths of
    length up to $2^{p_1(n)}$ that we represent by pairs $(c_1,c_2)$: $c_2$ is
    reachable from $c_1$ in at most $2^{p_1(n)}$ steps, moreover if the state
    of $c_1$ is existential, then each of the intermediate configurations
    before reaching $c_2$ must be existential, and likewise if the state of
    $c_1$ is universal. We implement ``at most'' by means of the ``stay
    transitions'' $(c,c)$ added to $\mathsf{Succ}$.  We compute these
    computation paths iteratively:
    \[
    \begin{array}{l}
      \mathsf{CP}_0 = \mathsf{Succ}\\
      \mathsf{CP}_{i+1} =
      \begin{array}[t]{@{}l}
        \mathsf{CP}_i \pipeline \group{}{s1/\text{succ},\, s2/\text{succ}}
        \pipeline \unwind{s1} \pipeline \unwind{s2} \pipeline%
        \match{s1.c2=s2.c1} \pipeline{}\\ %
        \match{(s1.c1.\text{state}\in Q_\exists) \leftrightarrow (s2.c1.\text{state}\in Q_\exists)} \pipeline%
        \project{\text{succ}/\{c1/s1.c1,\, c2/s2.c2\}}
      \end{array}
    \end{array}
    \]
    where $(p\in A)$, for a set $A$, is a shortcut for $\bigvee_{a\in A}
    (p=a)$, and $\varphi_1 \leftrightarrow \varphi_2$ is a shortcut for $(\neg
    \varphi_1 \lor \varphi_2) \land (\neg \varphi_2 \lor \varphi_1)$.

    We can now compute the sets $A_i$ of configurations that lead to an
    accepting state in $i$ alternations:
    \[
    \begin{array}{l}
      A_1 =
      \begin{array}[t]{@{}l}
      \mathsf{CP}_{p_1(n)} \pipeline%
      \group{}{\text{cp}/\text{succ}} \pipeline %
      \project{\text{cp},\, s/\text{cp}} \pipeline \unwind{s} \pipeline%
      \match{(s.c2.\text{state} \in F) \land (s.c1.\text{state} \in Q_\exists)} \pipeline{}\\%
      \project{\text{cp},\, a/s.c1} \pipeline%
      \group{\text{cp}}{a} \pipeline \project{\text{cp}/\id.\text{cp},\, a}
      \end{array}
      \\
      A_{i+1} =
      \begin{array}[t]{@{}l}
        A_i \pipeline %
      \project{\text{cp},\, a,\, s/\text{cp}} \pipeline \unwind{s} \pipeline
      \unwind{a} \pipeline%
      \project{\text{cp},\, s,\, \text{inAi}/(s.c2=a)} \pipeline%
      \group{\text{cp},\, s}{\text{inAi}} \pipeline%
      \match{\text{inAi}=[\falsevalue]} \pipeline{}\\%
      \project{\text{cp}/\id.\text{cp},\, s/\id.s} \pipeline%
      \match{(s.c1.\text{state}\in Q_\exists) \leftrightarrow (s.c1.\text{state}\in Q_\exists)}\pipeline%
      \project{\text{cp},\, a/s.c1} \pipeline%
      \group{\text{cp}}{a} \pipeline \project{\text{cp}/\id.\text{cp},\, a}
      \end{array}
    \end{array}
    \]

    \item Finally, we check that the initial computation is in $A_{p_2(n)}$.
    The initial configuration has a tape, where the input string $w$ of
    length $n$ is padded with $2^{p_1(n)}-n$ \#-symbols. Let $v_w$ be the
    nested value of depth $\lceil \log_2 n \rceil$ representing $w$ padded with
    $2^{\lceil \log_2 n \rceil} -n$ \#-symbols (it can be computed in
    \LOGSPACE). Then the initial configuration can be computed, and checked
    whether in $A_{p_2(n)}$ as follows:
    \[
    \begin{array}{l}
      C_0\mathsf{in}A_{p_2(n)} =
      \begin{array}[t]{@{}l}
        A_{p_2(n)} \pipeline \project{a,\, \text{tape}.l/v_w,\, \text{tape}.r/\text{"\#"}} \pipeline{}\\ %
        \underbrace{%
        \project{a,\, \text{tape}.l,\, \text{tape}.r/\{l/\text{tape}.r,\, r/\text{tape}.r\}}\pipeline\cdots\pipeline\project{a,\, \text{tape}.l,\, \text{tape}.r/\{l/\text{tape}.r,\, r/\text{tape}.r\}}%
        }_{\lceil\log_2n\rceil \text{ times}} \pipeline{} \\%
        \underbrace{%
        \project{a,\, \text{tape}.l/\text{tape},\, \text{tape}.r/\{l/\text{tape}.r ,\, r/\text{tape}.r\}} %
        \pipeline \cdots \pipeline%
        \project{a,\, \text{tape}.l/\text{tape},\, \text{tape}.r/\{l/\text{tape}.r ,\, r/\text{tape}.r\}} %
        }_{p_1(n)-\lceil\log_2n\rceil-1 \text{ times}} \pipeline{}\\%
        \project{a,\, c0.\text{tape}/\text{tape},\,
        c0.\text{state}/\text{"q0"}}\pipeline%
        \unwind{a} \pipeline \match{a=c0}
      \end{array}
    \end{array}
  \]
  (where in project $\text{tape}.r/\{l/\text{tape}.r ,\, r/\text{tape}.r\}$
  can be seen as a shortcut for
    $\text{tape}.r.l/\text{tape}.r ,\, \text{tape}.r.r/\text{tape}.r$. In fact,
    it is possible to write such a value definition in \mongodb. We will use
    this syntax for brevity also in what follows.)
    The part of the query that computes the initial configuration takes $v_w$,
    pads it with \#-symbols so as to have \valuefont{tape} having the value of
    depth $\lceil \log_2 n \rceil + 1$ where \valuefont{tape.r} consists
    entirely of \#'s. Then in the second line it increases the depth of the
    \valuefont{tape} value to $p_1(n)$ by iteratively assigning the previous
    value of \valuefont{tape} to \valuefont{tape.l} and duplicating the value
    of \valuefont{tape.r} to \valuefont{tape.r.l} and \valuefont{tape.r.r}. See
    \cite{Koch06} for more details.
  \end{itemize}
    Thus we obtain that $\{\tree(\lobject \id:1\robject)\} \pipeline
    C_0\mathsf{in}A_{p_2(n)}$ is non-empty iff $M$ accepts $w$.
\end{proof}

\begin{lemma}
  \mupgl is in \TAexppoly in combined complexity.
\end{lemma}
\begin{proof}
  Let $\q = C \pipeline s_1 \pipeline \cdots \pipeline s_n$ be an \mupgl query,
  and $D$ a \mongodb instance.  We provide an algorithm to check that
  $\ansmongo(\q,D)$ is non-empty.

  We assume that $\q$ is of the following form:
  \begin{itemize}
  \item we consider atoms of the form $p=v$, $p\neq v$, $p=p$, $p\neq p$,
    $\exists p$, $\neg\exists p$, and assume that the criteria in match stages,
    Boolean value definitions and conditions in conditional value definitions
    are monotone Boolean expressions over such atoms (i.e., use only
    conjunction and disjunction).

  \item each project stage is of the form
    $\project[\mathit{ni}]{p_1,\dots,p_n,\, q/d}$, that is, contains at most
    one projection element defining the value of a path $q$. We can achieve it
    for an arbitrary project stage
    $\project[\mathit{ni}]{p_1,\dots,p_n,\,q_1/d_1,\dots,q_m/d_m}$ by splitting
    it into $m$ project stages:
    $\project[\mathit{ni}]{p_1,\dots,p_n,\, q_1/d_1} \pipeline
    \project{p_1,\dots,p_n,\, q_1,\, q_2/d_2} \pipeline \cdots \pipeline
    \project{p_1,\dots,p_n,\, q_1,\dots,q_{m-1},q_m/d_m}$.
  \end{itemize}

  Let $\ansmongo(\q,D) = F_n$. The algorithm is to check whether there is a
  tree in $F_n$. We do it recursively as follows. Assume that $F \pipeline s =
  F'$ and we want check whether a tree satisfying a set $\psi$ of atoms (of the
  considered form) is in $F'$. Then, the check amounts to the following:
  \begin{itemize}
  \item if $s = \match{\varphi}$, then we guess atoms $e_1,\dots,e_m$ appearing
    in $\varphi$ so that assigning them the true value makes $\varphi$ true,
    add to $\psi$ the conditions $e_1,\dots,e_n$. If the new conditions $\psi'$
    are consistent, we check whether there is a tree satisfying $\psi'$ in
    $F$. Otherwise, we report a failure.

  \item if $s = \unwind{p}$, then we replace $p$ by $p.i$ in all conditions in
    $\psi$ about $p$ and check whether there is a tree with the new conditions
    $\psi'$ in $F$.

  \item if $s = \unwind[+]{p}$, then we guess whether with or without index,
    and in the former case we replace $p$ by $p.i$ in all conditions in $\psi$
    about $p$, in the latter $\psi$ is not changed. Then we check whether there
    is a tree with the new conditions $\psi'$ in $F$.

  \item if $s = \project{p_1,\dots,p_l,\, q/d}$, we do not do anything for
    $p_i$. As for $q/d$, we remove $\exists q$ if it is in $\psi$, and
    proceed as follows:
    \begin{enumerate}
    \item if $d$ is a path $q'$, we replace each occurrence of $q$ in $\psi$ by
      $q'$.

    \item if $d$ is a constant non-array value,
      \begin{enumerate}
      \item if there is a condition of the form $q=v$ in $\psi$, we remove it
        from $\psi$, check whether $v=d$, and if not, we report failure.
      \item if there is a condition of the form $q\neq v$ in $\psi$, we remove it
        from $\psi$, check whether $v\neq d$, and if not, we report failure.

      \item if there is a condition of the form $p'=v'$ in $\psi$, for a prefix
        $p'$ of $q$, we extract the value $v$ for $p$ (it will be the value of
        the subtree in $v$ reachable by path $q'$ such that $p'.q'=q$) if it is
        possible and proceed as in (a), otherwise we report a failure.

      \item if there is a condition of the form $q.p'=v'$ or $q.p'\neq v'$ in
        $\psi$, we extract the value definition $d'$ from $d$ reachable by path
        $p'$ if it is possible and proceed as in the case $q.p'/d'$, otherwise
        we report a failure.

      \item if there is a condition of the form $q.i=v'$ in $\psi$, we report a
        failure.
      \end{enumerate}

      \item if $d$ is a Boolean value definition
      \begin{enumerate}
      \item if there is a condition of the form $q=v$ (resp., $q\neq v$) in
        $\psi$, we remove it from $\psi$. Then we check whether $v$ is
        \truevalue or \falsevalue, if not, we report failure. Otherwise, we
        guess atoms $e_1,\dots,e_n$ appearing in $d$ so that $d$ evaluates to
        $v$ (resp., to the negation of $v$) under assigning the atoms $e_i$ the
        true value, and add to $\psi$ the conditions $e_1,\dots,e_n$.

      \item[(c)] analogous to 2.(c)

      \item[(d)] analogous to 2.(d)

      \item[(e)] analogous to 2.(e)
      \end{enumerate}

      \item if $d=[d_1,\dots,d_k]$,
      \begin{enumerate}
      \item if there is a condition of the form $q=v$ in $\psi$, we remove it
        from $\psi$, check whether $v$ is of the form $[v_1,\dots,v_k]$, if
        not, we report failure. Otherwise we guess the pairs $(d_i,v_j)$, and
        for each pair $(d_i,v_j)$ we break it down and proceed similarly the
        case as if we had $q'/d_i$ and a condition $q'=v_j$ in $\psi$.

      \item if there is a condition of the form $q\neq v$ in $\psi$, we remove it
        from $\psi$, check whether $v\neq d$, and if not, we report failure.

      \item if there is a condition of the form $p'=v'$ in $\psi$, for a prefix
        $p'$ of $q$, we extract the value $v$ for $p$ (it will be the value of
        the subtree in $v$ reachable by path $q'$ such that $p'.q'=q$) if it is
        possible and proceed as in (a), otherwise we report a failure.

      \item if there is a condition of the form $q.p'=v'$ or $q.p'\neq v'$ in
        $\psi$, we guess $d_i$, extract the value definition $d'$ from $d_i$
        reachable by path $p'$ if it is possible and proceed as in the case
        $q.p'/d'$, otherwise we report a failure.

      \item if there is a condition of the form $q.i=v'$ in $\psi$, we guess
        $d_i$ and proceed similarly to the case as if we had $q'/d_i$ and a
        condition $q'=v'$ in $\psi$.
      \end{enumerate}

    \item if $d$ is a conditional value definition $\cond{c}{d_1}{d_2}$, we
      guess atoms $e_1,\dots,e_n$ appearing in $c$ so that
      $\psi \cup \{e_1,\dots,e_n\}$ is consistent, and if $c$ evaluates to true
      under assigning the atoms $e_i$ the true value, then we consider the
      inductive case when $d$ is $d_1$, otherwise when $d$ is $d_2$. In any
      case, we add to $\psi$ the conditions $e_i$.

    \end{enumerate}
    Then we check whether the new conditions $\psi'$ are consistent. If not, we
    report failure. Otherwise we check whether there is a tree satisfying
    $\psi'$ in $F$.

  \item $s = \group{}{a_1/b_1,\dots,a_n/b_n}$. If there is a condition
    $\id\neq\nullvalue$ in $\psi$, we report failure. Otherwise, we remove all
    conditions on \id from $\psi$, and for each $a_j/b_j$ we proceed as
    follows:
    \begin{itemize}
    \item if $\psi$ contains a condition of the form $a_j=[]$, then we check
      that there is no tree satisfying $\exists b_j$ in $F$.
    \item if $\psi$ contains a condition of the form $a_j\neq[]$, then we check
      that there is a tree satisfying $\exists b_j$ in $F$.
    \item if $\psi$ contains a condition of the form $a_j=[v_1,\dots,v_k]$,
      $k>0$, then for each $i=1,\dots,k$, we check whether each tree in $F$
      satisfies $b_j\in\{v_1,\dots,v_k\}$.
    \item if $\psi$ contains a condition of the form $a_j\neq[v_1,\dots,v_k]$,
      $k>0$, then we guess either ``subset'' or ``superset'', in the former
      case we guess a subset $[u_1,\dots,u_m]$ of $[v_1,\dots,v_k]$ and check
      whether each tree in $F$ satisfies $b_j \in\{u_1,\dots,u_m\}$ if $m>0$,
      or whether each tree in $F$ satisfies $\neg\exists b_j$ if $m=0$; in the latter
      case we check whether there is a tree satisfying
      $\{(b_j\neq v_1),\dots,(b_j\neq v_k)\}$ in $F$.
    \item if $\psi$ contains a condition of the form $a_j=v$ for a non-array
      value, then we report failure.
    \item if $\psi$ contains a condition of the form $a_j.i=v$, then we replace
      it by $b_j=v$ and check whether there is a tree with the new conditions
      $\psi'$ in $F$.
    \item if $\psi$ contains a condition $\exists a_j$, then it is satisfied
      and can be removed. We check whether there is a tree in $F$.
    \end{itemize}

  \item $s = \group{g_1/y_1,\dots,g_m/y_m}{a_1/b_1,\dots,a_n/b_n}$. We replace
    all conditions of the form $\id.g_i=v$ in $\psi$ by $y_i=v$. For each
    $a_j/b_j$ by analogy with group by \nullvalue, where also need to take into
    account conditions on $\id.g_i$.

  \item $s = \lookup{p_1 = C_2.p_2}{p}$.
    \begin{itemize}
    \item if there is a condition of the form $p=v$ in $\psi$, we remove it
      from $\psi$ and
      \begin{itemize}
      \item if $v$ is not an array we report a failure.
      \item if $v=[]$
        \begin{itemize}
        \item    if there is a condition of the form $p_1=v'$ in $\psi$, we
          check whether there is a tree with $p_2=v'$ in $D.C_2$. If yes, we
          report failure. Otherwise, we check
          whether there is a tree with the new conditions $\psi'$ in $F$.
        \item otherwise, let $v_1,\dots,v_k$ be all values of $p_2$ in
          $D.C_2$. We add conditions $p_1\neq v_i$, for $i=1,\dots,k$, to
          $\psi$, and check whether there is a tree satisfying the new
          conditions $\psi'$ in $F$.
        \end{itemize}
      \item otherwise, $v=[v_1,\dots,v_k]$, $k>0$, and we check whether there
        are trees $\tree(v_1),\dots,\tree(v_k)$ in $D.C_2$. If not, we report
        failure. If yes, we check that the trees agree on the value $v''$ of
        $p_2$. If yes,
        \begin{itemize}
        \item    if there is a condition of the form $p_1=v'$ in $\psi$: if
          $v'\neq v''$, we report failure, otherwise we check whether there is
          a tree with the new conditions $\psi'$ in $F$.
        \item otherwise we add a condition $p_1=v''$ to $\psi$ and check
          whether there is a tree with the new conditions $\psi'$ in $F$.
        \end{itemize}
      \end{itemize}

    \item if $\psi$ contains $\exists p$, we remove it and check whether there
      is a tree satisfying the new conditions $\psi'$ in $F$.
    \end{itemize}
  \end{itemize}
  Once we reach the first stage, then we directly check whether there is a tree
  in $D.C$ satisfying the conditions, or whether all trees in $D.C$ satisfy the
  conditions.

  By analysing how we deal with various stages, we can see that both branching
  and alternations occur only because of the group stages. %
  The overal algorithm works in alternating exponential time with a polynomial
  (actually, linear) number of alternations: the ``depth'' of the checks is
  given by the number of stages, the branching and the number of alternations
  are bounded by the size of $\q$.
\end{proof}

\begin{corollarynum}{\ref{cor:NRA-complexity}}
  NRA is \TAexppoly-complete in combined complexity.
\end{corollarynum}
\begin{proof}
For the lower-bound, see~\cite{Koch06}.  The upper bound follows from
Theorems~\ref{thm:nra-to-mupgl-correct} and~\ref{lem:mupg-nexptime-complete}.
\end{proof}

\begin{lemmanum}{\ref{lem:mq-logspace-complete}}
  Boolean query evaluation for \mq queries is \LOGSPACE-complete in combined
  complexity.
\end{lemmanum}
\begin{proof}
  First, we prove the upper bound.
  Let $D$ be a \mongodb database, and $\q$ an \mq query of the form $C
  \pipeline \match{\varphi}$, where $\varphi$ is a criterion. We can view
  $\varphi$ as a Boolean formula constructed using the connectors $\land$,
  $\lor$ and $\neg$ starting from the atoms of the form $(p~\mathbf{op}~v)$ and
  $\exists p$, where $p$ is a path, $v$ a literal value, and $\mathbf{op}$ is a
  comparison operator.
  Given a tree $t$ and an atom $\alpha$ of the above form, we can check in
  \LOGSPACE whether $t \models \alpha$: for each node $x$ in $t$, we can check
  in \LOGSPACE if $\mathsf{path}(x,t) = p$ and we can check in \LOGSPACE if
  $\lnode(x) = v$.

  Now, we define a \LOGSPACE reduction from the problem of whether
  $\ansmongo(\q,D) \neq \emptyset$ to the problem of determining the truth
  value of a variable-free Boolean formula, known to be A\LOGTIME-complete
  \cite{Buss87}. We construct a Boolean formula $\psi$ as the disjunction of
  $\varphi_t$ for each $t \in D.C$, where $\varphi_t$ is a copy of $\varphi$,
  where each atom $\alpha$ is substituted with $1$ if $t \models \alpha$ and
  with $0$, otherwise. Then $\ansmongo(\q,D) \neq \emptyset$ iff the value of
  $\psi$ is true.

  \smallskip%
  We show the lower bound by \NCone reduction from the directed forest
  accessibility (DFA) problem known to be complete for \LOGSPACE under \NCone
  reducibility \cite{CoMc87}. The DFA problem is, given an acyclic directed
  graph $G$ of outdegree zero or one, nodes $u$ and $v$, to decide whether
  there is a directed path from $u$ to $v$.

  Let $G = (V, T)$, $u,v \in V$ such that $G$ has precisely two weakly
  connected components, $u$ has indegree 0 and $v$ has outdegree~0: the lower
  bound still holds in this case. Let $v'$ be the other vertex in $G$ with
  outdegree 0. We construct a tree $t=(N,E,\lnode,\ledge)$ and a path $p$ such
  that $t \models (\exists p)$ iff there is a directed path from $u$ to $v$ in
  $G$. We add a fresh node $r$ that will be the root of the tree with two
  children $v$ and $v'$, and a fresh node $l$ that will be the only child of
  $u$, also we invert all edges in $G$: $N = V \cup \{r,l\}$, $E = T^- \cup
  \{(r,v), (r,v'), (u,l)\}$.  Then we set $\ledge(r,v)=a$, $\ledge(r,v')=c$,
  $\ledge(u,l)=b$, and the rest of the edges is labeled by index 0. The node
  labels are set as $\lnode(r) = \objectlabel$, $\lnode(u) = \objectlabel$ and
  the rest of the nodes are labeled with $\arraylabel$.

  Now, the obtained tree $t$ is not a valid tree according to our definition of
  a tree, as the children of array nodes are not labeled by distinct
  indexes. However, by inspecting the semantics of $\eval{p}$, we see that $t
  \models (\exists p)$ iff $t' \models (\exists p)$, where $t'$ is the version
  of $t$ with all distinct indexes.
  Thus, we obtain that $t \models (\exists a.b)$ iff there is a directed path
  from $u$ to $v$ in~$G$.
\end{proof}

The project operator allows one to create new values by duplicating the
existing ones; hence, it can make trees grow exponentially in the size of the
query, and similarly with the group operator.  Nevertheless, we can still check
whether the answer to a query is non-empty in polynomial time by reusing the
``old'' tree nodes when it is necessary to duplicate values.
\begin{lemmanum}{\ref{lem:mpg-ptime-complete}}
  Query evaluation for $\mpgl$ queries is \PTIME-complete.
\end{lemmanum}

\begin{lemma}\label{lem:mp-ptime-hard}
  The query emptiness problem for $\mp$ queries is \PTIME-hard in
  combined complexity.
\end{lemma}
\begin{proof}
  The proof by a straightforward reduction from the Circuit Value problem,
  known to be \PTIME-complete. % \cite[Theorem~8.1]{Papa94}.
  For completeness, we provide the reduction. Given a monotone Boolean circuit
  $\C$ consisting of a finite set of assignments to Boolean variables
  $X_1,\ldots,X_n$ of the form $X_i=0$, $X_i=1$, $X_i = X_j \land X_k$, $j,k <
  i$, or $X_i = X_j \lor X_k$, $j,k < i$, where each $X_i$ appears on the
  left-hand side of exactly one assignment, check whether the value $X_n$ is
  $1$ in $\C$.

  We construct a query $\q$ such that on each non-empty forest $F$, $F
  \pipeline \q$ is non-empty iff the value $X_n$ is $1$ in $\C$. We set $\q =
  s_1 \pipeline \cdots \pipeline s_n \pipeline \match{\text{x}n=1}$, where for
  $i\in\{1,\dots,n\}$, $s_i =
  \project{\text{x}1,\dots,\text{x}i-1,~\text{x}i/\textit{ass}_i}$, where
  $\textit{ass}_i = v$, if $X_i = v$ for $v \in \{0,1\}$, $\textit{ass}_i =
  \text{x}j \land \text{x}k$, if $X_i = X_j \land X_k$, and $\textit{ass}_i =
  \text{x}j \lor \text{x}k$, if $X_i = X_j \lor X_k$.
\end{proof}

\begin{lemma}
  The query emptiness problem for $\mpgl$ queries is in \PTIME in combined
  complexity.
\end{lemma}
\begin{proof}
  We provide a \PTIME algorithm for checking whether, given an \mpgl $\q$ (over
  collection $C$), a forest $F_0$ for $C$, and forests $G_{C'}$ for each
  external collection $C'$ used by $\q$, $F_0 \pipeline \q$ is non-empty.

  %
  % Without loss of generality, we can assume that $\q$ consists of subqueries
  % $\q_1, \dots, \q_n$, where each $\q_i$ is a pipeline of project and group
  % stages ending with a match stage. The high-level idea of our algorithm is as
  % follows: we compute forests $F_1,\dots,F_n$ where $F_i$ is the smallest
  % subset of $F$ such that $F_i \pipeline \q_1\pipeline\cdots\pipeline\q_i =
  % F\pipeline \q_1\pipeline\cdots\pipeline\q_i$, and then check that $F_n$ is
  % non-empty.
  %
  % First of all, we construct an intermediate graph structure $G_{\q}$ for $\q$
  % that helps us to represent the shape of the intermediate results and to
  % efficiently check the match criteria or Boolean value definitions in project
  % stages. Let $\q = s_1 \pipeline \cdots \pipeline s_m$. The nodes of the graph
  % $G_{\q}$ are partitioned into $m+1$ layers $\ell_0, \dots, \ell_m$, where
  % $\ell_0$ corresponds to the input forest and $\ell_i$ corresponds to $s_i$,
  % for $i=1..m$. Formally,
  %
  The algorithm computes the result of each stage by representing the
  intermediate trees as DAGs in order to avoid exponential growth of trees that
  is possible due to multiple duplication of existing values. Suppose that $\q
  = s_1 \pipeline \cdots \pipeline s_m$. Then we compute $F_1, \dots, F_m$,
  where each $F_i$ is a set of DAGs, and we can obtain from $F_i$ the forest
  $F_0 \pipeline s_1 \pipeline \cdots \pipeline s_i$ by ``unravelling'' each
  DAG into a proper tree.

  We are going to consider connected DAGs with labeled nodes and edges and that
  have only one source node, that is, one node that has no incoming
  edges. Similarly to trees, a DAG is a tuple $(N,E,\lnode,\ledge)$, where $N$ is a
  set of nodes, $E$ is a successor relation, $\lnode: N \to V \cup
  \{\objectlabel,\arraylabel\}$ is a node labeling function, and $\ledge: E \to K
  \cup I$ is an edge labeling function such that
  \begin{inparaenum}[\it (i)]
  \item $(N,E)$ forms a DAG with a single node that has no incoming edges,
  \item a node labeled by a literal must be a node without outgoing edges,
  \item all outgoing edges of a node labeled by \objectlabel must be labeled by
    keys, and
  \item all outgoing edges of a node labeled by \arraylabel must be labeled by
    distinct indexes.
  \end{inparaenum}
  Clearly, a tree is a connected DAG with a single source node. We denote the
  source node of a DAG $t$ by $\treeroot(t)$. For a DAG $t$, the path type
  $\type(p,t)$, the interpretation of path $\eval{p}$, satisfaction of criteria
  $t \models \varphi$ and value definitions $t \models d$ is defined in the
  same way as for trees.

  \smallskip%
  First, we show, given a set $F$ of DAGs and a stage $s$, how to compute the
  set $F'$ of DAGs resulting from evaluating $s$ over $F$.
  \begin{itemize}
  \item Suppose $s$ is a match stage $\match{\varphi}$. Then $F' = \{t \mid t
    \in F \text{ and }t \models \varphi\}$. Clearly, $F' \subseteq F$, hence is
    linear in $F$ and $s$.

  \item Suppose $s$ is a project stage
    $\project{p_1,\dots,p_m,q_1/d_1,\dots,q_n/d_n}$. Let $t \in F$ be a DAG. We
    show how to transform it into a DAG $t'$ according to $s$.  Initially, $t'$
    contains one fresh node $r$ with $\lnode(r)=\objectlabel$. Then, for each
    $i\in\{1,\dots,n\}$, we do the following changes to $t'$. Suppose $q_i =
    k_1\cdots k_l$, we first insert into $t'$ fresh nodes $x_1,\dots,x_{l-1}$
    and edges $(x_j,x_{j+1})$ with $\ledge(r,x_1) = k_1$,
    $\ledge(x_{j-1},x_{j}) = k_j$, and $\lnode(x) = \objectlabel$ for $x \in
    \{x_1,\dots,x_{l-1}\}$.  Note that here, if $l=1$, $x_{l-1}$ refers to $r$.
    Then, by induction on the structure of $d_i$ we proceed as follows:
    \begin{enumerate}[(a)]
    \item $d_i$ is a literal value $v$: we insert a fresh node $x_l$ and an
      edge $(x_{l-1},x_{l})$ with $\ledge(x_{l-1},x_l) = k_l$ and $\lnode(x_l) = v$.

    \item $d_i$ is a path reference $p$. If $\eval{p}=\emptyset$, we remove
      from $t'$ all nodes $x_1,\dots,x_{l-1}$ (and edges) inserted previously.
      If $|\eval{p}|=1$, let $x_p \in \eval{p}$: we add to $t'$ the node $x_p$
      and its label, an edge $(x_{l-1},x_p)$ with $\ledge(x_{l-1},x_p) = k_l$, and
      copy all other nodes (hence the edges and labels) reachable from $x_p$ in
      $t$. Otherwise let $\eval{p} = \{y_1,\dots,y_m\}$: we insert into $t'$ a
      fresh node $x_l$ with $\lnode(x_l)=\arraylabel$, edges $(x_{l-1},x_l)$,
      $(x_l,y_1)$, \dots, $(x_l,y_m)$, with $\ledge(x_{l-1},x_l) = k_l$ and
      $\ledge(x_l,y_j) = j-1$, and copy all other nodes (and edges and labels)
      reachable from $y_j$ in $t$.

    \item if $d_i$ is a Boolean value definition, let $v_b$ be the Boolean
      value of $t \models d_i$. We proceed as in the case $d_i$ is a literal
      value $v_b$.

    \item if $d_i$ is a conditional value definition $\cond{d}{e_1}{e_2}$, then
      whenever $t \models d$, we proceed as in the case $d_i$ is $e_1$,
      otherwise as in the case $d_i$ is $e_2$.

    \item if $d_i$ is an array definition $[e_1,\dots,e_m]$, then we insert
      into $t'$ a fresh node $x_l$ with $\lnode(x_l) = \arraylabel$ and an edge
      $(x_{l-1},x_{l})$ with $\ledge(x_{l-1},x_l) = k_l$. Further, for each $e_j$,
      let $y_j$ be the node defined according to the structure of $e_j$ and the
      cases above (e.g., if $e_j$ is a literal value, then $y_j$ is a fresh
      node, and of $e_j$ is a path reference, it is an already existing in $t$
      node). We add an edge $(x_l,y_j)$ with $\ledge(x_l,y_j) = j-1$. Note that if
      $e_j$ is a path reference and this path does not exist in $t$, then it is
      equivalent to $e_j$ being \nullvalue.
    \end{enumerate}
    Thus, we have constructed the DAG $t'$ with a single source node. The size
    of $t'$ has grown at most linearly in the size of $t$ and $s$. In the
    resulting set $F'$, each DAG is obtained from exactly one DAG in $F$.

  \item Suppose $s$ is a group stage
    $\group{g_1/y_1,\dots,g_n/y_n}{a_1/b_1,\dots,a_m/b_m}$. If $n > 1$, let
    $F_1$ be a subset of $F$ such that
    \begin{itemize}
    \item[($\star$)] there exist indexes $i_1,\dots,i_k$, $k\leq n$, and values
      $v_{i_1},\dots,v_{i_k}$, such that for each DAG $t \in F_1$ the following
      holds: $t \models (\exists y_i)$ and $t \models (y_i=v_i)$ for $i\in
      \{i_1,\dots,i_k\}$ and $t \models \neg(\exists y_i)$ for $i \in
      \{1,\dots,n\}\setminus \{i_1,\dots,i_k\}$.
    \end{itemize}
    We show how to transform $F_1$ into a DAG $t_{F_1}$.
    Initially, $t_{F_1}$ contains two fresh nodes $r$ and $x_0$ with
    $\lnode(r)=\lnode(x_0)=\objectlabel$, and an edge $(r,x_0)$ with $\ledge(r,x_0) =
    \id$.  We now show how $t_{F_1}$ is built.

    First, for each $i\in\{i_1,\dots,i_k\}$, we proceed as follows. Fix a tree
    $t \in F_1$. Suppose $g_i = k_1\cdots k_l$, we insert into $t_{F_1}$ fresh
    nodes $x_1,\dots,x_{l-1}$ with $\lnode(x) = \objectlabel$ for $x \in
    \{x_1,\dots,x_{l-1}\}$, and edges $(x_j,x_{j+1})$ with $\ledge(x_0,x_1) =
    k_1$, $\ledge(x_{j-1},x_{j}) = k_j$. Then if $|\eval{y_i}|=1$, let $x_y \in
    \eval{y_i}$: we add to $t_{F_1}$ an edge $(x_{l-1},x_y)$ with
    $\ledge(x_{l-1},x_y) = k_l$, and copy all other nodes (hence the edges and
    labels) reachable from $x_y$ in $t$. Otherwise let $\eval{y_i} =
    \{z_1,\dots,z_h\}$: we insert into $t_{F_1}$ a fresh node $x_l$ with
    $\lnode(x_l)=\arraylabel$ and edges $(x_{l-1},x_l)$, $(x_l,z_1)$, \dots,
    $(x_l,z_h)$, with $\ledge(x_{l-1},x_l) = k_l$ and $\ledge(x_l,z_j) = j-1$, and
    copy all other nodes (hence the edges and labels) reachable from $z_j$ in
    $t$.

    Second, for each $i=[1..m]$, we proceed as follows.  We insert into
    $t_{F_1}$ a fresh node $x$ with $\lnode(x) = \arraylabel$ and an edge $(r,x)$
    with $\ledge(r,x) = a_i$. Now, for each DAG $t \in F_1$, we insert an element
    to the array rooted at $x$ as follows: if $\eval{b_i}=\emptyset$, then we
    do not insert anything into $t_{F_1}$; if $|\eval{b_i}|=1$, let
    $z\in\eval{b_i}$, we add to $t_{F_1}$ an edge $(x,z)$ with $\ledge(x,z)$ being
    the index of the new element, and copy all other nodes (hence the edges and
    labels) reachable from $z$ in $t$; otherwise let
    $\eval{b_i}=\{z_1,\dots,z_l\}$, we insert into $t_{F_1}$ a fresh node $z$
    with $\lnode(z) = \arraylabel$ and edges $(x,z),(z,z_1),\dots,(z,z_l)$ with
    $\ledge(x,z)$ being the index of the new element, $\ledge(z,z_j) = j-1$, and copy
    all other nodes (hence the edges and labels) reachable from $z_j$ in $t$.

    The resulting DAG $t_{F_1}$ has the single source node $r$, and its size is
    linear in the size of $F_1$ and $s$.  Let $F_1,\dots,F_l$ be the partition
    of $F$ into subsets satisfying ($\star$). Such a partition can be
    computed in time polynomial in $F$ and $s$: for each $t \in F$, we can
    determine its ``partition'' and then group the DAGs accordingly.  Then $F'$
    is obtained as $\{t_{F_1},\dots,t_{F_l}\}$, and its size is linear in the
    size of $F$ and $s$.

  \item Suppose $s$ is a lookup stage $\lookup{p_1=C.p_2}{p}$, and $G$ is the
    forest for $C$. Let $t \in F$ with the source node $x_0$, we show how to
    transform it into a DAG $t'$ according to $s$.  Initially $t'$ coincides
    with $t$. Suppose $p=k_1\cdots k_l$, we insert into $t'$ fresh nodes
    $x_1,\dots,x_l$ with $\lnode(x) = \objectlabel$ for $x \in
    \{x_1,\dots,x_{l-1}\}$, $\lnode(x_l) = \arraylabel$, and edges
    $(x_j,x_{j+1})$ with $\ledge(x_{j-1},x_{j}) = k_j$. Let $v$ be the value of
    $p_1$ in $t$, that is $v=\mathsf{value}(\subtree(t,p_1))$, and let $G_t$ be
    the subset of $G$ such that $\mathsf{value}(\subtree(g,p_2))=v$ for each $g
    \in G_t$. Then, for each $g \in G_t$, let $x_g$ be the root of $g$: we add
    to $t'$ an edge $(x_l,x_g)$ with $\ledge(x_l,x_g)$ being the consecutive
    index, and copy the whole tree $g$ to $t'$.

    The resulting DAG $t'$ is linear in the size of $F$, $G$ and $s$.
 \end{itemize}

 Now, we obtain that for a query $\q = s_1 \pipeline \cdots \pipeline s_m$ and
 an input forest $F_0$, each set of DAGs $F_i$, $i=[1..m]$ computed from
 $F_{i-1}$ and $s_i$ is linear in the size of $F_{i-1}$ and $s_i$, therefore
 $F_m$ is polynomial in the size of $F_0$ and $\q$. It should be clear that
 $F_0 \pipeline \q$ is non-empty iff $F_m$ is non-empty.
\end{proof}

Next, we show that adding unwind causes the loss of tractability, while project
and lookup do not add complexity.

\begin{lemmanum}{\ref{lem:mupl-np-complete}}
  Boolean query evaluation for \muq and \mupl queries is \NP-complete in
  combined complexity.
\end{lemmanum}
\begin{lemma}\label{lem:mu-np-hard}
  Boolean query evaluation for \muq is \NP-hard in combined complexity.
\end{lemma}
\begin{proof}
  We prove the lower bound by reduction from the Boolean satisfiability
  problem. Let $\varphi$ be a Boolean formula over $n$ variables x1, \dots,
  xn. We fix a collection name $C$, and construct a collection $F$ for $C$ and
  an \muq query $\q$ such that $\ansmongo(\q, F)$ is non-empty iff $\varphi$ is
  satisfiable.

  $F$ contains a single document $d$ of the form
  \valuefont{\object{''x1'':\,[true,false],\,\dots,\,``xn'':\,[true,false]}},
  and $\q$ is the query: $C \pipeline \unwind{\text{x1}} \pipeline \dots
  \pipeline \unwind{\text{xn}} \pipeline \match{\varphi}$, denoted $\q_{\NP}$,
  where $\varphi$ can be viewed as a criterion.
\end{proof}
\begin{corollary}
  The query emptiness problem for \mup queries is \NP-hard in query complexity.
\end{corollary}
\begin{proof}
  Since it is possible to use project to create copies of arrays, we can modify
  the above reduction so that $F$ contains a single document of the form
  \valuefont{\object{"values":\,[true,false]}}, and $\q = C \pipeline
  \project{\text{x}1/\text{values},~\dots,~\text{x}n/\text{values}} \pipeline
  \q_{\NP}$.
\end{proof}
\begin{corollary}
  The query emptiness problem for \mul queries is \NP-hard in query complexity.
\end{corollary}
\begin{proof}
  Now, we can use lookup to create copies of arrays. In this case again, $F$
  contains two documents of the form \valuefont{\object{"values":\,true}} and
  \valuefont{\object{"values":\,false}}. The query is as follows: $\q = C
  \pipeline \lookup{\text{dummy}=C.\text{dummy}}{\text{x}1} \pipeline \cdots
  \pipeline \lookup{\text{dummy}=C.\text{dummy}}{\text{x}n} \pipeline
  \unwind{\text{x}n} \pipeline \cdots \pipeline \unwind{\text{x}n}\pipeline
  \match{\varphi'}$, where $\varphi'$ is the variant of $\varphi$ where each
  variable $x$ is replaced by $x.\text{values}$.
\end{proof}

\begin{lemma}\label{lem:mupl-in-np}
  Boolean query evaluation for \mupl is in \NP in combined complexity.
\end{lemma}
\begin{proof}
  We modify the \PTIME algorithm for \mpgl as follows. Given an \mupl $\q$
  (over collection $C$), a forest $F_0$ for $C$, and forests $G_{C'}$ for each
  external collection $C'$ used by $\q$, we compute in non-deterministic
  polynomial time $F_0\pipeline \q$ and check whether the result is empty or
  not.

  We only show how to compute the set $F'$ of DAGs resulting from evaluating an
  unwind stage $s$ over a set $F$ of DAGs.
  \begin{itemize}
  \item Suppose $s=\unwind{p}$. Let $t \in F$, we show how to transform it into
    $F_t$, which is either the empty set or a singleton set $\{t'\}$, for a DAG
    $t'$. If $p$ is first level array in $t$, let $\{x_a\} = \eval{p}$, and
    $\{x_1,\dots,x_n\}$ all nodes such that $(x_a,x_i)$ are edges in $t$. If
    $n=0$, then $F_t = \emptyset$. Otherwise, we guess $k \in \{1,\dots,n\}$
    and $F_t = \{t'\}$.  Initially, the new DAG $t'$ coincides with $t$ but on
    the nodes reachable from $x_a$. If $\lnode(x_k)=\ell$ in $t$, then
    $\lnode(x_a) = \ell$ in $t'$ and we add to $t'$ the edges $(x_a,y)$ such
    that $(x_k,y)$ is in $t$ and copy to $t'$ all other nodes (hence the edges
    and labels) reachable from $y$ in $t$.

  \item Suppose $s=\unwind[+]{p}$. The difference with the previous stage is
    that if $n=0$, then $F_t =\{t\}$.
  \end{itemize}
  $F'$ is obtained as $\bigcup_{t \in F}F_t$. Clearly, $F'$ is linear in the
  size of $F$ and $s$.

  As a query contains a linear number of unwind stages, our algorithm requires
  to do a linear number of guesses (of polynomial size), and the whole
  computation runs in polynomial time.
\end{proof}

To conclude, we also show that evaluation of \mp queries with additional array
operators \emph{filter}, \emph{map} and \emph{setUnion} is \NP-hard in query
complexity.  The map operator $m_d(p)$ allows to transform each element inside
an array $p$ according to the new definition $d$, and the filter operator
$f_d(p)$ filters the elements of an array $p$ that satisfy $d$:
\begin{center}
  \begin{tabular}{l}
    \{\$filter: \{ input: \meta{PathRef}, as: \meta{Path}, cond: \meta{ValueDef} \}\}\\
    \{\$map: \{ input: \meta{PathRef}, as: \meta{Path}, in: \meta{ValueDef} \}\}\\
    \{\$setUnion: [\meta{List}<\meta{ValueDef}>\}\\
  \end{tabular}\qquad
  $\begin{array}{rl}
     d ~\DEF& f_d(p)\\
     \mid~& m_d(p)\\
     \mid~& d_1 \cup d_2\\
   \end{array}$
 \end{center}

% \begin{center}
%   \begin{tabular}{l}
%     \{\$map: \{ input: \meta{PathRef}, as: \meta{Path}, in: \meta{ValueDef} \}\}\\
%   \end{tabular}
%   %
%   $\begin{array}{rl}
%     d~\DEF& m_d(p)\\
%   \end{array}$
% \end{center}

\begin{lemma}{\label{lem:mp-np-hard-with-filter-and-map}}
  The query emptiness problem for \mp queries with filter, map and set union
  operators is \NP-hard in query complexity.
\end{lemma}
\begin{proof}
  Proof by reduction from the Boolean satisfiability problem. Let $\varphi$ be
  a Boolean formula over $n$ variables x1, \dots, x$n$. We construct a query
  $\q$ such that for each non-empty forest $F$, $F \pipeline \q$ is non-empty
  iff $\varphi$ is satisfiable.
  \begin{align}
    \q =\,
    &\project{\text{a}0/\{\text{x}1=0\},~\text{a}1/\{\text{x}1=1\}}\pipeline
    \project{\text{a}/[\text{a}0, \text{a}1]}\pipeline{}
    \tag{a$_1$}\\
    &\project{\text{a}0/m_{\{\text{x}1/\text{a.x}1,~ \text{x}2/0\}}(a),~
      \text{a}1/m_{\{\text{x}1/\text{a.x}1,~ \text{x}2/1\}}(a)}\pipeline
    \project{\text{a}/(\text{a}0 \cup \text{a}1)}\pipeline{}
    \tag{a$_2$}\\
    & \dots\notag\\
    &\project{
      \text{a}0/m_{\{\text{x}1/\text{a.x}1,\dots,\text{x}(n-1)/\text{a.x}(n-1),~\text{x}n/0\}}(\text{a}),~
      \text{a}1/m_{\{\text{x}1/\text{a.x}1,\dots,\text{x}(n-1)/\text{a.x}(n-1),~\text{x}n/1\}}(\text{a})}
    \pipeline \project{\text{a}/(\text{a}0 \cup \text{a}1)}\pipeline{}
    \tag{a$_n$}\\
    &\project{\text{assignments}/f_{\varphi}(\text{a})}\pipeline{}
    \tag{filter}\\
    &\match{\text{assignments}\neq[]} \notag
  \end{align}
  The stages (a$_1$) to (a$_n$) construct an array \valuefont{a} of $2^n$
  elements, where each element is an object encoding an assignment to the
  variables x1, \dots, x$n$. In the stage (a$_i$), the map operator is used to
  extend each current element with the an assignment to the variable x$i$. The
  (filter) stage then uses the filter operator to check for each element of the
  big array, whether it is a satisfying assignment, and if not, it is removed
  from the array. Finally, match will check that the resulting array is
  non-empty. If it is the case, then we have a satisfying assignment. All
  satisfying assignments will be stored in \valuefont{a}.
  % An actual query encoding the translation can be found in
  % Section~\ref{sec:query-mp-with-map-and-filter}.
\end{proof}

\begin{definition}
  Given a set $\S$ of type constraints, \mquery $\q$ is \emph{well-typed} for
  $\S$, if for each \mongodb instance $D$ satisfying $\S$, $\ansmongo(\q,D)$ is
  a well-typed forest.
\end{definition}

\begin{theorem}
  The problem of checking whether a query is well-typed is \TAexppoly-hard.
\end{theorem}
\begin{proof}
  Let $\q = C_0\mathsf{in}A_{p_2(n)}$ be the pipeline from the proof of
  Lemma~\ref{lem:TAexppoly-lower-bound}. Further, let
  $s_{\mathit{nwt}} = \project{\text{nonWellTypedPath}/[0, [1,2]]}$ and
  $\S = \{(C, \tree(\object{\id: \tliteral}))\}$. Then we have that the query
  $\q_{\mathit{nwt}} = C \pipeline \q \pipeline s_{\mathit{nwt}}$ is not
  well-typed for $\S$ iff the Turing machine $M$ accepts $w$ (see
  Lemma~\ref{lem:TAexppoly-lower-bound}).  When $M$ accepts $w$, then
  $\q_{\mathit{nwt}}$ is not well-typed for $\S$, and the witness input forest
  for it is $\{\tree(\lobject \id:1\robject)\}$. When $\M$ does not accept $w$,
  then $\ansmongo(\q_{\mathit{nwt}}, D)$ is empty (hence, well-typed) for each
  instance $D$ satisfying $\S$, as the value of \id is never used by~$\q$.
\end{proof}

\begin{theorem}
  Given a set of constraints $\S$, the problem of checking whether each stage
  in a query is well-typed for its input type is DP-complete.
\end{theorem}
\begin{proof}
  The upper bound follows from the algorithm reported in
  Section~\ref{sec:mquery2nra}. The lower bound is a straightforward reduction
  from the satisfiability and validity problems for Boolean formulas.
\end{proof}

\end{document}